\documentclass[12pt]{amsart}

\usepackage{amsmath,thmtools,amsfonts,amsbsy}
\usepackage{eucal}
\usepackage{dsfont}

\usepackage{xy}
\xyoption{matrix} \xyoption{arrow} \xyoption{arc} \xyoption{color}

\usepackage{subfigure}
\usepackage{tikz}
\usetikzlibrary{arrows}

\usepackage{yhmath}

\usepackage{enumerate}
\usepackage{enumitem}
\makeatletter
\def\namedlabel#1#2{\begingroup
    #2%
    \def\@currentlabel{#2}%
    \phantomsection\label{#1}\endgroup
}
\makeatother

\usepackage[latin1]{inputenc}  %%%%%% French and italian accented characters

\numberwithin{equation}{section}

\makeatletter
\newtheoremstyle{corsivo}
   {\medskipamount}{\medskipamount}%
   {\itshape}{}%
   {\bfseries}{}%
   { }
   {\thmname{#1}\thmnumber{\@ifnotempty{#1}{ }\@upn{#2}}%
    \thmnote{ {\bfseries(#3)}}.}%
\makeatother

\theoremstyle{corsivo}
\newtheorem{thm}{Theorem}%[section]
\newtheorem{lemma}{Lemma}

\newtheorem{prop}{Proposition}

\makeatletter
\newtheoremstyle{dritto}
   {\medskipamount}{\medskipamount}%
   {\rmfamily}{}%
   {\bfseries}{}%
   { }
   {\thmname{#1}\thmnumber{\@ifnotempty{#1}{ }\@upn{#2}}%
    \thmnote{ {\bfseries(#3)}}.}%
\makeatother

\theoremstyle{dritto}
\declaretheorem[style=dritto,name=Definition,qed=$\Diamond$]{dfn}
\declaretheorem[style=dritto,name=Remark,qed=$\Diamond$]{rmk}

\declaretheorem[style=dritto,qed=$\Diamond$]{assumption}

\newcommand{\eps}{\varepsilon}

\newcommand{\Id}{\mathds{1}}
\newcommand{\B}{\mathbb{B}}
\newcommand{\Bred}{\mathbb{B}\sub{eff}}
\newcommand{\C}{\mathbb{C}}
\newcommand{\R}{\mathbb{R}}
\newcommand{\Z}{\mathbb{Z}}

\newcommand{\T}{\mathbb{T}}
\newcommand{\PB}{\mathcal{P}}
\newcommand{\Hi}{\mathcal{H}}

\newcommand{\U}{\mathcal{U}}

\newcommand{\BH}{\mathcal{B}(\mathcal{H})}
\newcommand{\scal}[2]{\left\langle #1, #2 \right\rangle}
\newcommand{\norm}[1]{\left\| #1 \right\|}

\newcommand{\bra}[1]{\left\langle #1 \right|}
\newcommand{\ket}[1]{\left| #1 \right\rangle}
\newcommand{\eu}{\mathrm{e}}
\newcommand{\iu}{\mathrm{i}}
\newcommand{\di}{\mathrm{d}}
\newcommand{\act}{\triangleleft}
\newcommand{\sub}[1]{_{\mathrm{#1}}}
\newcommand{\Tr}{^{\mathsf{T}}}

\DeclareMathOperator{\tr}{tr}
\DeclareMathOperator{\Ran}{Ran}
\DeclareMathOperator{\Ker}{Ker}
\DeclareMathOperator{\Fr}{Fr}
\DeclareMathOperator{\diag}{diag}
\DeclareMathOperator{\Pf}{Pf}
\DeclareMathOperator{\Span}{Span}

\newcommand{\ie}{{\sl i.\,e.\ }}
\newcommand{\eg}{{\sl e.\,g.\ }}
\newcommand{\set}[1]{ \left\{  #1 \right\}}

\newcommand{\effective}{effective\ }
\newcommand{\primo}[1]{#1\sub{new}}
\newcommand{\eff}[1]{#1\sub{eff}}
\newcommand{\uc}[1]{#1\sub{uc}}

\newcommand{\soloarXiv}[1]{#1}
\let\oldfootnote\footnote
\renewcommand{\footnote}[1]{\oldfootnote{\  #1}}

\title[$\Z_2$ invariants of topological insulators]{$\Z_2$ invariants of topological insulators \\[1mm] as geometric obstructions}
\author{Domenico Fiorenza \and Domenico Monaco \and Gianluca Panati}
\usepackage%[pdftex,%
{hyperref}

\begin{document}

\begin{abstract}
We consider a gapped periodic quantum system with time-reversal symmetry of fermionic (or odd) type, \ie the time-reversal operator squares to $-\Id$. We investigate the existence of periodic and time-reversal invariant Bloch frames in dimensions $2$ and $3$. In $2d$, the obstruction to the existence of such a frame is shown to be encoded in a $\Z_2$-valued topological invariant, which can be computed by a simple algorithm. We prove that the latter agrees with the Fu-Kane index. In $3d$, instead, four $\Z_2$ invariants emerge from the construction, again related to the Fu-Kane-Mele indices. When no topological obstruction is present, we provide a constructive algorithm yielding explicitly a periodic and time-reversal invariant Bloch frame. The result is formulated in an abstract setting, so that it applies both to discrete models and to continuous ones.

\bigskip

\end{abstract}

\maketitle

\tableofcontents

%%%%% SECTION 1

\section{Introduction}

In the recent past, the solid state physics community has developed an increasing interest in phenomena having topological and geometric origin. The first occurrence of systems displaying different quantum phases which can be labelled by topological indices can be traced back at least to the seminal paper by Thouless, Kohmoto, Nightingale and den Nijs \cite{TKNN}, in the context of the Integer Quantum Hall Effect. The first topological invariants to make their appearance in the condensed matter literature were thus \emph{Chern numbers}: two distinct insulating quantum phases, which cannot be deformed one into the other by means of continuous (adiabatic) transformations without closing the gap between energy bands, are indexed by different \emph{integers} (see \cite{Graf review} and references therein). These topological invariants are related to an observable quantity, namely to the transverse (Hall) conductivity of the system under consideration \cite{TKNN,Graf review}; the fact that the topological invariant is an integer explains why the observable is quantized. Beyond the realm of Quantum Hall systems, similar non-trivial topological phases appear whenever time-reversal symmetry is broken, as early foreseen by Haldane \cite{Haldane88}. Since this pioneering observation, the field of \emph{Chern insulators} flourished \cite{Fuchs,Experiment,FruchartCarpentier2013}.

More recently, a new class of materials has been first theorized and then experimentally realized, where instead interesting topological quantum phases arise while preserving time-reversal symmetry: these materials are the so-called \emph{time-reversal symmetric} (TRS) \emph{topological insulators} (see \cite{Ando, HasanKane} for recent reviews). The peculiarity of these materials is that different quantum phases are labelled by \emph{integers modulo 2}; from a phenomenological point of view, these indices are connected to the presence of spin edge currents responsible for the Quantum Spin Hall Effect \cite{KaneMele2005, KaneMele_graphene}. It is crucial for the display of these currents that time-reversal symmetry is of \emph{fermionic} (or \emph{odd}) type, that is, the time-reversal operator $\Theta$ is such that $\Theta^2 = - \Id$. 

%\medskip

In a milestone paper \cite{KaneMele2005}, Kane and Mele consider a tight-binding model governing the dynamics of an electron in a $2$-dimensional honeycomb lattice subject to nearest- and next-to-nearest-neighbour hoppings, similarly to what happens in the Haldane model \cite{Haldane88}, with the addition of further terms, including time-reversal invariant spin-orbit interaction. This prototype model is used to propose a $\Z_2$ index to label the topological phases of $2d$ TRS topological insulators, and to predict the presence of observable currents in Quantum Spin Hall systems. An alternative formulation for this $\Z_2$ index is then provided by Fu and Kane in \cite{FuKa}, where the authors also argue that this index measures the obstruction to the existence of a continuous periodic Bloch frame, which is moreover compatible with time-reversal symmetry. Similar indices appear also in $3$-dimensional systems \cite{FuKaneMele}.

Since the proposals by Fu, Kane and Mele, there has been an intense activity in the community aimed at the explicit construction of smooth symmetric Bloch frames, in order to connect the possible topological obstructions to the $\Z_2$ indices \cite{Vanderbilt3}, and to study the localization of Wannier functions in TRS topological insulators \cite{Vanderbilt1,Vanderbilt2}. However, while the geometric origin of the integer-valued topological invariants is well-established (as was mentioned above, they represent Chern numbers of the \emph{Bloch bundle}, in the terminology of \cite{Panati}), the situation is less  clear for the $\Z_2$-valued indices of TRS topological insulators. Many interpretations of the $\Z_2$ indices have been given, using homotopic or $K$-theoretic classifications \cite{Altland Zirnbauer, Moore Balents, Kitaev, RyuSchnyder2010, Kennedy, KennedyZirnbauer14, KennedyZirnbauer14b}, $C^*$-algebraic and functional-analytic approaches \cite{Prodan1, Prodan2, Schulz-BaldesCMP, Schulz-Baldes13}, the bulk-edge correspondence \cite{Schulz-Baldes12,GrafPorta}, monodromy arguments \cite{Prodan3}, or gauge-theoretic methods \cite{Frohlich, Lyon, Lyon15}. Recently, a more geometric approach, using techniques from equivariant cohomology and elaborating on \cite{FKMM}, has been proposed in \cite{deNittisGomi}. However, we believe that a clear and simple topological explanation of how they arise from the symmetries of the system is still in a initiatory stage.%missing in the literature.

In this paper, we provide a geometric characterization of these $\Z_2$ indices as topological obstructions to the existence of continuous periodic and time-reversal symmetric Bloch frames, thus substantiating the claim in \cite{FuKa} on mathematical grounds. We consider a gapped periodic quantum system in the presence of fermionic time-reversal symmetry (compare Assumption \ref{proj}), and we investigate whether there exists a global continuous Bloch frame that is both periodic and time-reversal symmetric. While in $1d$ this always exists, a topological obstruction may arise in $2d$. We show in Section \ref{sec:Algorithm} that such an obstruction is encoded in a $\Z_2$ index $\delta$, which is moreover a \emph{topological invariant} of the system, with respect to those continuous deformations which preserve the symmetries. We prove that $\delta \in \Z_2$ agrees with the Fu-Kane index $\Delta \in \Z_2$ \cite{FuKa}, thus providing a proof that the latter is a topological invariant (Sections \ref{sec:Fu-Kane} and \ref{sec:EvalVertices}). In Section \ref{sec:3d} we investigate the same problem in $3d$, yielding to the definition of four $\Z_2$-valued topological obstructions, which are compared with the indices proposed by Fu, Kane and Mele in \cite{FuKaneMele}. In all cases where there is no topological obstruction (\ie the $\Z_2$ topological invariants vanish), we also provide an \emph{explicit algorithm} to construct a global smooth Bloch frame which is periodic and time-reversal symmetric (see also Appendix \ref{app:Smoothing}). Lastly, in Section \ref{sec:completeness} we prove that these $\Z_2$ topological invariants, defined by obstruction theory, characterize completely the topological class of the quantum system at hand, in all dimensions $d \le 3$.

%\medskip

%The main advantage of our method is that, being geometric in nature, it is based only on the \emph{fundamental symmetries} of the Hamiltonian modeling the system, namely invariance by lattice translations (\ie periodicity) and fermionic time-reversal symmetry. No further assumptions on the Hamiltonian and its gaps are needed in our approach, thus making it \emph{model-independent}; in particular, it applies both to continuous and to tight-binding models, and both to the $2$-dimensional and $3$-dimensional setting. To the best of our knowledge, our method appears to be the first obstruction-theoretic characterization of the $\Z_2$ invariants in the pioneering field of \emph{$3$-dimensional} TRS topological insultators. The method proposed here encompasses all models studied by the community, in particular the Fu-Kane-Mele models in $2d$ and $3d$ \cite{KaneMele2005,FuKa,FuKaneMele}. More general tight-binding models in $2$-dimensions were considered \eg in \cite{Schulz-Baldes12} and \cite{GrafPorta}, where an equivalence between the edge and the bulk index for TRS topological insulators is proved: our general result applies also to the \emph{time-reversal invariant bundles} (in the terminology of \cite{GrafPorta}) considered there, up to an identification of the coordinates on the basis torus.

A similar obstruction-theoretic approach to the invariants of $2$-dimensional topological insulators was adopted in \cite{GrafPorta}. In particular, aiming at a proof of the bulk-edge correspondence, there the authors associate a $\Z_2$ index to the \emph{time-reversal invariant bundle} associated to the bulk Hamiltonian for a semi-infinite crystal, enjoying a time-reversal symmetry of fermionic type.

Even though our starting assumptions on the family of projectors, to which we associate $\Z_2$-valued topological invariants, are modeled on the properties of the spectral projectors of a time-reversal symmetric Hamiltonian retaining \emph{full} periodicity in dimension $d \le 3$ (compare Assumption \ref{proj}), the setting of \cite{GrafPorta} is also covered by our method: our results can be applied to the family of projectors associated to the bulk time-reversal invariant bundle, up to an identification of the coordinates on the basis torus.

Indeed, the main advantage of our method is that, being geometric in nature, it is based only on the \emph{fundamental symmetries} of the system, namely invariance by (lattice) translations (\ie periodicity) and fermionic time-reversal symmetry. This makes our approach \emph{model-independent}; in particular, it applies both to continuous and to tight-binding models, and both to the $2$-dimensional and $3$-dimensional setting. To the best of our knowledge, our method appears to be the first obstruction-theoretic characterization of the $\Z_2$ invariants in the pioneering field of \emph{$3$-dimensional} TRS topological insultators. The method proposed here encompasses all models studied by the community, in particular the Fu-Kane-Mele models in $2d$ and $3d$ \cite{KaneMele2005,FuKa,FuKaneMele}, as well as more general tight-binding models in $2$ dimensions like the ones considered \eg in 
\cite{Schulz-Baldes12} and, as already mentioned, in \cite{GrafPorta}.

Another strong point in our approach is that the construction is algorithmic in nature, and gives also a way to \emph{compute} the $\Z_2$ invariants in a given system (see formulae \eqref{delta} and \eqref{easierdelta}). This makes our proposal well-suited for numerical implementation.

%\bigskip 

\noindent \textbf{Acknowledgments.} We thank H.\ Schulz-Baldes, G.M.\ Graf and M.\ Porta for interesting discussions, as well as the Reviewers for the careful reading of the manuscript and the stimulating remarks. We are grateful to the {\it Institut Henri Poincar\'e}  for the kind hospitality in the framework of the trimester ``Variational and Spectral Methods in Quantum Mechanics'', organized by M.\,J.\ Esteban and M.\ Lewin, and to the {\it Erwin Schr\"odinger Institute} for the kind hospitality in the framework of the thematic programme ``Topological Phases of Quantum Matter'', organized by N.\ Read, J.\ Yngvason, and M.\ Zirnbauer. 

\noindent The project was supported by INdAM-GNFM (Progetto Giovane Ricercatore) and by MIUR (Progetto PRIN 2012).

%%% SEND ME TO:
%%% a) S. Teufel, H. Schulz-Baldes, M. Porta, G.M. Graf
%%% b) D. Vanderbilt, R. Resta
%%% c) M.J. Esteban, M. Lewin (@IHP)
%%% d) [a mano a Vienna] Y. Avroni, J. Bellissard

\newpage

%%%%% SECTION 2

\section{Setting and main results} \label{sec:problem}

\subsection{Statement of the problem and main results} \label{sec:state}

We consider a gapped periodic quantum system with fermionic time-reversal symmetry, and we focus on the family of spectral eigenprojectors up to the gap, in Bloch-Floquet representation. In most of the applications, these projectors read
\begin{equation} \label{SpectralProj} 
P(k) = \sum_{n \in \mathcal{I}\sub{occ}} \ket{u_n(k)} \bra{u_n(k)}, \quad k \in \R^d,
\end{equation}
where $u_n(k)$ are the periodic parts of the Bloch functions, and the sum runs over all occupied bands.

Abstracting from specific models, we let $\Hi$ be a separable Hilbert space with scalar product $\scal{\cdot}{\cdot}$, $\BH$ denote the algebra of bounded linear operators on $\Hi$, and $\U(\Hi)$ the group of unitary operators on $\Hi$. We also consider a maximal lattice $\Lambda = \Span_\Z \set{e_1, \ldots, e_d} \simeq \Z^d \subset \R^d$: in applications, $\Lambda$ is the dual lattice to the periodicity Bravais lattice $\Gamma$ in position space.
The object of our study will be a family of orthogonal projectors $\set{P(k)}_{k \in \R^d} \subset \BH$, satisfying the following
\begin{assumption} \label{proj}
The family of orthogonal projectors $\set{P(k)}_{k \in \R^d}$ enjoys the following properties:
\begin{enumerate}
\item[\namedlabel{item:smooth}{(P$_1$)}] \emph{smoothness}: the map $\R^d \ni k \mapsto P(k) \in \BH$ is $C^\infty$-smooth;
\item[\namedlabel{item:tau}{(P$_2$)}] \emph{$\tau$-covariance}: the map $k \mapsto P(k)$ is covariant with respect to a unitary representation%
\footnote{This means that $\tau(0) = \Id_{\Hi}$ and $\tau(\lambda_1 + \lambda_2) = \tau(\lambda_1) \tau(\lambda_2)$ for all $\lambda_1, \lambda_2 \in \Lambda$. It follows in particular that $\tau(\lambda)^{-1} = \tau(\lambda)^* = \tau(-\lambda)$ for all $\lambda \in \Lambda$.} %
$\tau \colon \Lambda \to \U(\Hi)$ of the lattice $\Lambda$ on the Hilbert space $\Hi$, \ie
\[ P(k+\lambda) = \tau(\lambda) P(k) \tau(\lambda)^{-1}, \quad \text{for all } k \in \R^d, \text{ for all } \lambda \in \Lambda; \]
\item[\namedlabel{item:TRS}{(P$_{3,-}$)}] \emph{time-reversal symmetry}: the map $k \mapsto P(k)$ is time-reversal symmetric, \ie there exists an antiunitary
operator%
\footnote{Recall that a surjective antilinear operator $\Theta \colon \Hi \to \Hi$ is called \emph{antiunitary} if $\scal{\Theta \psi_1}{\Theta \psi_2} = \scal{\psi_2}{\psi_1}$ for all $\psi_1, \psi_2 \in \Hi$.} %
$\Theta \colon \Hi \to \Hi$, called the \emph{time-reversal operator}, such that 
\[ \Theta^2 = - \Id_{\Hi} \quad \text{and} \quad P(-k) = \Theta P(k) \Theta^{-1}. \] 
\end{enumerate}

Moreover, the unitary representation $\tau \colon \Lambda \to \U(\Hi)$ and the time-reversal operator $\Theta \colon \Hi \to \Hi$ satisfy
\begin{equation} \label{item:TRtau} \tag{P$_4$} 
\Theta \tau(\lambda) = \tau(\lambda)^{-1} \Theta \quad \text{for all } \lambda \in \Lambda. \qedhere
\end{equation}
\end{assumption}

Assumption \ref{proj} is satisfied by the spectral eigenprojectors of most Hamiltonians modelling gapped periodic quantum systems, in presence of fermionic time-reversal symmetry. Provided the Fermi energy lies in a spectral gap, the map $k \mapsto P(k)$ defined in \eqref{SpectralProj} will be smooth (compare \ref{item:smooth}), while $\tau$-covariance and (fermionic) time-reversal symmetry (properties \ref{item:tau} and \ref{item:TRS}) are inherited from the corresponding symmetries of the Hamiltonian. In particular, several well-established models satisfy the previous Assumption, including the eigenprojectors for the tight-binding Hamiltonians proposed by Fu, Kane and Mele in \cite{KaneMele2005,FuKa,FuKaneMele}, as well as in many \emph{continuous} models. Finally, Assumption \ref{proj} is satisfied also in the tight-binding models studied in \cite{GrafPorta}: the family of projectors is associated to the vector bundle used by Graf and Porta to define a bulk index, under a suitable identification of the variables%
\footnote{The coordinates $(k_1, \ldots, k_d)$ are expressed in terms of a basis $\set{e_1, \ldots, e_d} \subset \R^d$ generating the lattice $\Lambda$ as $\Lambda = \Span_\Z \set{e_1,\ldots,e_d}$.} %
$(k_1, k_2)$ with the variables $(k,z)$ appearing in \cite{GrafPorta}.

For a family of projectors satisfying Assumption \ref{proj}, it follows from \ref{item:smooth} that the rank $m$ of the projectors $P(k)$ is constant in $k$. We will assume that $m < + \infty$; property \ref{item:TRS} then gives that $m$ must be even. Indeed, the formula
\[ (\phi, \psi) := \scal{\Theta \phi}{\psi} \quad \text{for} \quad \phi, \psi \in \Hi \]
defines a bilinear, skew-symmetric, non-degenerate form on $\Hi$; its restriction to $\Ran P(0) \subset \Hi$ (which is an invariant subspace for the action of $\Theta$ in view of \ref{item:TRS}) is then a \emph{symplectic form}, and a symplectic vector space is necessarily even-dimensional.

The goal of our analysis will be to characterize the possible obstructions to the existence of a \emph{continuous symmetric Bloch frame} for the family $\set{P(k)}_{k \in \R^d}$, which we define now.

\begin{dfn}[(Symmetric) Bloch frame] \label{dfn:Bloch}
Let $\set{P(k)}_{k \in \R^d}$ be a family of projectors satisfying Assumption \ref{proj}, and let also $\Omega$ be a region in $\R^d$. A \textbf{Bloch frame} for $\set{P(k)}_{k \in \R^d}$ on $\Omega$ is a collection of maps $\Omega \ni k \mapsto \phi_{a}(k) \in \Hi$, $a \in \set{1, \ldots, m}$, such that for all $k \in \Omega$ the set $\Phi(k) := \set{\phi_1(k), \ldots, \phi_m(k)}$ is an orthonormal basis spanning $\Ran P(k)$. When $\Omega = \R^d$, the Bloch frame is said to be \emph{global}. A Bloch frame is called
\begin{enumerate}[label=$(\mathrm{F}_\arabic*)$,ref=$(\mathrm{F}_\arabic*)$]
\setcounter{enumi}{-1}
\item \label{item:F0} \emph{continuous} if all functions $\phi_a \colon \Omega \to \Hi$, $a \in \set{1, \ldots, m}$, are continous;
\item \label{item:F1} \emph{smooth} if all functions $\phi_a \colon \Omega \to \Hi$, $a \in \set{1, \ldots, m}$, are $C^\infty$-smooth.
\end{enumerate}

We also say that a global Bloch frame is 
\begin{enumerate}[resume*]
\item \label{tau-cov} \emph{$\tau$-equivariant} if 
\[ \phi_a(k + \lambda) = \tau(\lambda) \phi_a(k) \quad \text{for all } k \in \R^d, \: \lambda \in \Lambda, \: a \in \set{1, \ldots, m}; \]
\item \label{tr} \emph{time-reversal invariant} if
\[ \phi_b(-k) = \sum_{a = 1}^{m} \Theta \phi_a(k) \eps_{ab} \quad \text{for all } k \in \R^d, \: b \in \set{1, \ldots, m} \]
for some unitary and skew-symmetric matrix $\eps = (\eps_{ab})_{1 \le a,b \le m} \in \U(\C^m)$, $\eps_{ab} = - \eps_{ba}$.
\end{enumerate}

A global Bloch frame which is both $\tau$-equivariant and time-reversal invariant is called \textbf{symmetric}.
\end{dfn}

We are now in position to state our goal: we seek the answer to the following

\medskip

\noindent {\bf Question (Q$_d$).} {\it Let $d \le 3$. Given a family of projectors $\set{P(k)}_{k \in \R^d}$ satisfying Assumption \ref{proj} above, is it possible to find a global \emph{symmetric} Bloch frame for $\set{P(k)}_{k \in \R^d}$, which varies \emph{continuously} in $k$, \ie a global Bloch frame satisfying \ref{item:F0}, \ref{tau-cov} and \ref{tr}?}

\medskip

We will address this issue via an algorithmic approach. We will show that the existence of such a global continuous symmetric Bloch frame is in general \emph{topologically obstructed}. Explicitly, the main results of this paper are the following.

\begin{thm}[Answer to (Q$_1$)] \label{thm:Q1}
Let $d=1$, and let $\set{P(k)}_{k \in \R}$ be a family of projectors satisfying Assumption \ref{proj}. Then there exists a global continuous symmetric Bloch frame for $\set{P(k)}_{k \in \R}$. Moreover, such a Bloch frame can be explicitly constructed.
\end{thm}

The proof of Theorem \ref{thm:Q1} is contained in Section \ref{sec:EdgeExtension} (see Remark \ref{rmk:Q1}).

\begin{thm}[Answer to (Q$_2$)] \label{thm:Q2}
Let $d=2$, and let $\set{P(k)}_{k \in \R^2}$ be a family of projectors satisfying Assumption \ref{proj}. Then there exists a global continuous symmetric Bloch frame for $\set{P(k)}_{k \in \R^2}$ if and only if 
\begin{equation} \label{delta=0,2d} 
\delta(P) = 0 \in \Z_2, 
\end{equation}
where $\delta(P)$ is defined in \eqref{delta}. Moreover, if \eqref{delta=0,2d} holds, then such a Bloch frame can be explicitly constructed.
\end{thm}

The proof of Theorem \ref{thm:Q2}, leading to the definition of the $\Z_2$ index $\delta(P)$, is the object of Section \ref{sec:Algorithm}. Moreover, in Section \ref{sec:TopInv} we prove that $\delta(P)$ is actually a \emph{topological invariant} of the family of projectors (Proposition \ref{deltaTopInv}), which agrees with the Fu-Kane index (Theorem \ref{thm:delta=Delta}).

\begin{thm}[Answer to (Q$_3$)] \label{thm:Q3}
Let $d=3$, and let $\set{P(k)}_{k \in \R^3}$ be a family of projectors satisfying Assumption \ref{proj}. Then there exists a global continuous symmetric Bloch frame for $\set{P(k)}_{k \in \R^3}$ if and only if 
\begin{equation} \label{delta=0,3d} 
\delta_{1,0}(P) = \delta_{1,+}(P) = \delta_{2,+}(P) = \delta_{3,+}(P) = 0 \in \Z_2,
\end{equation}
where $\delta_{1,0}(P)$, $\delta_{1,+}(P)$, $\delta_{2,+}(P)$ and $\delta_{3,+}(P)$ are defined in \eqref{3Ddelta}. Moreover, if \eqref{delta=0,3d} holds, then such a Bloch frame can be explicitly constructed.
\end{thm}

The proof of Theorem \ref{thm:Q3}, leading to the definition of the four $\Z_2$ invariants $\delta_{1,0}(P), \delta_{1,+}(P), \delta_{2,+}(P)$ and $\delta_{3,+}(P)$, is the object of Section \ref{sec:3d}.

\begin{rmk}[Smooth Bloch frames]
Since the family of projectors $\set{P(k)}_{k \in \R^d}$ satisfies the smoothness assumption \ref{item:smooth}, one may ask whether global \emph{smooth} symmetric Bloch frames exist for $\set{P(k)}_{k \in \R^d}$, \ie global Bloch frames satisfying \ref{item:F1}, \ref{tau-cov} and \ref{tr}. We show in Appendix \ref{app:Smoothing} that, whenever a global \emph{continuous} symmetric Bloch frame exists, then one can also find an arbitrarily close symmetric Bloch frame which is also \emph{smooth}.
\end{rmk}

\subsection{\texorpdfstring{Properties of the reshuffling matrix $\eps$}{Properties of the reshuffling matrix epsilon}}

We introduce some further notation. Let $\Fr(m, \Hi)$ denote the set of \emph{$m$-frames}, namely $m$-tuples of orthonormal vectors in $\Hi$. If $\Phi = \set{\phi_1, \ldots, \phi_m}$ is an $m$-frame, then we can obtain a new frame in $\Fr(m,\Hi)$ by means of a unitary matrix $M \in \U(\C^m)$, setting
\[ (\Phi \act M)_b := \sum_{a = 1}^{m} \phi_a M_{ab}. \]
This defines a free right action of $\U(\C^m)$ on $\Fr(m,\Hi)$.

Moreover, we can extend the action of the unitary $\tau(\lambda) \in \U(\Hi)$, $\lambda \in \Lambda$, and of the time-reversal operator $\Theta \colon \Hi \to \Hi$ to $m$-frames, by setting
\[ \left( \tau_\lambda \Phi \right)_a := \tau(\lambda) \phi_a \quad \text{and} \quad (\Theta \Phi)_a := \Theta \phi_a \qquad \text{for } \Phi = \set{\phi_1, \ldots, \phi_m} \in \Fr(m,\Hi). \]
The unitary $\tau_\lambda$ commutes with the $\U(\C^m)$-action, \ie
\[ \tau_\lambda \left( \Phi \act M \right) = \left(\tau_\lambda \Phi \right) \act M, \quad \text{for all } \Phi \in \Fr(m,\Hi), \: M \in \U(\C^m), \]
because $\tau(\lambda)$ is a linear operator on $\Hi$. Notice instead that, by the antilinearity of $\Theta$, one has
\[ \Theta (\Phi \act M) = (\Theta \Phi) \act \overline{M}, \quad \text{for all } \Phi \in \Fr(m,\Hi), \: M \in \U(\C^m). \]

We can recast properties \ref{tau-cov} and \ref{tr} for a global Bloch frame in this notation as
\begin{equation} \label{Tau-Cov} \tag{$\mathrm{F}_3'$}
\Phi(k+\lambda) = \tau_\lambda \Phi(k), \quad \text{for all } k \in \R^d
\end{equation}
and
\begin{equation} \label{TR} \tag{$\mathrm{F}_4'$}
\Phi(-k) = \Theta \Phi(k) \act \eps, \quad \text{for all } k \in \R^d.
\end{equation}

\begin{rmk}[Compatibility conditions on $\eps$] \label{rmk:eps}
Observe that, by antiunitarity of $\Theta$, we have that for all $\phi \in \Hi$
\begin{equation} \label{phiThetaphi}
\scal{\Theta \phi}{\phi} = \scal{\Theta \phi}{\Theta^2 \phi} = - \scal{\Theta \phi}{\phi}
\end{equation}
and hence $\scal{\Theta \phi}{\phi} = 0$; the vectors $\phi$ and $\Theta \phi$ are always orthogonal. This motivates the presence of the ``reshuffling'' unitary matrix $\eps$ in \ref{tr}: the na\"{i}ve definition of time-reversal symmetric Bloch frame, namely $\phi_a(-k) = \Theta \phi_a(k)$, would be incompatible with the fact that the vectors $\set{\phi_a(k)}_{a=1, \ldots, m}$ form a basis for $\Ran P(k)$ for example at $k=0$. Notice, however, that if property \ref{item:TRS} is replaced by
\begin{equation} \tag{$\mathrm{P}_{3,+}$}
\Theta^2 = \Id_{\Hi} \quad \text{and} \quad P(-k) = \Theta P(k) \Theta^{-1},
\end{equation}
then \eqref{phiThetaphi} does not hold anymore, and one can indeed impose the compatibility of a Bloch frame $\Phi$ with the time-reversal operator by requiring that $\Phi(-k) = \Theta \Phi(k)$. Indeed, one can show \cite{FiPaPi} that under this modified assumption there is \emph{no topological obstruction} to the existence of a global smooth symmetric Bloch frame for all $d \le 3$.

We have thus argued why the presence of the reshuffling matrix $\eps$ in condition \ref{tr} is necessary. The further assumption of skew-symmetry on $\eps$ is motivated as follows. Assume that $\Phi = \set{\Phi(k)}_{k \in \R^d}$ is a time-reversal invariant Bloch frame. Consider Equation \eqref{TR} with $k$ and $-k$ exchanged, and act on the right with $\eps^{-1}$ to both sides, to obtain
\[ \Phi(k) \act \eps^{-1} = \Theta \Phi(-k). \]
Substituting again the expression in \eqref{TR} for $\Phi(-k)$ on the right-hand side, one gets
\[ \Phi(k) \act \eps^{-1} = \left( \Theta^2 \Phi(k) \right) \act \overline{\eps} = \Phi(k) \act (- \overline{\eps}) \]
and hence we deduce that $\eps^{-1} = - \overline{\eps}$. On the other hand, by unitarity $\eps^{-1} = \overline{\eps}\Tr$ and hence $\eps\Tr = - \eps$. So $\eps$ must be not only unitary, but also skew-symmetric. In particular $\eps \overline{\eps} = - \Id$.

Notice that, according to \cite[Theorem 7]{Hua}, the matrix $\eps$, being unitary and skew-symmetric, can be put in the form
\begin{equation} \label{eps=J}
\begin{pmatrix} 0 & 1 \\ -1 & 0 \end{pmatrix} \oplus \cdots \oplus \begin{pmatrix} 0 & 1 \\ -1 & 0 \end{pmatrix}
\end{equation}
in a suitable orthonormal basis. Hence, up to a reordering of this basis, there is no loss of generality in assuming that $\eps$ is in the \emph{standard symplectic form}%
\footnote{The presence of a ``symplectic'' condition may seem unnatural in the context of complex Hilbert spaces. However, as was already remarked, the time-reversal symmetry operator induces naturally a symplectic structure on the invariant subspaces $\Ran P(k_\lambda)$.}%A more abstract viewpoint, based on quiver-theoretic techniques, can indeed motivate the appearence of the standard symplectic matrix \cite{Cerulli}.}%
\begin{equation} \label{symplectic}
\eps = \begin{pmatrix} 0 & \Id_n \\ - \Id_n & 0 \end{pmatrix}
\end{equation}
where $n = m/2$ (remember that $m$ is even). We will make use of this fact later on.
\end{rmk}

\soloarXiv{
\begin{rmk}[Geometric reinterpretation] \label{rmk:Geometry}
Let us recast the above definitions in a more geometric language. Given a smooth and $\tau$-covariant family of projectors $\set{P(k)}_{k \in \R^d}$ one can construct a vector bundle $\PB \to \T^d$, called the \emph{Bloch bundle}, having the (Brillouin) $d$-torus $\T^d := \R^d / \Lambda$ as base space, and whose fibre over the point $k \in \T^d$ is the vector space $\Ran P(k)$ (see \cite[Section 2.1]{Panati} for details). The main result in \cite{Panati} (see also \cite{PaMo}) is that, if $d \le 3$ and if $\set{P(k)}_{k \in \R^d}$ is also time-reversal symmetric, then the Bloch bundle $\PB \to \T^d$ is trivial, in the category of $C^\infty$-smooth vector bundles. This is equivalent to the existence of a global $\tau$-equivariant Bloch frame: this can be seen as a section of the \emph{frame bundle} $\Fr(\PB) \to \T^d$, which is the principal $\U(\C^m)$-bundle whose fibre over the point $k \in \T^d$ is the set of orthonormal frames in $\Ran P(k)$.

The time-reversal operator $\Theta$ induces by restriction a (non-vertical) automorphism of $\PB$, \ie a morphism
\[ \xymatrix{ \PB \ar[r]^{\widehat{\Theta}} \ar[d] & \PB \ar[d] \\ \T^d \ar[r]^{\theta} & \T^d} \]
where $\theta \colon \T^d \to \T^d$ denotes the involution $\theta(k) = -k$. This means that a vector in the fibre $\Ran P(k)$ is mapped via $\widehat{\Theta}$ into a vector in the fibre $\Ran P(-k)$. The morphism $\widehat{\Theta} \colon \PB \to \PB$ still satisfies $\widehat{\Theta}^2 = - \Id$, \ie it squares to the vertical automorphism of $\PB$ acting fibrewise by multiplication by $-1$.
\end{rmk}
}

\newpage

%%%%% SECTION 3

\section{\texorpdfstring{Construction of a symmetric Bloch frame in $2d$}{Construction of a symmetric Bloch frame in 2d}} \label{sec:Algorithm}

In this Section, we tackle Question (Q$_d$) stated in Section \ref{sec:state} for $d=2$.

\subsection{Effective unit cell, vertices and edges} \label{sec:Bred}

Consider the unit cell
\[ \B := \set{k = \sum_{j=1}^{2} k_j e_j \in \R^2: -\frac{1}{2} \le k_i \le \frac{1}{2}, \: i = 1, 2}.  \] 
Points in $\B$ give representatives for the quotient Brillouin torus $\T^2 = \R^2 / \Lambda$, \ie any point $k \in \R^2$ can be written (in an a.e.-unique way) as $k = k' + \lambda$, with $k' \in \B$ and $\lambda \in \Lambda$.

Properties \ref{item:tau} and \ref{item:TRS} for a family of projections reflect the relevant symmetries of $\R^2$: the already mentioned \emph{inversion symmetry} $\theta(k) = -k$ and the \emph{translation symmetries} $t_\lambda(k) = k+\lambda$, for $\lambda \in \Lambda$. These transformations satisfy the commutation relation $\theta t_\lambda = t_{-\lambda} \theta$. Consequently, they form a subgroup of the affine group $\mathrm{Aut}(\R^2)$, consisting of the set $\set{t_{\lambda}, \theta t_\lambda}_{\lambda \in \Lambda}$. Periodicity (or rather, $\tau$-covariance) for families of projectors and, correspondingly, Bloch frames allows one to focus one's attention to points $k \in \B$. Implementing also the inversion or time-reversal symmetry restricts further the set of points to be considered to the \emph{\effective unit cell}
\[ \Bred := \set{k = (k_1, k_2) \in \B: k_1 \ge 0}.  \]

A more precise statement is contained in Proposition \ref{global}. Let us first introduce some further terminology. We define the \emph{vertices} of the \effective unit cell to be the points $k_\lambda \in \Bred$ which are fixed by the transformation $t_\lambda \theta$. One immediately realizes that
\[ t_\lambda \theta(k_\lambda) = k_\lambda \quad \Longleftrightarrow \quad k_\lambda = \frac{1}{2} \lambda, \]
\ie vertices have half-integer components in the basis $\set{e_1, e_2}$. Thus, the \effective unit cell contains exactly six vertices, namely
\begin{equation} \label{vertices}
\begin{aligned}
v_1 = (0,0), \quad v_2 = \left( 0, - \frac{1}{2} \right), \quad v_3 = \left( \frac{1}{2}, - \frac{1}{2} \right), \\
v_4 = \left(\frac{1}{2} ,0 \right), \quad v_5 = \left( \frac{1}{2}, \frac{1}{2} \right), \quad v_6 = \left( 0, \frac{1}{2} \right).
\end{aligned}
\end{equation}
We also introduce the oriented \emph{edges} $E_i$, joining two consecutive vertices $v_i$ and $v_{i+1}$ (the index $i$ must be taken modulo $6$).

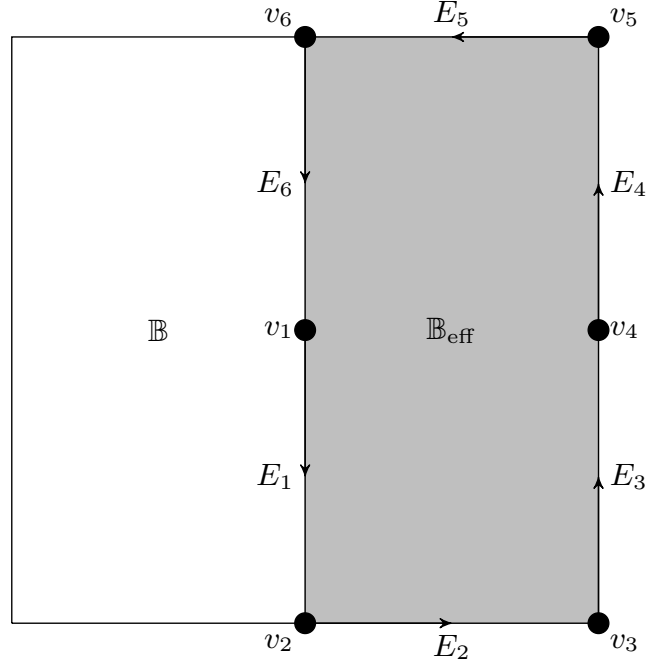
\begin{figure}[ht]
\begin{scriptsize}
\centerline{%
\scalebox{1.3}{%
\begin{tikzpicture}[>=stealth']
\filldraw [lightgray] (0,-3) -- (3,-3) -- (3,3) -- (0,3) -- (0,-3);
\draw (-3,-3) -- (-3,3) -- (3,3) -- (3,-3) -- (-3,-3);
\draw (0,-3) -- (0,3);
\draw [->] (0,0) -- (0,-1.5) node [anchor=east] {$E_1$};
\draw [->] (0,-3) -- (1.5,-3) node [anchor=north] {$E_2$};
\draw [->] (3,-3) -- (3,-1.5) node [anchor=west] {$E_3$};
\draw [->] (3,0) -- (3,1.5) node [anchor=west] {$E_4$};
\draw [->] (3,3) -- (1.5,3) node [anchor=south] {$E_5$};
\draw [->] (0,3) -- (0,1.5) node [anchor=east] {$E_6$};
\filldraw [black] (0,0) circle (3pt) node [anchor=east] {$v_1$}
					(0,-3) circle (3pt) node [anchor=north east] {$v_2$}
					(3,-3) circle (3pt) node [anchor=north west] {$v_3$}
					(3,0) circle (3pt) node [anchor=west] {$v_4$}
					(3,3) circle (3pt) node [anchor=south west] {$v_5$}
					(0,3) circle (3pt) node [anchor=south east] {$v_6$};
\draw (1.5,0) node {$\Bred$}
      (-1.5,0) node {$\B$};
\end{tikzpicture}
} %end scalebox
} %end centerline
\end{scriptsize}
\caption{The \effective unit cell (shaded area), its vertices and its edges. We use adapted coordinates $(k_1,k_2)$ such that $k = k_1 e_1 + k_2 e_2$.}
\label{fig:Bred}
\end{figure}

We start with an auxiliary \emph{extension} result, which reduces the problem of the existence of a \emph{global} continuous symmetric Bloch frame to that of a Bloch frame defined only on the \effective unit cell $\Bred$, satisfying further conditions on its boundary.

\begin{prop} \label{global}
Let $\set{P(k)}_{k \in \R^2}$ be a family of orthogonal projectors satisfying Assumption \ref{proj}. Assume that there exists a global continuous symmetric Bloch frame $\Phi = \set{\Phi(k)}_{k \in \R^2}$ for $\set{P(k)}_{k \in \R^2}$. Then $\Phi$ satisfies the \emph{vertex conditions}
\begin{equation} \label{VertexCondition} \tag{$\mathrm{V}$}
\Phi(k_\lambda) = \tau_\lambda \Theta \Phi(k_\lambda) \act \eps, \quad k_\lambda \in \set{ v_1, \ldots, v_6 }
\end{equation}
and the \emph{edge symmetries}
\begin{equation} \label{EdgeSymmetries} \tag{$\mathrm{E}$}
\begin{aligned}
&\Phi(\theta(k)) = \Theta \Phi(k) \act \eps && \text{for } k \in E_1 \cup E_6, \\
&\Phi(t_{e_2}(k)) = \tau_{e_2} \Phi(k) && \text{for } k \in E_2, \\
&\Phi(\theta t_{-e_1}(k)) = \tau_{e_1} \Theta \Phi(k) \act \eps \quad && \text{for } k \in E_3 \cup E_4, \\
&\Phi(t_{-e_2}(k)) = \tau_{-e_2} \Phi(k) && \text{for } k \in E_5.
\end{aligned}
\end{equation}

Conversely, let $\eff{\Phi} = \set{\eff{\Phi}(k)}_{k \in \Bred}$ be a continuous Bloch frame for $\set{P(k)}_{k \in \R^2}$, defined on the \effective unit cell $\Bred$ and satisfying the vertex conditions \eqref{VertexCondition} and the edge symmetries \eqref{EdgeSymmetries}. Then there exists a \emph{global continuous symmetric} Bloch frame $\Phi$ whose restriction to $\Bred$ coincides with $\eff{\Phi}$.
\end{prop}
\begin{proof}
Let $\Phi$ be a global Bloch frame as in the statement of the Proposition. Then
conditions \eqref{Tau-Cov} and \eqref{TR} imply that at the six vertices
\[ \Phi(k_\lambda) = \Phi(t_\lambda \theta(k_\lambda)) =  \tau_\lambda \Phi(\theta(k_\lambda)) = \tau_\lambda \Theta \Phi(k_\lambda) \act \eps, \]
that is, $\Phi$ satisfies the vertex conditions \eqref{VertexCondition}. The edge symmetries \eqref{EdgeSymmetries} can be checked similarly, again by making use of \eqref{Tau-Cov} and \eqref{TR}.

Conversely, assume that a continuous Bloch frame $\eff{\Phi}$ is given on $\Bred$, and satisfies \eqref{VertexCondition} and \eqref{EdgeSymmetries}. We extend the definition of $\eff{\Phi}$ to the unit cell $\B$ by setting
\[ \uc{\Phi}(k) := \begin{cases}
\eff{\Phi}(k) & \text{if } k \in \Bred, \\
\Theta \eff{\Phi}(\theta(k)) \act \eps & \text{if } k \in \B \setminus \Bred.
\end{cases} \]
The definition of $\uc{\Phi}$ can in turn be extended to $\R^2$ by setting
\[ \Phi(k) := \tau_\lambda \uc{\Phi}(k') \quad \text{if } k = k' + \lambda \text{ with } k' \in \B, \: \lambda \in \Lambda. \]
The vertex conditions and the edge symmetries ensure that the above defines a global \emph{continuous} Bloch frame; moreover, by construction $\Phi$ is also \emph{symmetric}, in the sense of Definition \ref{dfn:Bloch}.
\end{proof}

In view of Proposition \ref{global}, our strategy to examine Question (Q$_2$) will be to consider a continuous Bloch frame $\Psi$ defined over the \effective unit cell $\Bred$ (whose existence is guaranteed by the fact that $\Bred$ is contractible and no further symmetry is required), and try to modify it in order to obtain a new continuous Bloch frame $\Phi$, which is defined on the effective unit cell and satisfies also the vertex conditions and the edge symmetries; by the above extension procedure one obtains a global continuous \emph{symmetric} Bloch frame. Notice that, since both are orthonormal frames in $\Ran P(k)$, the given Bloch frame $\Psi(k)$ and the unknown symmetric Bloch frame $\Phi(k)$ differ by the action of a unitary matrix $U(k)$:
\begin{equation} \label{Healer}
\Phi(k) = \Psi(k) \act U(k), \quad U(k) \in \U(\C^m).
\end{equation}
Thus, we can equivalently treat the family $\Bred \ni k \mapsto U(k) \in \U(\C^m)$ as our unknown.

\subsection{Solving the vertex conditions} \label{sec:VertexConditions}

Let $k_\lambda$ be one of the six vertices in \eqref{vertices}. If $\Phi$ is a symmetric Bloch frame, then, by Proposition \ref{global}, $\Phi(k_\lambda)$ satisfies the vertex condition \eqref{VertexCondition}, stating the equality between the two frames $\Phi(k_\lambda)$ and $\tau_\lambda \Theta \Phi(k_\lambda) \act \eps$. For a general Bloch frame $\Psi$, instead, $\Psi(k_\lambda)$ and $\tau_\lambda \Theta \Psi(k_\lambda) \act \eps$ may very well be different. Nonetheless, they are both orthonormal frames in $\Ran P(k_\lambda)$, so there exists a unique unitary matrix $U\sub{obs}(k_\lambda) \in \U(\C^m)$ such that
\begin{equation} \label{UobsVertices}
\Psi(k_\lambda) \act U\sub{obs}(k_\lambda) = \tau_\lambda \Theta \Psi(k_\lambda) \act \eps.
\end{equation}
The \emph{obstruction unitary} $U\sub{obs}(k_\lambda)$ must satisfy a compatibility condition. In fact, by applying $\tau_\lambda \Theta$ to both sides of \eqref{UobsVertices} we obtain that
\begin{align*}
\tau_\lambda \Theta \left( \Psi(k_\lambda) \act U\sub{obs}(k_\lambda) \right) & = \tau_\lambda \Theta \left( \tau_\lambda \Theta \Psi(k_\lambda) \act \eps \right) = \\
& = \tau_\lambda \Theta \tau_\lambda \Theta \Psi(k_\lambda) \act \overline{\eps} = \\
& = \tau_\lambda \tau_{-\lambda} \Theta^2 \Psi(k_\lambda) \act \overline{\eps} = \\
& = \Psi(k_\lambda) \act (-\overline{\eps})
\end{align*}
where in the second-to-last equality we used the commutation relation \eqref{item:TRtau}. On the other hand, the left-hand side of this equality is also given by
\begin{align*}
\tau_\lambda \Theta \left( \Psi(k_\lambda) \act U\sub{obs}(k_\lambda) \right) & = \tau_\lambda \Theta \Psi(k_\lambda) \act \overline{U\sub{obs}}(k_\lambda) = \\
& = \left( \Psi(k_\lambda) \act \left(U\sub{obs}(k_\lambda) \eps^{-1} \right) \right) \act \overline{U\sub{obs}}(k_\lambda) = \\
& = \Psi(k_\lambda) \act \left( U\sub{obs}(k_\lambda) \eps^{-1} \overline{U\sub{obs}}(k_\lambda) \right).
\end{align*}
By the freeness of the action of $\U(\C^m)$ on frames and the fact that $\eps^{-1} = - \overline{\eps}$ by Remark \ref{rmk:eps}, we deduce that
\begin{equation} \label{Compatibility}
U\sub{obs}(k_\lambda) \, \overline{\eps} \, \overline{U\sub{obs}}(k_\lambda) = \overline{\eps}, \quad \text{\ie} \quad U\sub{obs}(k_\lambda)\Tr \eps = \eps \, U\sub{obs}(k_\lambda).
\end{equation}

Now, notice that the given Bloch frame $\Psi(k)$ and the unknown symmetric Bloch frame $\Phi(k)$, satisfying the vertex condition \eqref{VertexCondition}, differ by the action of a unitary matrix $U(k)$, as in \eqref{Healer}. We want to relate the obstruction unitary $U\sub{obs}(k_\lambda)$ to the unknown $U(k_\lambda)$. In order to do so, we rewrite \eqref{VertexCondition} as
\begin{align*}
\Psi(k_\lambda) \act U(k_\lambda) & = \Phi(k_\lambda) = \tau_\lambda \Theta \Phi(k_\lambda) \act \eps = \\
& = \tau_\lambda \Theta \left( \Psi(k_\lambda) \act U(k_\lambda) \right) \act \eps = \\
& = \tau_\lambda \Theta \Psi(k_\lambda) \act \left( \overline{U}(k_\lambda) \eps \right) = \\
& = \left( \Psi(k_\lambda) \act \left( U\sub{obs}(k_\lambda) \eps^{-1} \right) \right) \act \left( \overline{U}(k_\lambda) \eps \right) = \\
& = \Psi(k_\lambda) \act \left( U\sub{obs}(k_\lambda) \eps^{-1} \overline{U}(k_\lambda) \eps \right).
\end{align*}
Again by the freeness of the $\U(\C^m)$-action, we conclude that
\begin{equation} \label{Uobs-U}
U(k_\lambda) = U\sub{obs}(k_\lambda) \eps^{-1} \overline{U}(k_\lambda) \eps, \quad \text{\ie} \quad U\sub{obs}(k_\lambda) = U(k_\lambda) \eps^{-1} U(k_\lambda)\Tr \eps.
\end{equation}

The next Lemma establishes the equivalence between the two conditions \eqref{Compatibility} and \eqref{Uobs-U}.

\begin{lemma} \label{V->U}
Let $\eps \in \U(\C^m) \cap {\bigwedge}^2 \C^m$ be a unitary and skew-symmetric matrix. The following conditions on a unitary matrix $V \in \U(\C^m)$ are equivalent:
\begin{enumerate}[label=$(\alph*)$]
\item $V$ is such that $V\Tr \eps = \eps V$;
\item there exists a matrix $U \in \U(\C^m)$ such that $V = U \eps^{-1} U\Tr \eps$.
\end{enumerate}
\end{lemma}
\begin{proof}
The implication $(b) \Rightarrow (a)$ is obvious. We prove the implication $(a) \Rightarrow (b)$.

Every unitary matrix can be diagonalized by means of a unitary transformation. Hence there exist a unitary matrix $W \in \U(\C^m)$ and a diagonal matrix $\Lambda = \diag(\lambda_1, \ldots, \lambda_m)$ such that
\[ V = W \eu^{\iu \Lambda} W^* \]
where the collection $\set{\eu^{\iu \lambda_1}, \ldots, \eu^{\iu \lambda_m}}$ forms the spectrum of $V$. The condition $V\Tr \eps = \eps V$ is then equivalent to
\[ \overline{W} \eu^{\iu \Lambda} W\Tr \eps = \eps W \eu^{\iu \Lambda} W^* \quad \Longleftrightarrow \quad \eu^{\iu \Lambda} W\Tr \eps W = W\Tr \eps W \eu^{\iu \Lambda}, \]
\ie the matrix $A := W\Tr \eps W$ commutes with the diagonal matrix $\eu^{\iu \Lambda}$. 

According to \cite[Lemma in \S 6]{Hua}, for every unitary matrix $Z$ there exists a unitary matrix $Y$ such that $Y^2 = Z$ and that $Y B = B Y$ whenever $Z B = B Z$. We can apply this fact in the case where $Z = \eu^{\iu \Lambda}$ is diagonal, and give an explicit form of $Y$: Normalize the arguments $\lambda_i$ of the eigenvalues of $V$ so that $\lambda_i \in [0,2\pi)$, and define $Y := \eu^{\iu \Lambda/2} = \diag\left(\eu^{\iu \lambda_1/2}, \ldots,\eu^{\iu \lambda_m/2}\right)$.

We claim now that the matrix
\[ U := W \eu^{\iu \Lambda/2} W^*, \]
which is clearly unitary, satisfies condition $(b)$ in the statement. Indeed, upon multiplying by $W W^* = \Id$, we get that
\[ U \eps^{-1} U\Tr \eps = W \eu^{\iu \Lambda/2} W^* \eps^{-1} \overline{W} \eu^{\iu \Lambda/2} W\Tr \eps W W^* = W \eu^{\iu \Lambda/2} A^{-1} \eu^{\iu \Lambda/2} A W^* = W \eu^{\iu \Lambda} W^* = V. \]
This concludes the proof of the Lemma.
\end{proof}

The above result allows us to solve the vertex condition, or equivalently the equation \eqref{Uobs-U} for $U(k_\lambda)$, by applying Lemma \ref{V->U} to $V = U\sub{obs}(k_\lambda)$ and $U = U(k_\lambda)$.

\subsection{Extending to the edges} \label{sec:EdgeExtension}

To extend the definition of the symmetric Bloch frame $\Phi(k)$ (or equivalently of the matrix $U(k)$ appearing in \eqref{Healer}) also for $k$ on the edges $E_i$ which constitute the boundary $\partial \Bred$, we use the path-connectedness of the group $\U(\C^m)$. Indeed, we can choose a continuous path%
\soloarXiv{
\footnote{Explicitly, a continuous path of unitaries $W \colon [0,1/2] \to \U(\C^m)$ connecting two unitary matrices $U_1$ and $U_2$ can be constructed as follows. Diagonalize $U_1^{-1} U_2 = P D P^*$, where $D = \diag\left(\eu^{\iu \mu_1}, \ldots, \eu^{\iu \mu_m}\right)$ and $P \in \U(\C^m)$. For $t \in [0,1/2]$ set
\[ W(t) = U_1 P D_t P^*, \quad \text{where} \quad D_t := \diag\left(\eu^{\iu 2 t \mu_1}, \ldots, \eu^{\iu 2t \mu_m}\right). \]
One easily realizes that $W(0) = U_1$, $W(1/2) = U_2$ and $W(t)$ depends continuously on $t$, as required.}} %
$W_i \colon [0, 1/2] \to \U(\C^m)$ such that $W_i(0) = U(v_i)$ and $W_i(1/2) = U(v_{i+1})$, where $v_i$ and $v_{i+1}$ are the end-points of the edge $E_i$. Now set
\[ \widetilde{U}(k) := \begin{cases}
W_1(-k_2) & \text{if } k \in E_1, \\
W_2(k_1) & \text{if } k \in E_2, \\
W_3(k_2+1/2) & \text{if } k \in E_3.
\end{cases} \]
In this way we obtain a continuous map $\widetilde{U} \colon E_1 \cup E_2 \cup E_3 \to \U(\C^m)$. Let $\widetilde{\Phi}(k) := \Psi(k) \act \widetilde{U}(k)$ for $k \in E_1 \cup E_2 \cup E_3$; we extend this frame to a $\tau$-equivariant, time-reversal invariant frame on $\partial \Bred$ by setting
\begin{equation} \label{Psihat}
\widehat{\Phi}(k) := \begin{cases}
\widetilde{\Phi}(k) & \text{if } k \in E_1 \cup E_2 \cup E_3, \\
\tau_{e_1} \Theta \widetilde{\Phi}(\theta t_{-e_1}(k)) \act \eps & \text{if } k \in E_4, \\
\tau_{e_2} \widetilde{\Phi}(t_{-e_2}(k)) & \text{if } k \in E_5, \\
\Theta \widetilde{\Phi}(\theta(k)) \act \eps & \text{if } k \in E_6.
\end{cases}
\end{equation}

By construction, $\widehat{\Phi}(k)$ satisfies all the edge symmetries for a symmetric Bloch frame $\Phi$ listed in \eqref{EdgeSymmetries}, as one can immediately check.

\begin{rmk}[Proof of Theorem \ref{thm:Q1}] \label{rmk:Q1}
The above argument also shows that, when $d=1$, global continuous symmetric Bloch frames for a family of projectors $\set{P(k)}_{k \in \R}$ satisfying Assumption \ref{proj} can always be constructed. Indeed, the edge $E_1 \cup E_6$ can be regarded as a $1$-dimensional unit cell $\B^{(1)}$, and the edge symmetries on it coincide exactly with properties \ref{tau-cov} and \ref{tr}. Thus, by forcing $\tau$-equivariance, one can extend the definition of the frame continuously on the whole $\R$, as in the proof of Proposition \ref{global}. Hence, this proves Theorem \ref{thm:Q1}.
\end{rmk}

\subsection{Extending to the face: a $\Z_2$ obstruction} \label{sec:FaceExtension}

In order to see whether it is possible to extend the frame $\widehat{\Phi}$ to a continuous symmetric Bloch frame $\Phi$ defined on the whole \effective unit cell $\Bred$, we first introduce the unitary map $\widehat{U}(k)$ which maps the input frame $\Psi(k)$ to the frame $\widehat{\Phi}(k)$, \ie such that \begin{equation} \label{hatU}
\widehat{\Phi}(k) = \Psi(k) \act \widehat{U}(k), \quad k \in \partial \Bred
\end{equation}
(compare \eqref{Healer}). This defines a continuous map $\widehat{U} \colon \partial \Bred \to \U(\C^m)$; we are interested in finding a continuous extension $U \colon \Bred \to \U(\C^m)$ of $\widehat{U}$ to the \effective unit cell. 

From a topological viewpoint, $\partial \Bred$ is homeomorphic to a circle $S^1$. It is well-known \cite[Thm. 17.3.1]{Dubrovin} that, if $X$ is a topological space, then a continuous map $f \colon S^1 \to X$ defines an element in the fundamental group $\pi_1(X)$ by taking its homotopy class $[f]$. Moreover, $f$ extends to a continuous map $F \colon D^2 \to X$, where $D^2$ is the $2$-dimensional disc enclosed by the circle $S^1$, if and only if $[f] \in \pi_1(X)$ is the trivial element. In our case, the space $X$ is the group $\U(\C^m)$, and it is also well-known \cite[Ch. 8, Sec. 12]{Husemoller} that the exact sequence of groups
\[ 1 \longrightarrow \mathcal{S}\U(\C^m) \longrightarrow \U(\C^m) \xrightarrow{\det} U(1) \longrightarrow 1 \]
induces an isomorphism $\pi_1(\U(\C^m)) \simeq \pi_1(U(1))$. On the other hand, the degree homomorphism \cite[\S 13.4(b)]{Dubrovin}
\begin{equation} \label{deg} 
\deg \colon \pi_1(U(1)) \stackrel{\sim}{\longrightarrow} \Z, \quad [\varphi \colon S^1 \to U(1)] \mapsto \frac{1}{2\pi \iu} \oint_{S^1} \di z \, \partial_z \log \varphi(z)
\end{equation}
establishes an isomorphism of groups $\pi_1(U(1)) \simeq \Z$. We conclude that a continuous map $f \colon \partial \Bred \to \U(\C^m)$ can be continuously extended%
\footnote{Explicit formulas for such a continuous extension can be found in \cite[Remark 4.5]{FiPaPi}.} %
to $F \colon \Bred \to \U(\C^m)$ if and only if $\deg([\det f]) \in \Z$ is zero.

In our case, we want to extend the continuous map $\widehat{U} \colon \partial \Bred \to \U(\C^m)$ to the whole \effective unit cell $\Bred$. However, rather than checking whether $\deg([\det \widehat{U}])$ vanishes, it is sufficient to find a unitary-matrix-valued map that ``unwinds'' the determinant of $\widehat{U}(k)$, while preserving the relevant symmetries on Bloch frames. More precisely, the following result holds.

\begin{prop} \label{Xi}
Let $\widehat{\Phi}$ be the Bloch frame defined on $\partial \Bred$ that appears in \eqref{hatU}, satisfying the vertex conditions \eqref{VertexCondition} and the edge symmetries \eqref{EdgeSymmetries}. Assume that there exists a continuous map $X \colon \partial \Bred \to \U(\C^m)$ such that
\begin{enumerate}[label=$(\mathrm{X}_{\arabic*})$,ref=$(\mathrm{X}_{\arabic*})$]
\item \label{XiDeg} $\deg([\det X]) = - \deg([\det \widehat{U}])$, and
\item \label{XiSymm} also the frame $\widehat{\Phi} \act X$ satisfies \eqref{VertexCondition} and \eqref{EdgeSymmetries}.
\end{enumerate}
Then there exists a global continuous symmetric Bloch frame $\Phi$ that extends $\widehat{\Phi} \act X$ to the whole $\R^2$.

Conversely, if $\Phi$ is a global continuous symmetric Bloch frame, then its restriction to $\partial \Bred$ differs from $\widehat{\Phi}$ by the action of a unitary-matrix-valued continuous map $X$, satisfying \ref{XiDeg} and \ref{XiSymm} above.
\end{prop}
\begin{proof}
If a map $X$ as in the statement of the Proposition exists, then the map $U := \widehat{U} X \colon \partial \Bred \to \U(\C^m)$ satisfies $\deg([\det U]) = 0$ (because $\deg$ is a group homomorphism), and hence extends continuously to $\eff{U} \colon \Bred \to \U(\C^m)$. This allows to define a continuous symmetric Bloch frame $\eff{\Phi}(k) := \Psi(k) \act \eff{U}(k)$ on the whole \effective unit cell $\Bred$, and by Proposition \ref{global} this definition can be then extended continuously to $\R^2$ to obtain the desired global continuous symmetric Bloch frame $\Phi$.

Conversely, if a global continuous symmetric Bloch frame $\Phi$ exists, then its restriction $\eff{\Phi}$ to the boundary of the \effective unit cell satisfies $\eff{\Phi}(k) = \Psi(k) \act \eff{U}(k)$ for some unitary matrix $\eff{U}(k) \in \U(\C^m)$, and moreover $\deg([\det \eff{U}]) = 0$ because $\eff{U} \colon \partial \Bred \to \U(\C^m)$ extends to the whole \effective unit cell. From \eqref{hatU} we deduce that $\eff{\Phi}(k) = \widehat{\Phi}(k) \act \left(  \widehat{U}(k)^{-1} \eff{U}(k) \right)$; the unitary matrix $X(k) :=  \widehat{U}(k)^{-1} \eff{U}(k)$ then satisfies $\deg([\det X]) = - \deg([\det \widehat{U}])$, and, when restricted to $\partial \Bred$, both $\eff{\Phi}$ and $\widehat{\Phi}$ have the same symmetries, namely \eqref{VertexCondition} and \eqref{EdgeSymmetries}.
\end{proof}

Proposition \ref{Xi} reduces the question of existence of a global continuous symmetric Bloch frame to that of existence of a continuous map $X \colon \partial \Bred \to \U(\C^m)$ satisfying conditions \ref{XiDeg} and \ref{XiSymm}. We begin by imposing condition \ref{XiSymm} on $X$, and then check its compatibility with \ref{XiDeg}. 

We spell out explicitly what it means for the Bloch frame $\widehat{\Phi} \act X$ to satisfy the edge symmetries \eqref{EdgeSymmetries}, provided that $\widehat{\Phi}$ satisfies them. For $k=(0,k_2) \in E_1 \cup E_6$, we obtain that
\begin{align*}
\widehat{\Phi}(0,-k_2) \act X(0,-k_2) & = \Theta \left( \widehat{\Phi}(0,k_2) \act X(0,k_2) \right) \act \eps \\
& \Updownarrow \\
\Theta \widehat{\Phi}(0,k_2) \act \left( \eps X(0,-k_2) \right) & = \Theta \widehat{\Phi}(0,k_2) \act \left( \overline{X}(0,k_2) \eps \right),
\end{align*}
by which we deduce that
\begin{equation} \label{E1E6}
\eps X(0,-k_2) = \overline{X}(0,k_2) \, \eps, \quad k_2 \in \left[ -1/2, 1/2 \right].
\end{equation}

Similarly, for $k=(1/2,k_2) \in E_3 \cup E_4$, we obtain
\begin{equation} \label{E3E4}
\eps X(1/2,-k_2) = \overline{X}(1/2,k_2) \, \eps, \quad k_2 \in \left[ -1/2, 1/2 \right].
\end{equation}

Finally, the conditions \eqref{EdgeSymmetries} for $k \in E_2$ and $k \in E_5$ are clearly the inverse one of each other, so we can treat both at once. For $k=(k_1,1/2) \in E_5$, we obtain that
\begin{equation} \label{E2E5}
X(k_1,-1/2) = X(k_1,1/2), \quad k_1 \in \left[ 0, 1/2 \right].
\end{equation} 

Thus we have shown that condition \ref{XiSymm} on $X$ is equivalent to the relations \eqref{E1E6}, \eqref{E3E4} and \eqref{E2E5}. Notice that these contain also the relations satisfied by $X(k)$ at the vertices $k = k_\lambda$, which could be obtained by imposing that the frame $\widehat{\Phi} \act X$ satifies the vertex conditions \eqref{VertexCondition} whenever $\widehat{\Phi}$ does. Explicitly, these relations on $X(k_\lambda)$ read
\begin{equation} \label{XiSympl}
\eps X(k_\lambda) = \overline{X}(k_\lambda)\, \eps, \quad \text{\ie} \quad X(k_\lambda)\Tr \eps \, X(k_\lambda) = \eps.
\end{equation}
This relation has interesting consequences. Indeed, in view of Remark \ref{rmk:eps}, we may assume that $\eps$ is in the standard symplectic form \eqref{symplectic}. Then \eqref{XiSympl} implies that the matrices $X(k_\lambda)$ belong to the \emph{symplectic group} $\mathrm{Sp}(2n, \C)$. As such, they must be unimodular \cite{Mackey}, \ie
\begin{equation} \label{unimodular} \det X(k_\lambda) = 1. \end{equation}

We now proceed in establishing how the properties on $X$ we have deduced from \ref{XiSymm} influence the possible values that the degree of the map $\xi := \det X \colon \partial \Bred \to U(1)$ can attain. The integral on the boundary $\partial \Bred$ of the \effective unit cell splits as the sum of the integrals over the oriented edges $E_1, \ldots, E_6$:
\[ \deg([\xi]) = \frac{1}{2 \pi \iu} \oint_{\partial \Bred} \di z \, \partial_z \log \det X(z) = \sum_{i=1}^{6} \frac{1}{2 \pi \iu} \int_{E_i} \di z \, \partial_z \log \det X(z). \]

Our first observation is that all the summands on the right-hand side of the above equality are integers. Indeed, from \eqref{unimodular} we deduce that all maps $\xi_i := \det X \big|_{E_i} \colon E_i \to U(1)$, $i=1,\ldots,6$, are indeed periodic, and hence have well-defined degrees: these are evaluated exactly by the integrals appearing in the above sum. We will denote by $S^1_i$ the edge $E_i$ with its endpoints identified: we have thus established that
\begin{equation} \label{DegSum}
\deg([\xi]) = \sum_{i=1}^{6} \frac{1}{2 \pi \iu} \oint_{S^1_i} \di z \, \partial_z \log \det X(z) = \sum_{i=1}^{6} \deg([\xi_i]).
\end{equation}

From Equation \eqref{E2E5}, the integrals over $E_2$ and $E_5$ compensate each other, because the integrands are the same but the orientations of the two edges are opposite. We thus focus our attention on the integrals over $S^1_1$ and $S^1_6$ (respectively on $S^1_3$ and $S^1_4$). From Equations \eqref{E1E6} and \eqref{E3E4}, we deduce that if $k_* \in \set{0, 1/2}$ then
\[ \overline{X}(k_*,k_2) = \eps X(k_*, -k_2) \eps^{-1} \]
which implies in particular that
\[ {\left( \det X(k_*,k_2) \right)}^{-1} = \overline{\det X(k_*,k_2)} = \det X(k_*, -k_2). \]
Thus, calling $z = -k_2$ the coordinate on $S^1_1$ and $S^1_6$ (so that the two circles are oriented positively with respect to $z$), we can rewrite the above equality for $k_* = 0$ as
\[ \xi_6(z) = \overline{\xi_1(-z)} \]
so that
\begin{align*}
\deg([\xi_6]) & = - \deg([\overline{\xi_1}]) && \text{(because evaluation at $(-z)$ changes the orientation of $S^1_1$)} \\
& = \deg([\xi_1]) && \text{(because if $\varphi \colon S^1 \to U(1)$ then $\deg([\overline{\varphi}]) = - \deg([\varphi])$)}.
\end{align*}
Similarly, for $k_* = 1/2$ we get (using this time the coordinate $z = k_2$ on $S^1_3$ and $S^1_4$)
\[ \deg([\xi_4]) = \deg([\xi_3]). \]
Plugging both the equalities that we just obtained in \eqref{DegSum}, we conclude that
\begin{equation} \label{degdetXi}
\begin{aligned}
\deg([\xi]) = 2 \left( \deg([\xi_1]) + \deg[\xi_3] \right) \in 2 \cdot \Z.
\end{aligned}
\end{equation}

We have thus proved the following

\begin{prop} \label{degXieven}
Let $X \colon \partial \Bred \to \U(\C^m)$ satisfy condition \ref{XiSymm}, as in the statement of Proposition \ref{Xi}. Then the degree of its determinant is even:
\[ \deg([\det X]) \in 2 \cdot \Z. \]
\end{prop}

By the above Proposition, we deduce the following $\Z_2$ classification for symmetric families of projectors in $2$ dimensions.

\begin{thm} \label{thm:Q2'}
Let $\set{P(k)}_{k \in \R^2}$ be a family of orthogonal projectors satisfying Assumption \ref{proj}. Let $\widehat{U} \colon \partial \Bred \to \U(\C^m)$ be defined as in \eqref{hatU}. Then there exists a global continuous symmetric Bloch frame $\Phi$ for $\set{P(k)}_{k \in \R^2}$ if and only if 
\[ \deg([\det \widehat{U}]) \equiv 0 \bmod{2}. \]
\end{thm}
\begin{proof}
By Proposition \ref{Xi}, we know that the existence of a global continuous symmetric Bloch frame is equivalent to that of a continuous map $X \colon \partial \Bred \to \U(\C^m)$ satisfying conditions \ref{XiDeg} and \ref{XiSymm}, so that in particular it should have $\deg([\det X]) = - \deg([\det \widehat{U}])$. In view of Proposition \ref{degXieven}, condition \ref{XiSymm} cannot hold in the case where $\deg([\det \widehat{U}])$ is odd.

In the case in which $\deg([\det \widehat{U}])$ is even, instead, it just remains to exhibit a map $X \colon \partial \Bred \to \U(\C^m)$ satisfying \eqref{E1E6}, \eqref{E3E4} and \eqref{E2E5}, and such that
\[ \deg([\det X]) = - \deg([\det \widehat{U}]) = - 2s, \quad s \in \Z. \]
In the basis where $\eps$ is of the form \eqref{eps=J}, define
\[ X(k) := \begin{cases}
\diag \left( e^{-2 \pi \iu s (k_2+1/2)} , e^{-2 \pi \iu s (k_2+1/2)} , 1 , \ldots, 1 \right) & \text{if } k \in E_3 \cup E_4, \\
\Id & \text{otherwise}.
\end{cases} \]
One checks at once that $X(k)$ satisfies \eqref{E1E6}, \eqref{E3E4} and \eqref{E2E5} -- which are equivalent to \ref{XiSymm}, as shown before --, and defines a continuous map $X \colon \partial \Bred \to \U(\C^m)$. Since $X$ is constant on $E_1$, formula \eqref{degdetXi} for the degree of the determinant of $X(k)$ simplifies to
\[ \deg([\det X]) = 2 \left( \frac{1}{2 \pi \iu} \oint_{S^1_3} \di z \, \partial_z \log \det X(z) \right) = 2 \left( \frac{1}{2 \pi \iu} \int_{-1/2}^{0} \di k_2 \, \partial_{k_2} \log \det X(1/2,k_2) \right). \]
One immediately computes $\deg([\det X]) = - 2 s$, as wanted.
\end{proof}

\goodbreak

The index
\begin{equation} \label{delta}
\delta(P) := \deg([\det \widehat{U}]) \bmod{2}
\end{equation}
is thus the \textbf{$\Z_2$ topological invariant} (see Section \ref{sec:TopInv} below) of the family of projectors $\set{P(k)}_{k \in \R^2}$, satisfying Assumption \ref{proj}, which encodes the obstruction to the existence of a global continuous symmetric Bloch frame. One of our main results, Theorem \ref{thm:Q2}, is then reduced to Theorem \ref{thm:Q2'}. 

\subsection{\texorpdfstring{Well-posedness of the definition of $\delta$}{Well-posedness of the definition of delta}} \label{app:Independence}

In the construction of the previous Subsection, leading to the definition \eqref{delta} of the $\Z_2$ index $\delta$, a number of choices has to be performed, namely the input frame $\Psi$ and the interpolation $\widetilde{U}$ on $E_1 \cup E_2 \cup E_3$. This Subsection is devoted to showing that the value of the index $\delta(P) \in \Z_2$ is \emph{independent} of such choices, and thus is really associated with the bare family of projectors $\set{P(k)}_{k \in \R^2}$. Moreover, the index $\delta$ is also independent of the choice of a basis $\set{e_1, e_2}$ for the lattice $\Lambda$, as will be manifest from the equivalent formulation \eqref{easydelta} of the invariant we will provide in Section \ref{sec:EvalVertices}.

\subsubsection{\texorpdfstring{\textbf{Gauge independence of $\delta$}}{Gauge independence of delta}} \label{sec:Gauge}

As a first step, we will prove that the $\Z_2$ index $\delta$ is independent of the choice of the input Bloch frame $\Psi$.

Indeed, assume that another input Bloch frame $\primo{\Psi}$ is chosen: the two frames will be related by a continuous unitary gauge transformation, say
\begin{equation} \label{gauge}
\primo{\Psi}(k) = \Psi(k) \act G(k), \quad G(k) \in \U(\C^m), \quad k \in \Bred.
\end{equation}
These two frames will produce, via the procedure illustrated above, two symmetric Bloch frames defined on $\partial \Bred$, namely
\[ \widehat{\Phi}(k) = \Psi(k) \act \widehat{U}(k) \quad \text{and} \quad \primo{\widehat{\Phi}}(k) = \primo{\Psi}(k) \act \primo{\widehat{U}}(k), \qquad k \in \partial \Bred. \]
From \eqref{gauge} we can rewrite the second equality as
\[ \primo{\widehat{\Phi}}(k) = \Psi(k) \act \left(G(k) \primo{\widehat{U}}(k)\right), \quad k \in \partial \Bred. \]
The two frames $\widehat{\Phi}$ and $\primo{\widehat{\Phi}}$ both satisfy the vertex conditions \eqref{VertexCondition} and the edge symmetries \eqref{EdgeSymmetries}, hence the matrix $X(k) := \left(G(k) \primo{\widehat{U}}(k)\right)^{-1} \widehat{U}(k)$, which transforms $\primo{\widehat{\Phi}}$ into $\widehat{\Phi}$, enjoys condition \ref{XiSymm} from Proposition \ref{Xi}. Applying Proposition \ref{degXieven} we deduce that
\begin{equation} \label{U=GU'} 
\deg([\det \widehat{U}]) \equiv \deg\left(\left[\det \left(G \primo{\widehat{U}} \right)\right]\right) \mod 2. 
\end{equation}

Observe that, since the degree is a group homomorphism,
\[ \deg\left(\left[\det \left(G \primo{\widehat{U}}\right)\right]\right) = \deg([\det \primo{\widehat{U}}]) + \deg([\det G]). \]
The matrix $G(k)$ is by hypothesis defined and continuous on the whole effective unit cell: this implies that the degree of its determinant along the boundary of $\Bred$ vanishes, because its restriction to $\partial \Bred$ extends continuously to the interior of $\Bred$. Thus we conclude that $\deg([\det (G \primo{\widehat{U}})]) = \deg([\det \primo{\widehat{U}}])$; plugging this in \eqref{U=GU'} we conclude that
\[ \delta = \deg([\det \widehat{U}]) \equiv \deg([\det \primo{\widehat{U}}]) = \primo{\delta} \mod 2, \]
as we wanted.

\subsubsection{\textbf{Invariance under edge extension}} \label{app:EdgeInvariance}

Recall that, after solving the vertex conditions and finding the value $U(k_\lambda) = \widehat{U}(k_\lambda)$ at the vertices $k_\lambda$, we interpolated those -- using the path-connectedness of the group $\U(\C^m)$ -- to obtain the definition of $\widehat{U}(k)$ first for $k \in E_1 \cup E_2 \cup E_3$, and then, imposing the edge symmetries, extended it to the whole $\partial \Bred$ (see Sections \ref{sec:VertexConditions} and \ref{sec:EdgeExtension}). We now study how a change in this interpolation affects the value of $\deg([\det \widehat{U}])$.

Assume that a different interpolation $\primo{\widetilde{U}}(k)$ has been chosen on $E_1 \cup E_2 \cup E_3$, leading to a different unitary-matrix-valued map $\primo{\widehat{U}} \colon \partial \Bred \to \U(\C^m)$. Starting from the input frame $\Psi$, we thus obtain two different Bloch frames:
\[ \widehat{\Phi}(k) = \Psi(k) \act \widehat{U}(k) \quad \text{and} \quad \primo{\widehat{\Phi}}(k) = \Psi(k) \act \primo{\widehat{U}}(k), \qquad k \in \partial \Bred. \]
From the above equalities, we deduce at once that $\widehat{\Phi}(k) = \primo{\widehat{\Phi}}(k) \act X(k)$, where $X(k) := \primo{\widehat{U}}(k)^{-1} \widehat{U}(k)$. We now follow the same line of argument as in the previous Subsection. Since both $\widehat{\Phi}$ and $\primo{\widehat{\Phi}}$ satisfy the vertex conditions \eqref{VertexCondition} and the edge symmetries \eqref{EdgeSymmetries} by construction, the matrix-valued map $X \colon \partial \Bred \to \U(\C^m)$ just defined enjoys condition \ref{XiSymm}, as in the statement of Proposition \ref{Xi}. It follows now by Proposition \ref{degXieven} that the degree $\deg([\det X])$ is even. Since the degree defines a group homomorphism, this means that
\[ \deg([\det X]) = \deg([\det \widehat{U}]) - \deg([\det \primo{\widehat{U}}]) \equiv 0 \bmod 2 \]
or equivalently
\[ \delta = \deg([\det \widehat{U}]) \equiv \deg([\det \primo{\widehat{U}}]) = \primo{\delta} \bmod 2, \]
as claimed.

\subsection{\texorpdfstring{Topological invariance of $\delta$}{Topological invariance of delta}} \label{sec:TopInv}

The aim of this Subsection is to prove that the definition \eqref{delta} of the $\Z_2$ index $\delta(P)$ actually provides a \emph{topological invariant} of the family of projectors $\set{P(k)}_{k \in \R^2}$, with respect to those continuous deformation preserving the relevant symmetries specified in Assumption \ref{proj}. More formally, the following result holds.

\begin{prop} \label{deltaTopInv}
Let $\set{P_0(k)}_{k \in \R^2}$ and $\set{P_1(k)}_{k \in \R^2}$ be two families of projectors satisfying Assumption \ref{proj}. Assume that there exists a homotopy $\set{P_t(k)}_{k \in \R^2}$, $t \in [0,1]$, between $\set{P_0(k)}_{k \in \R^2}$ and $\set{P_1(k)}_{k \in \R^2}$, such that $\set{P_t(k)}_{k \in \R^2}$ satisfies Assumption \ref{proj} for all $t \in [0,1]$. Then
\[ \delta\left(P_0\right) = \delta\left(P_1\right) \in \Z_2. \]
\end{prop}
\begin{proof}
The function $(t, k) \mapsto P_t(k)$ is a continuous function on the compact set $[0,1] \times \Bred$, and hence is uniformly continuous. Thus, there exists $\mu > 0$ such that
\[ \norm{P_t(k) - P_{\tilde{t}}(\tilde{k})}_{\mathcal{B}(\Hi)} < 1 \quad \text{if } \max \set{\left|k-\tilde{k}\right|, \left|t-\tilde{t}\right|} < \mu. \]
In particular, choosing $t_0 < \mu$, we have that
\begin{equation} \label{norm<1}
\norm{P_t(k) - P_0(k)}_{\mathcal{B}(\Hi)} < 1 \quad \text{if } t \in [0,t_0], \text{ uniformly in } k \in \Bred.
\end{equation}
We will show that $\delta(P_0) = \delta(P_{t_0})$. Iterating this construction a finite number of times will prove that $\delta(P_0) = \delta(P_1)$.

In view of \eqref{norm<1}, the Kato-Nagy unitary \cite[Sec. I.6.8]{Kato}
\[ W(k) := \left( \Id - (P_0(k) - P_{t_0}(k))^2 \right)^{-1/2} \big( P_{t_0}(k) P_0(k) + (\Id - P_{t_0}(k))(\Id - P_0(k))\big) \: \in \U(\Hi)\]
is well-defined and provides an intertwiner between $\Ran P_0(k)$ and $\Ran P_{t_0}(k)$, namely
\[ P_{t_0}(k) = W(k) P_0(k) W(k)^{-1}. \]
Moreover, one immediately realizes that the map $k \mapsto W(k)$ inherits from the two families of projectors the following properties:
\begin{enumerate}[label=$(\mathrm{W}_\arabic*)$,ref=$(\mathrm{W}_\arabic*)$]
\item \label{item:Wsmooth} the map $\Bred \ni k \mapsto W(k) \in \U(\Hi)$ is $C^\infty$-smooth;
\item \label{item:Wtau} the map $k \mapsto W(k)$ is $\tau$-covariant, in the sense that whenever $k$ and $k + \lambda$ are both in $\Bred$ for $\lambda \in \Lambda$, then
\[ W(k+\lambda) = \tau(\lambda) W(k) \tau(\lambda)^{-1}; \]
\item \label{item:WTRS} the map $k \mapsto W(k)$ is time-reversal symmetric, in the sense that whenever $k$ and $-k$ are both in $\Bred$, then
\[ W(-k) = \Theta W(k) \Theta^{-1}. \]
\end{enumerate}

Let now $\set{\Psi_0(k)}_{k \in \Bred}$ be a continuous Bloch frame for $\set{P_0(k)}_{k \in \R^2}$. Extending the action of the unitary $W(k) \in \U(\Hi)$ to $m$-frames component-wise, we can define $\Psi_{t_0}(k) = W(k) \Psi_0(k)$ for $k \in \Bred$, and obtain a continuous Bloch frame $\set{\Psi_{t_0}(k)}_{k \in \Bred}$ for $\set{P_{t_0}(k)}_{k \in \R^2}$. Following the procedure illustrated in Sections \ref{sec:VertexConditions} and \ref{sec:EdgeExtension}, we can produce two symmetric Bloch frames for $\set{P_0(k)}_{k \in \R^2}$ and $\set{P_{t_0}(k)}_{k \in \R^2}$, namely
\[ \widehat{\Phi}_0(k) = \Psi_0(k) \act \widehat{U}_0(k) \quad \text{and} \quad \widehat{\Phi}_{t_0}(k) = \Psi_{t_0}(k) \act \widehat{U}_{t_0}(k), \quad k \in \partial \Bred. \]
From the above equalities, we deduce that
\begin{align*}
\widehat{\Phi}_{t_0}(k) & = \Psi_{t_0}(k) \act \widehat{U}_{t_0}(k) = \left( W(k) \Psi_0(k) \right) \act \widehat{U}_{t_0}(k) = \\
& = W(k) \left( \Psi_0(k) \act \widehat{U}_{t_0}(k) \right) = W(k) \left( \widehat{\Phi}_0(k) \act \left( \widehat{U}_0(k)^{-1} \widehat{U}_{t_0}(k) \right) \right) = \\
& = \left( W(k) \widehat{\Phi}_0(k) \right) \act \left( \widehat{U}_0(k)^{-1} \widehat{U}_{t_0}(k) \right)
\end{align*}
where we have used that the action of a linear operator in $\mathcal{B}(\Hi)$, extended component-wise to $m$-frames, commutes with the right-action of $\U(\C^m)$. It is easy to verify that, in view of properties \ref{item:Wsmooth}, \ref{item:Wtau} and \ref{item:WTRS}, the frame $\set{W(k) \widehat{\Phi}_0(k)}_{k \in \partial \Bred}$ gives a continuous Bloch frame for $\set{P_{t_0}(k)}_{k \in \R^2}$ which still satisfies the vertex conditions \eqref{VertexCondition} and the edge symmetries \eqref{EdgeSymmetries}, since $\widehat{\Phi}_0$ does. Hence the matrix $X(k) := \widehat{U}_0(k)^{-1} \widehat{U}_{t_0}(k)$ satisfies hypothesis \ref{XiSymm}, and in view of Proposition \ref{degXieven} its determinant has even degree. We conclude that
\[ \delta(P_0) = \deg([\det \widehat{U}_0]) \equiv \deg([\det \widehat{U}_{t_0}]) = \delta(P_{t_0}) \bmod 2. \qedhere \]
\end{proof}

\newpage

%%%%% SECTION 4

\section{Comparison with the Fu-Kane index} \label{sec:Fu-Kane}

In this Section, we will compare our $\Z_2$ invariant $\delta$ with the $\Z_2$ invariant $\Delta$ proposed by Fu and Kane \cite{FuKa}, and show that they are equal. For the reader's convenience, we recall the definition of $\Delta$, rephrasing it in our terminology.

\medskip

Let $\Psi(k) = \set{\psi_1(k), \ldots, \psi_m(k)}$ be a global \emph{$\tau$-equivariant}%
\footnote{This assumption is crucial in the definition of the Fu-Kane index, whereas the definition of our $\Z_2$ invariant does not need $\Psi$ to be already $\tau$-equivariant. Nonetheless, the existence of a global $\tau$-equivariant Bloch frame is guaranteed by a straightforward modification of the proof in \cite{Panati}, as detailed in \cite{PaMo}.} %
Bloch frame. Define the unitary matrix $w(k) \in \U(\C^m)$ by
\[ w(k)_{ab} := \scal{\psi_a(-k)}{\Theta \psi_b(k)}. \]
In terms of the right action of matrices on frames, one can equivalently say that $w(k)$ is the matrix such that%
\soloarXiv{
\footnote{Indeed, spelling out \eqref{w(k)} one obtains
\[ \Theta \psi_b(k) = \sum_{c=1}^{m} \psi_c(-k) w(k)_{cb} \]
and taking the scalar product of both sides with $\psi_a(-k)$ yields
\[ \scal{\psi_a(-k)}{\Theta \psi_b(k)} = \sum_{c=1}^{m} \scal{\psi_a(-k)}{\psi_c(-k)} w(k)_{cb} = \sum_{c=1}^{m} \delta_{ac} w(k)_{cb} = w(k)_{ab} \]
because $\set{\psi_a(-k)}_{a=1, \ldots, m}$ is an \emph{orthonormal} frame in $\Ran P(-k)$.}}%
\begin{equation} \label{w(k)}
\Theta \Psi(k) = \Psi(-k) \act w(k).
\end{equation}
Comparing \eqref{w(k)} with \eqref{TR}, we see that if $\Psi$ were already symmetric then $w(k) = \eps^{-1}$ for all $k \in \R^d$.

One immediately checks how $w(k)$ changes when the inversion or translation symmetries are applied to $k$. One has that
\begin{align*}
w(\theta(k))_{ab} & = \scal{\psi_a(k)}{\Theta \psi_b(-k)} = \scal{\Theta^2 \psi_b(-k)}{\Theta \psi_a(k)} = \\ 
& = - \scal{\psi_b(-k)}{\Theta \psi_a(k)} = - w(k)_{ba}
\end{align*}
or in matrix form $w(\theta(k)) = - w(k)\Tr$. Moreover, by using the $\tau$-equivariance of the frame $\Psi$ one obtains
\begin{align*}
w(t_\lambda(k))_{ab} & = \scal{\psi_a(\theta t_\lambda(k))}{\Theta \psi_b(t_\lambda(k))} = \scal{\psi_a(t_{-\lambda} \theta(k))}{\Theta \tau(\lambda) \psi_b(k)} = \\
& = \scal{\tau(-\lambda) \psi_a(-k)}{\tau(-\lambda) \Theta \psi_b(k)} = \scal{\psi_a(-k)}{\Theta \psi_b(k)} = w(k)_{ab}
\end{align*}
because $\tau(-\lambda)$ is unitary; in matrix form we can thus write $w(t_{\lambda}(k)) = w(k)$.

Combining both these facts, we see that at the vertices $k_\lambda$ of the \effective unit cell we get
\[ - w(k_\lambda)\Tr = w(\theta(k_\lambda)) = w(t_{-\lambda}(k_\lambda)) = w(k_\lambda). \]
In other words, the matrix $w(k_\lambda)$ is skew-symmetric, and hence it has a well-defined Pfaffian, satisfying $\left(\Pf w(k_\lambda)\right)^2 = \det w(k_\lambda)$. The Fu-Kane index $\Delta$ is then defined as \cite[Equations (3.22) and (3.25)]{FuKa}
\begin{equation} \label{Delta}
\Delta := P_\theta(1/2) - P_\theta(0) \bmod 2,
\end{equation}
where for $k_* \in \set{0, 1/2}$
\begin{equation} \label{Ptheta}
P_\theta(k_*) := \frac{1}{2\pi \iu} \left( \int_{0}^{1/2} \di k_2 \, \tr\left(w(k_*,k_2)^* \partial_{k_2} w(k_*,k_2) \right) - 2 \log \frac{\Pf w(k_*,1/2)}{\Pf w(k_*,0)} \right).
\end{equation}

\medskip

We are now in position to prove the above-mentioned equality between our $\Z_2$ invariant $\delta$ and the Fu-Kane index $\Delta$.

\begin{thm} \label{thm:delta=Delta}
Let $\set{P(k)}_{k \in \R^2}$ be a family of projectors satisfying Assumption \ref{proj}, and let $\delta = \delta(P) \in \Z_2$ be as in \eqref{delta}. Let $\Delta \in \Z_2$ be the Fu-Kane index defined in \eqref{Delta}. Then
\[ \delta = \Delta \in \Z_2. \]
\end{thm}
\begin{proof}
First we will rewrite our $\Z_2$ invariant $\delta$, in order to make the comparison with $\Delta$ more accessible. Recall that $\delta$ is defined as the degree of the determinant of the matrix $\widehat{U}(k)$, for $k \in \partial \Bred$, satisfying $\widehat{\Phi}(k) = \Psi(k) \act \widehat{U}(k)$, where $\widehat{\Phi}(k)$ is as in \eqref{Psihat}. The unitary $\widehat{U}(k)$ coincides with the unitary $\widetilde{U}(k)$ for $k \in E_1 \cup E_2 \cup E_3$, where $\widetilde{U}(k)$ is a continuous path interpolating the unitary matrices $U(v_i)$, $i = 1, 2, 3, 4$, \ie the solutions to the vertex conditions (compare Sections \ref{sec:VertexConditions} and \ref{sec:EdgeExtension}). The definition of $\widehat{U}(k)$ for $k \in E_4 \cup E_5 \cup E_6$ is then obtained by imposing the edge symmetries on the corresponding frame $\widehat{\Phi}(k)$, as in Section \ref{sec:EdgeExtension}.

We first compute explicitly the extension of $\widehat{U}$ to $k \in E_4 \cup E_5 \cup E_6$. For $k \in E_4$, say $k = (1/2, k_2)$ with $k_2 \ge 0$, we obtain
\begin{align*}
\widehat{\Phi}(k) & = \tau_{e_1} \Theta \widetilde{\Phi} \left( t_{e_1} \theta(k) \right) \act \eps = \\
& = \tau_{e_1} \Theta \left( \Psi \left( t_{e_1} \theta(k) \right) \act \widetilde{U}\left(t_{e_1} \theta(k)\right) \right) \act \eps = \\
& = \left( \tau_{e_1} \tau_{-e_1} \Theta \Psi(-k) \right) \act \left( \overline{\widetilde{U}}(1/2,-k_2) \eps \right) = \\
& = \left( \Psi(k) \act w(-k) \right) \act \left( \overline{\widetilde{U}}(1/2,-k_2) \eps \right) = \\
& = \Psi(k) \act \left( w(-k) \overline{\widetilde{U}}(1/2,-k_2) \eps \right).
\end{align*}
This means that
\begin{equation} \label{hatUE4}
\widehat{U}(k) = w(-k) \overline{\widetilde{U}}(1/2,-k_2) \eps, \quad k = (1/2, k_2) \in E_4.
\end{equation}

Analogously, for $k \in E_5$, say $k = (k_1,1/2)$, we obtain
\begin{equation} \label{hatUE5}
\widehat{U}(k) = \widetilde{U}(k_1,-1/2), \quad k = (k_1,1/2) \in E_5.
\end{equation}

Finally, for $k \in E_6$, say $k = (0, k_2)$ with $k_2 \ge 0$, we obtain in analogy with \eqref{hatUE4}
\begin{equation} \label{hatUE6}
\widehat{U}(k) = w(-k) \overline{\widetilde{U}}(0,-k_2) \eps, \quad k = (0, k_2) \in E_6.
\end{equation}
Notice that, as $w(-1/2,k_2) = w(t_{-e_1}(1/2,k_2)) = w(1/2,k_2)$, we can actually summarize \eqref{hatUE4} and \eqref{hatUE6} as
\[ \widehat{U}(k_*,k_2) = w(k_*,-k_2) \overline{\widetilde{U}}(k_*,-k_2) \eps, \quad k_* \in \set{0, 1/2}, \: k_2 \in [0,1/2]. \]
Since $\widehat{U}$ and $\widetilde{U}$ coincide on $E_1$ and $E_3$, we can further rewrite this relation as
\begin{equation} \label{hatUE4E6}
w(k_*,k_2) = \widehat{U}(k_*,-k_2) \eps^{-1} \widehat{U}(k_*,k_2)\Tr, \quad (k_*,k_2) \in S^1_*,
\end{equation}
where $S^1_*$ denotes the edge $E_1 \cup E_6$ for $k_* = 0$ (respectively $E_3 \cup E_4$ for $k_* = 1/2$), with the edge-points identified: the identification is allowed in view of \eqref{hatUE5} and the fact that $w(t_\lambda(k)) = w(k)$.

\medskip

We are now able to compute the degree of the determinant of $\widehat{U}$. Firstly, an easy computation%
\soloarXiv{
\footnote{If $\widehat{U}(z) \in \U(\C^m)$ has spectrum $\set{\eu^{\iu \lambda_j(z)}}_{j=1,\ldots,m}$, then
\[ \partial_z \log \det \widehat{U}(z) = \sum_{j=1}^{m} \eu^{-\iu \lambda_j(z)} \partial_z \eu^{\iu \lambda_j(z)}. \]
On the other hand, writing the spectral decomposition of $\widehat{U}(z)$ as
\[ \widehat{U}(z) = \sum_{j=1}^{m} \eu^{\iu \lambda_j(z)} P_j(z), \quad P_j(z)^* = P_j(z) = P_j(z)^2, \quad \tr(P_j(z)) = 1, \]
one immediately deduces that
\[ \widehat{U}(z)^* = \sum_{j=1}^{m} \eu^{-\iu \lambda_j(z)} P_j(z) \quad \text{and} \quad \partial_z \widehat{U}(z) = \sum_{j=1}^{m} \left( \partial_z \eu^{\iu \lambda_j(z)} \right) P_j(z) + \eu^{\iu \lambda_j(z)} (\partial_z P_j(z)). \]

Notice now that, taking the derivative with respect to $z$ of the equality $P_j(z)^2 = P_j(z)$, one obtains
\[ P_j(z) \partial_z P_j(z) = \left(\partial_z P_j(z) \right) \left( \Id - P_j(z) \right) \]
so that
\[ \partial_z P_j(z) = P_j(z) \left( \partial_z P_j(z) \right) \left( \Id - P_j(z) \right) + \left( \Id - P_j(z) \right) \left( \partial_z P_j(z) \right) P_j(z). \]
From this, it follows by the ciclicity of the trace and the relation $P_j(z) P_\ell(z) = \delta_{j,\ell} P_j(z) = P_\ell(z) P_j(z)$ that
\begin{align*}
\tr \left( P_\ell(z) \left(\partial_z P_j(z) \right) \right) = & \tr \left( P_j(z) P_\ell(z) \left( \partial_z P_j(z) \right) \left( \Id - P_j(z) \right) \right) + \\
& + \tr \left( \left( \Id - P_j(z) \right) P_\ell(z) \left( \partial_z P_j(z) \right) P_j(z) \right) = 0.
\end{align*}

We are now able to compute
\begin{align*}
\tr \left( \widehat{U}(z)^* \partial_z \widehat{U}(z) \right) & = \sum_{j, \ell=1}^{m}  \tr \left( \eu^{-\iu \lambda_\ell(z)} P_\ell(z) \left( \partial_z \eu^{\iu \lambda_j(z)} \right) P_j(z)\right) + \\
& \quad + \eu^{-\iu \left(\lambda_\ell(z) - \lambda_j(z)\right)} \tr \left(P_\ell(z)\left( \partial_j P(z) \right) \right) = \\
& = \sum_{j,\ell=1}^{m} \eu^{-\iu \lambda_\ell(z)} \left( \partial_z \eu^{\iu \lambda_j(z)} \right) \delta_{j,\ell} \tr \left( P_j(z) \right) = \sum_{j=1}^{m} \eu^{-\iu \lambda_j(z)} \partial_z \eu^{\iu \lambda_j(z)}.
\end{align*}
}} %
shows that
\begin{align*}
\deg([\det \widehat{U}]) & = \frac{1}{2 \pi \iu} \oint_{\partial \Bred} \di z \, \partial_z \log \det \widehat{U}(z) = \\
& = \frac{1}{2 \pi \iu} \oint_{\partial \Bred} \di z \, \tr \left( \widehat{U}(z)^* \partial_z \widehat{U}(z) \right) = \sum_{i=1}^{6} \frac{1}{2 \pi \iu} \int_{E_i} \di z \, \tr \left( \widehat{U}(z)^* \partial_z \widehat{U}(z) \right).
\end{align*}
Clearly we have
\begin{equation} \label{E1+E2+E3}
\int_{E_i} \di z \, \tr \left( \widehat{U}(z)^* \partial_z \widehat{U}(z) \right) = \int_{E_i} \di z \, \tr \left( \widetilde{U}(z)^* \partial_z \widetilde{U}(z) \right) \quad \text{for } i = 1, 2, 3,
\end{equation}
because $\widehat{U}$ and $\widetilde{U}$ coincide on $E_1 \cup E_2 \cup E_3$. Using now Equation \eqref{hatUE5} we compute
\begin{align*}
\int_{E_5} \di z \, \tr \left( \widehat{U}(z)^* \partial_z \widehat{U}(z) \right) & = \int_{1/2}^{0} \di k_1 \, \tr \left( \widehat{U}(k_1,1/2)^* \partial_{k_1} \widehat{U}(k_1,1/2) \right) = \\
& = - \int_{0}^{1/2} \di k_1 \, \tr \left( \widetilde{U}(k_1,-1/2)^* \partial_{k_1} \widetilde{U}(k_1,-1/2) \right) = \\
& = - \int_{E_2} \di z \, \tr \left( \widetilde{U}(z)^* \partial_z \widetilde{U}(z) \right)
\end{align*}
or equivalently, in view of \eqref{E1+E2+E3} for $i=2$,
\begin{equation} \label{E2+E5}
\int_{E_2 + E_5} \di z \, \tr \left( \widehat{U}(z)^* \partial_z \widehat{U}(z) \right) = 0.
\end{equation}

Making use of \eqref{hatUE4E6}, we proceed now to evaluate the integrals on $E_1, E_3, E_4$ and $E_6$, which give the non-trivial contributions to $\deg([\det \widehat{U}])$. Indeed, Equation \eqref{hatUE4E6} implies
\begin{align*}
w(k_*,k_2)^* & = \overline{\widehat{U}}(k_*,k_2) \, \eps \, \widehat{U}(k_*,-k_2)^*, \\
\partial_{k_2} w(k_*,k_2) & = - \partial_{k_2} \widehat{U}(k_*,-k_2) \eps^{-1} \widehat{U}(k_*,k_2)\Tr + \widehat{U}(k_*,-k_2) \eps^{-1} \partial_{k_2} \widehat{U}(k_*,k_2)\Tr,
\end{align*} 
from which one computes
\begin{align*}
\tr ( w(k_*,k_2)^* \partial_{k_2} w(k_*,k_2) ) = & - \tr \left( \widehat{U}(k_*,-k_2)^* \partial_{k_2} \widehat{U}(k_*,-k_2) \right) + \\ 
& + \tr \left( \widehat{U}(k_*,k_2)^* \partial_{k_2} \widehat{U}(k_*,k_2) \right).
\end{align*}
Integrating both sides of the above equality for $k_2 \in [0,1/2]$ yields to
\begin{equation} \label{quasiDelta}
\begin{aligned}
\int_{0}^{1/2} \di k_2 \, \tr ( w(k_*,k_2)^* \partial_{k_2} w(k_*,k_2) ) & = - \int_{-1/2}^{0} \di k_2 \, \tr \left( \widehat{U}(k_*,k_2)^* \partial_{k_2} \widehat{U}(k_*,k_2) \right) + \\
& \quad + \int_{0}^{1/2} \di k_2 \, \tr \left( \widehat{U}(k_*,k_2)^* \partial_{k_2} \widehat{U}(k_*,k_2) \right) = \\
& = (-1)^{1+2 k_*} \oint_{S^1_*} \di z \, \tr \left( \widehat{U}(z)^* \partial_{z} \widehat{U}(z) \right) + \\
& \quad + 2 \int_{0}^{1/2} \di k_2 \, \tr \left( \widehat{U}(k_*,k_2)^* \partial_{k_2} \widehat{U}(k_*,k_2) \right).
\end{aligned}
\end{equation}
The sign $s_* := (-1)^{1+2 k_*}$ appearing in front of the integral along $S^1_*$ depends on the different orientations of the two circles, for $k_* = 0$ and for $k_* = 1/2$, parametrized by the coordinate $z = k_2$.

Now, notice that
\begin{equation} \label{detU/detU}
\begin{aligned}
\int_{0}^{1/2} \di k_2 \, \tr \left( \widehat{U}(k_*,k_2)^* \partial_{k_2} \widehat{U}(k_*,k_2) \right) & = \int_{0}^{1/2} \di k_2 \, \partial_{k_2} \log \det \widehat{U}(k_*,k_2) \\
& \equiv \log \frac{\det \widehat{U}(k_*,1/2)}{\det \widehat{U}(k_*,0)} \bmod 2 \pi \iu.
\end{aligned}
\end{equation}
Furthermore, evaluating Equation \eqref{hatUE4E6} at the six vertices $k_\lambda$, we obtain
\[ w(k_\lambda) = \widehat{U}(k_\lambda) \eps^{-1} \widehat{U}(k_\lambda)\Tr \]
which implies that
\begin{equation} \label{Pfw=detU}
\Pf w(k_\lambda) = \det \widehat{U}(k_\lambda) \Pf \eps^{-1}
\end{equation}
by the well known property $\Pf(C A C\Tr) = \det(C) \Pf(A)$, for a skew-symmetric matrix $A$ and a matrix $C \in M_m(\C)$ \cite{Mackey}.

Substituting the latter equality in the right-hand side of Equation \eqref{detU/detU} allows us to rewrite \eqref{quasiDelta} as
\begin{multline*}
\frac{1}{2 \pi \iu} \left( \int_{0}^{1/2} \di k_2 \, \tr ( w(k_*,k_2)^* \partial_{k_2} w(k_*,k_2) ) - 2 \log \frac{\Pf w(k_*,1/2)}{\Pf w(k_*,0)} \right) \equiv \\ \equiv \frac{s_*}{2 \pi \iu} \int_{S^1_*} \di z \, \tr \left( \widehat{U}(z)^* \partial_{z} \widehat{U}(z) \right) \bmod 2.
\end{multline*}
On the left-hand side of this equality, we recognize $P_\theta(k_*)$, as defined in \eqref{Ptheta}. Taking care of the orientation of $S^1_*$, we conclude, also in view of \eqref{E2+E5}, that
\begin{align*}
\Delta & \equiv P_\theta(1/2) - P_\theta(0) \bmod 2 \\
& \equiv \frac{1}{2 \pi \iu} \int_{E_1 + E_3 + E_4 + E_6} \di z \, \tr \left( \widehat{U}(z)^* \partial_{z} \widehat{U}(z) \right) \bmod 2 \\
& \equiv \deg([\det\widehat{U}]) \equiv \delta \bmod 2.
\end{align*}
This concludes the proof.
\end{proof}

\begin{rmk}[On the r\^ole of Pfaffians]
Equation \eqref{hatUE4E6} is the crucial point in the above argument. Indeed, since $\det \eps = 1$ (compare Remark \ref{rmk:eps}), from \eqref{hatUE4E6} it follows that
\begin{equation} \label{w=U(k)U(-k)}
\det w(k_*,k_2) = \det \widehat{U}(k_*,-k_2) \det \widehat{U}(k_*,k_2), \quad (k_*, k_2) \in S^1_*.
\end{equation}
If we evaluate this equality at the six vertices $k_\lambda$, we obtain
\begin{equation} \label{detw=detU2} 
\det w(k_\lambda) = \left( \det \widehat{U}(k_\lambda) \right)^2 = \left( \Pf w(k_\lambda) \right)^2. 
\end{equation}
Looking at Equation \eqref{w=U(k)U(-k)}, we realize that the expression $\det \widehat{U}(k_*,k_2)$ serves as a ``continuous prolongation'' along the edges $E_1, E_3, E_4, E_6$ of the Pfaffian $\Pf w(k_\lambda)$, which is well-defined only at the six vertices $k_\lambda$ (where the matrix $w$ is skew-symmetric), in a way which is moreover compatible with time-reversal symmetry, since in \eqref{w=U(k)U(-k)} both $\det \widehat{U}(k_*,k_2)$ and $\det \widehat{U}(k_*,-k_2)$ appear. This justifies the rather misterious and apparently {\it ad hoc} presence of Pfaffians in the Fu-Kane formula for the $\Z_2$ index $\Delta$.
\end{rmk}

In view of the equality $\Delta = \delta \in \Z_2$, we have a clear intepretation of the Fu-Kane index as the obstruction to the existence of a global continuous symmetric Bloch frame, as claimed in \cite[App. A]{FuKa}.

\newpage

%%%%% SECTION 5

\section{\texorpdfstring{A simpler formula for the $\Z_2$ invariant}{A simpler formula for the Z_2-invariant}} \label{sec:EvalVertices}

The fact that the $\Z_2$ invariant $\delta$, as defined in \eqref{delta}, is well-defined and independent of the choice of an interpolation of the vertex unitaries $U(k_\lambda)$ (see Section \ref{app:EdgeInvariance}) shows that its value depends only on the value of $\widehat{U}$ at the vertices. In this Section, we will provide a way to compute $\delta \in \Z_2$ using data coming just from the four \emph{time-reversal invariant momenta} $v_1, \ldots, v_4$ which are inequivalent modulo translational and inversion symmetries of $\R^2$. This should be compared with \cite[Equation (3.26)]{FuKa}.

In the previous Section, we have rewritten our invariant as $\delta = \widehat{P}_\theta(1/2) - \widehat{P}_\theta(0) \bmod 2$, with
\begin{align*}
- \widehat{P}_\theta(k_*) := \frac{1}{2 \pi \iu} & \left(  \int_{-1/2}^{0} \di k_2 \, \tr \left( w(k_*,k_2)^* \partial_{k_2} w(k_*,k_2) \right) + \right. \\
& - \left. 2 \int_{-1/2}^{0} \di k_2 \, \tr \left( \widehat{U}(k_*,k_2)^* \partial_{k_2} \widehat{U}(k_*,k_2)\right) \right) = \\
= \frac{1}{2 \pi \iu} & \left( \log \frac{\det w(k_*, 0)}{\det w(k_*, -1/2)} - 2 \log \frac{\det \widehat{U}(k_*, 0)}{\det \widehat{U}(k_*, -1/2)} \right).
\end{align*}
In view of \eqref{detw=detU2}, the above expression can be rewritten as
\[ \delta = \sum_{i=1}^{4} \widehat{\eta}_{v_i} \bmod 2, \quad \text{where} \quad \widehat{\eta}_{v_i} := \frac{1}{2 \pi \iu} \left( \log \left(\det \widehat{U}(v_i)\right)^2 - 2 \log \det \widehat{U}(v_i) \right). \]
Notice that, since $\deg([\det \widehat{U}])$ is an integer, the value $(-1)^{\deg([\det \widehat{U}])}$ is independent of the choice of the determination of $\log(-1)$ and is determined by the parity of $\deg([\det \widehat{U}])$, \ie by $\delta = \deg([\det \widehat{U}]) \bmod 2$. Moreover, this implies also that
\[ (-1)^\delta = \prod_{i=1}^{4} (-1)^{\widehat{\eta}_{v_i}} \]
and each $(-1)^{\widehat{\eta}_{v_i}}$ can be computed with any determination of $\log(-1)$ (as long as one chooses the same for all $i \in \set{1,\ldots, 4}$). Noticing that, using the principal value of the complex logarithm,
\begin{equation} \label{log-1}
(-1)^{(\log \alpha)/(2\pi \iu)} = (\eu^{\iu \pi})^{(\log \alpha)/(2\pi \iu)} = \eu^{(\log \alpha)/2} = \sqrt{\alpha},
\end{equation} 
it follows that we can rewrite the above expression for $(-1)^\delta$ as
\begin{equation} \label{easydelta}
(-1)^{\delta} = \prod_{i=1}^{4} \frac{\sqrt{\left(\det \widehat{U}(v_i)\right)^2}}{\det \widehat{U}(v_i)}
\end{equation}
where the branch of the square root is chosen in order to evolve continuously from $v_1$ to $v_2$ along $E_1$, and from $v_3$ to $v_4$ along $E_3$. This formula is to be compared with \cite[Equation (3.26)]{FuKa} for the Fu-Kane index $\Delta$, namely
\begin{equation} \label{DeltaFromPoints}
(-1)^{\Delta} = \prod_{i=1}^{4} \frac{\sqrt{\det w(v_i)}}{\Pf w(v_i)}.
\end{equation}

Recall now that the value of $\widehat{U}$ at the vertices $k_\lambda$ is determined by solving the vertex conditions, as in Section \ref{sec:VertexConditions}: $\widehat{U}(k_\lambda) = U(k_\lambda)$ is related to the obstruction unitary $U\sub{obs}(k_\lambda)$ by the relation \eqref{Uobs-U}. In particular, from \eqref{Uobs-U} we deduce that
\[ \det U\sub{obs}(k_\lambda) = \left( \det U(k_\lambda) \right)^2 \]
so that \eqref{easydelta} may be rewritten as
\begin{equation} \label{easierdelta}
(-1)^{\delta} = \prod_{i=1}^{4} \frac{\sqrt{\det U\sub{obs}(v_i)}}{\det U(v_i)}.
\end{equation}

This reformulation shows that our $\Z_2$ invariant can be computed starting just from the ``input'' Bloch frame $\Psi$ (provided it is \emph{continuous} on $\Bred$), and more specifically from its values at the vertices $k_\lambda$. Indeed, the obstruction $U\sub{obs}(k_\lambda)$ at the vertices is defined by \eqref{UobsVertices} solely in terms of $\Psi(k_\lambda)$; moreover, $U(k_\lambda)$ is determined by $U\sub{obs}(k_\lambda)$ as explained in the proof of Lemma \ref{V->U}.

We have thus the following \textbf{algorithmic recipe} to compute $\delta$:
\begin{itemize}
\item given a continuous Bloch frame $\Psi(k) = \set{\psi_a(k)}_{a=1,\ldots,m}$, $k \in \Bred$, compute the unitary matrix
\[ U\sub{obs}(k_\lambda)_{a,b} = \sum_{c=1}^{m} \scal{\psi_a(k_\lambda)}{\tau(\lambda) \Theta \psi_c(k_\lambda)} \eps_{cb}, \]
defined as in \eqref{UobsVertices}, at the four inequivalent time-reversal invariant momenta $k_\lambda = v_1, \ldots, v_4$;
\item compute the spectrum $\set{ \eu^{\iu \lambda_1^{(i)}} , \ldots, \eu^{\iu \lambda_m^{(i)}} } \subset U(1)$ of $U\sub{obs}(v_i)$, so that in particular 
\[ \det U\sub{obs}(v_i) = \exp \left( \iu \left( \lambda_1^{(i)} + \cdots + \lambda_m^{(i)} \right) \right); \]
normalize the arguments of such phases so that $\lambda_j^{(i)} \in [0, 2 \pi)$ for all $j \in \set{1, \ldots, m}$;
\item compute $U(v_i)$ as in the proof of Lemma \ref{V->U}: in particular we obtain that
\[ \det U(v_i) = \exp \left( \iu \left( \frac{ \lambda_1^{(i)} }{2} + \cdots + \frac{ \lambda_m^{(i)} }{2} \right) \right); \]
\item finally, compute $\delta$ from the formula \eqref{easierdelta}, \ie
\[ (-1)^{\delta} = \prod_{i=1}^{4} \dfrac{ \sqrt{\exp \left( \iu \left( \lambda_1^{(i)} + \cdots + \lambda_m^{(i)} \right) \right)} }{\exp \left( \iu \left( \dfrac{ \lambda_1^{(i)} }{2} + \cdots + \dfrac{ \lambda_m^{(i)} }{2} \right) \right)} . \]
\end{itemize}

\newpage 

%%%%% SECTION 6

\section{\texorpdfstring{Construction of a symmetric Bloch frame in $3d$}{Construction of a symmetric Bloch frame in 3d}} \label{sec:3d}

In this Section, we investigate the existence of a global continuous symmetric Bloch frame in the $3$-dimensional setting. In particular, we will recover the four $\Z_2$ indices proposed by Fu, Kane and Mele \cite{FuKaneMele}. We will first focus on the topological obstructions that arise, and then provide the construction of a symmetric Bloch frame in $3d$ when the procedure is unobstructed.

\subsection{Vertex conditions and edge extension}

The $3$-dimensional unit cell is defined, in complete analogy with the $2$-dimensional case, as
\[ \B^{(3)} := \set{k = \sum_{j=1}^{3} k_j e_j \in \R^3: -\frac{1}{2} \le k_j \le \frac{1}{2}, \: j = 1, 2, 3}  \]
(the superscript $^{(3)}$ stands for ``$3$-dimensional''), where $\set{e_1, e_2, e_3}$ is a basis in $\R^3$ such that $\Lambda = \Span_\Z\set{e_1, e_2, e_3}$. If a continuous Bloch frame is given on $\B^{(3)}$, then it can be extended to a continuous $\tau$-equivariant Bloch frame on $\R^3$ by considering its $\tau$-translates, provided it satisfies the obvious compatibility conditions of $\tau$-periodicity on the faces of the unit cell (\ie on its boundary). Similarly, if one wants to study frames which are also time-reversal invariant, then one can restrict the attention to the \effective unit cell (see Figure \ref{fig:Bred3})
\[ \Bred^{(3)} := \set{k = (k_1, k_2, k_3) \in \B^{(3)}: k_1 \ge 0}.  \]

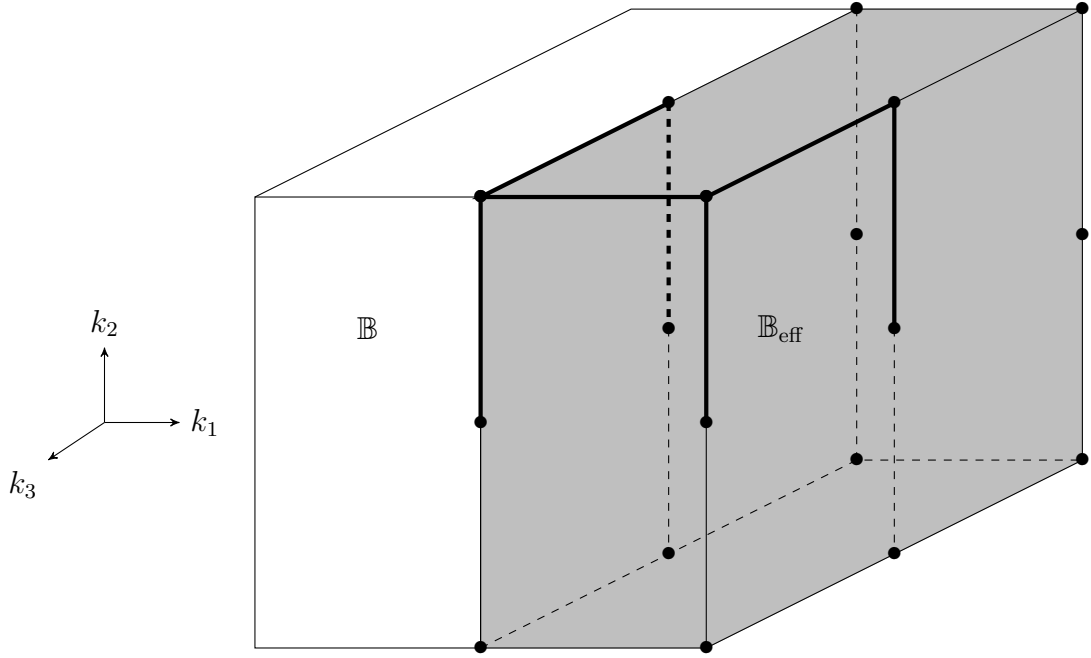
\begin{figure}[ht]
\centerline{%
\scalebox{1}{%
\begin{tikzpicture}[>=stealth']
\filldraw [lightgray] (0,3) -- (0,-3) -- (3,-3) -- (8,-0.5) -- (8,5.5) -- (5,5.5) -- (0,3);
%% USELESS HALF
\draw (0,3) -- (0,-3) -- (-3,-3) -- (-3,3) -- (0,3);
\draw (-3,3) -- (2,5.5) -- (5,5.5);
%% EFFECTIVE UNIT CELL
\draw (0,3) node {$\bullet$}
	   -- (0,-3) node {$\bullet$}
	   -- (3,-3) node {$\bullet$}
	   -- (3,3) node {$\bullet$}
	   -- (0,3) node {$\bullet$}
	   -- (5,5.5) node {$\bullet$}
	   -- (8,5.5) node {$\bullet$}
	   -- (3,3) node {$\bullet$};
\draw (3,-3)
	   -- (8,-0.5) node {$\bullet$}
	   -- (8,5.5);
\draw [dashed] (0,-3)
				 -- (5,-0.5) node {$\bullet$}
				 -- (8,-0.5);
\draw [dashed] (5,-0.5) -- (5,5.5);
\draw [ultra thick] (5.5,4.25) node {$\bullet$} 
					  -- (5.5,1.25) node {$\bullet$};
\draw [dashed] (5.5,1.25) -- (5.5,-1.75) node {$\bullet$};
\draw [dashed, ultra thick] (2.5,4.25) node {$\bullet$} 
						 -- (2.5,1.25) node {$\bullet$};
\draw [dashed] (2.5,1.25) -- (2.5,-1.75) node {$\bullet$};
\draw [ultra thick] (2.5,4.25) -- (0,3) -- (3,3) -- (5.5,4.25);
\draw [ultra thick] (0,3) -- (0,0) node {$\bullet$};
\draw [ultra thick] (3,3) -- (3,0) node {$\bullet$};
\draw (8,2.5) node {$\bullet$}
	  (5,2.5) node {$\bullet$}
	  (4,1.25) node {$\Bred$}
	   (-1.5,1.25) node {$\B$};
%% REFERENCE SYSTEM
\draw [->] (-5,0) -- (-4,0) node [anchor=west] {$k_1$};
\draw [->] (-5,0) -- (-5.75,-.5) node [anchor=north east] {$k_3$};
\draw [->] (-5,0) -- (-5,1) node [anchor=south] {$k_2$};
\end{tikzpicture}
} %end scalebox
} %end centerline
\caption{The $3$-dimensional \effective unit cell (shaded area).}
\label{fig:Bred3}
\end{figure}

\emph{Vertices} of the (\effective\!\!) unit cell are defined again as those points $k_\lambda$ which are invariant under the transformation $t_\lambda \theta$: these are the points with half-integer coordinates with respect to the lattice generators $\set{e_1, e_2, e_3}$. The \emph{vertex conditions} for a symmetric Bloch frame read exactly as in \eqref{VertexCondition}, and can be solved analogously to the $2$-dimensional case with the use of Lemma \ref{V->U}. This leads to the definition of a unitary matrix $U(k_\lambda)$ at each vertex. The definition of $U(\cdot)$ can be extended to the \emph{edges} joining vertices of $\Bred^{(3)}$ by using the path-connectedness of the group $\U(\C^m)$, as in Section \ref{sec:EdgeExtension}. Actually, we need to choose such an extension only on the edges which are plotted with a thick line in Figure \ref{fig:Bred3}: then we can obtain the definition of $U(\cdot)$ to all edges by imposing \emph{edge symmetries}, which are completely analogous to those in \eqref{EdgeSymmetries}.

\subsection{Extension to the faces: four $\Z_2$ obstructions} \label{sec:3dObs}

We now want to extend the definition of the matrix $U(k)$, mapping an input Bloch frame $\Psi(k)$ to a symmetric Bloch frame $\Phi(k)$, to $k \in \partial \Bred^{(3)}$. The boundary of $\Bred^{(3)}$ consists now of six faces, which we denote as follows:
\begin{align*}
F_{1, 0} & := \set{(k_1, k_2, k_3) \in \partial \Bred^{(3)} : k_1 = 0}, \\
F_{1, +} & := \set{(k_1, k_2, k_3) \in \partial \Bred^{(3)} : k_1 = \frac{1}{2}}, \\
F_{i, \pm} & := \set{(k_1, k_2, k_3) \in \partial \Bred^{(3)} : k_i = \pm \frac{1}{2}}, \quad i \in \set{2, 3}.
\end{align*}

Notice that $F_{1,0}$ and $F_{1,+}$ are both isomorphic to a \emph{full} unit cell $\B^{(2)}$ (the superscript $^{(2)}$ stands for ``$2$-dimensional''), while the faces $F_{i,\pm}$, $i \in \set{2, 3}$, are all isomorphic to a $2d$ \emph{\effective}\!\! unit cell $\Bred^{(2)}$. 

We start by considering the four faces $F_{i,\pm}$, $i \in \set{2, 3}$. From the $2$-dimensional algorithm, we know that to extend the definition of $U(\cdot)$ from the edges to one of these faces we need for the associated $\Z_2$ invariant $\delta_{i, \pm}$ to vanish. Not all these four invariants are independent, though: indeed, we have that $\delta_{i,+} = \delta_{i,-}$ for $i=2,3$. In fact, suppose for example that $\delta_{i,-}$ vanishes, so that we can extend the Bloch frame $\Phi$ to the face $F_{i,-}$. Then, by setting $\Phi(t_{e_i}(k)) := \tau_{e_i} \Phi(k)$ for $k \in F_{i,-}$, we get an extension of the frame $\Phi$ to $F_{i,+}$: it follows that also $\delta_{i,+}$ must vanish. Viceversa, exchanging the roles of the subscripts $+$ and $-$ one can argue that if $\delta_{i,+} =0$ then also $\delta_{i,-}=0$; in conclusion, $\delta_{i,+} = \delta_{i,-} \in \Z_2$, as claimed. We remain for now with only two independent $\Z_2$ invariants, namely $\delta_{2,+}$ and $\delta_{3,+}$.

\medskip

We now turn our attention to the faces $F_{1,0}$ and $F_{1,+}$. As we already noticed, these are $2d$ unit cells $\B^{(2)}$. Each of them contains three thick-line edges (as in Figure \ref{fig:Bred3}), on which we have already defined the Bloch frame $\Phi$: this allows us to test the possibility to extend it to the \effective unit cell which they enclose. If we are indeed able to construct such an extension to this \effective unit cell $\Bred^{(2)}$, then we can extend it to the whole face by using the time reversal operator $\Theta$ in the usual way: in fact, both $F_{1,0}$ and $F_{1,+}$ are such that if $k$ lies on it then also $-k$ lies on it (up to periodicity, or equivalently translational invariance). The obstruction to the extension of the Bloch frame to $F_{1,0}$ and $F_{1,+}$ is thus encoded again in two $\Z_2$ invariants $\delta_{1,0}$ and $\delta_{1,+}$; together with $\delta_{2,+}$ and $\delta_{3,+}$, they represent the obstruction to the continuous extension of the symmetric Bloch frame from the edges of $\Bred^{(3)}$ to the boundary $\partial \Bred^{(3)}$.

\medskip

Suppose now that we are indeed able to obtain such an extension to $\partial \Bred^{(3)}$, \ie that all four $\Z_2$ invariants vanish. Then we have a map $\widehat{U} \colon \partial \Bred^{(3)} \to \U(\C^m)$, such that at each $k \in \partial \Bred^{(3)}$ the unitary matrix $\widehat{U}(k)$ maps the reference Bloch frame $\Psi(k)$ to a symmetric Bloch frame $\widehat{\Phi}(k)$: $\widehat{\Phi}(k) = \Psi(k) \act \widehat{U}(k)$. In order to get a continuous symmetric Bloch frame defined on the whole $\Bred^{(3)}$, it is thus sufficient to extend the map $\widehat{U} \colon \partial \Bred^{(3)} \to \U(\C^m)$ to a continuous map $U \colon \Bred^{(3)} \to \U(\C^m)$. 

Topologically, the boundary of the \effective unit cell in $3$ dimensions is equivalent to a sphere: $\partial \Bred^{(3)} \simeq S^2$. Moreover, it is known that, if $X$ is any topological space, then a continuous map $f \colon S^2 \to X$ extends to a continuous map $F \colon D^3 \to X$, defined on the $3$-ball $D^3$ that the sphere encircles, if and only if its homotopy class $[f] \in [S^2; X] = \pi_2(X)$ is trivial. In our case, where $X = \U(\C^m)$, it is a well-known fact that $\pi_2(\U(\C^m)) = 0$ \cite[Ch.8, Sec. 12]{Husemoller}, so that actually \emph{any} continuous map from the boundary of the \effective unit cell to the unitary group extends continuously%
\footnote{Again, such a continuous extension can be explicitly provided, as shown in \cite[Remark 4.8]{FiPaPi}.} %
to a map defined on the whole $\Bred^{(3)}$. In other words, there is no topological obstruction to the continuous extension of a symmetric frame from $\partial \Bred^{(3)}$ to $\Bred^{(3)}$, and hence to $\R^3$: indeed, one can argue, as in Proposition \ref{global}, that a symmetric Bloch frame on $\Bred^{(3)}$ can always be extended to all of $\R^3$, by imposing $\tau$-equivariance and time-reversal symmetry.

In conclusion, we have that all the topological obstruction that can prevent the existence of a global continuous symmetric Bloch frame is encoded in the four $\Z_2$ invariants $\delta_{1,0}$ and $\delta_{i,+}$, for $i \in \set{1, 2, 3}$, given by
\begin{equation} \label{3Ddelta} 
\delta_{*}(P) := \delta\left( P \big|_{F_*} \right), \quad F_* \in \set{F_{1,0}, F_{1,+}, F_{2,+}, F_{3,+}}. 
\end{equation}

\subsection{\texorpdfstring{Proof of Theorem \ref{thm:Q3}}{Proof of Theorem 3}}

We have now understood how topological obstruction may arise in the construction of a global continuous symmetric Bloch frame, by sketching the $3$-dimensional analogue of the method developed in Section \ref{sec:Algorithm} for the $2d$ case. This Subsection is devoted to detailing a more precise constructive algorithm for such a Bloch frame in $3d$, whenever there is no topological obstruction to its existence. As in the $2d$ case, we will start from a ``na\"{i}ve'' choice of a continuous Bloch frame $\Psi(k)$, and symmetrize it to obtain the required symmetric Bloch frame $\Phi(k) = \Psi(k) \act U(k)$ on the \effective unit cell $\Bred^{(3)}$ (compare \eqref{Healer}). We will apply this scheme first on vertices, then on edges, then (whenever there is no topological obstruction) to faces, and then to the interior of $\Bred^{(3)}$.

The two main ``tools'' in this algorithm are provided by the following Lemmas, which are ``distilled'' from the procedure elaborated in Section \ref{sec:Algorithm}.

\begin{lemma}[From $0d$ to $1d$] \label{0d->1d}
Let $\set{P(k)}_{k \in \R}$ be a family of projectors satisfying Assumption \ref{proj}. Denote $\Bred^{(1)} = \set{k = k_1 e_1 : -1/2 \le k_1 \le 0} \simeq [-1/2,0]$. Let $\Phi(-1/2)$ be a frame in $\Ran P(-1/2)$ such that
\[ \Phi(-1/2) = \tau_{-e_1} \Theta \Phi(-1/2) \act \eps, \]
\ie $\Phi$ satisfies the vertex condition \eqref{VertexCondition} at $k_{-e_1} = -(1/2) e_1 \in \Bred^{(1)}$. Then one constructs a continuous Bloch frame $\set{\eff{\Phi}(k)}_{k \in \Bred^{(1)}}$ such that
\[ \eff{\Phi}(-1/2) = \Phi(-1/2) \quad \text{and} \quad \eff{\Phi}(0) = \Theta \eff{\Phi}(0) \act \eps. \]
\end{lemma}
\begin{proof}
Solve the vertex condition at $k_\lambda = 0 \in \Bred^{(1)}$ using Lemma \ref{V->U}, and then extend the frame to $\Bred^{(1)}$ as in Section \ref{sec:EdgeExtension}.
\end{proof}

\begin{lemma}[From $1d$ to $2d$] \label{1d->2d}
Let $\set{P(k)}_{k \in \R^2}$ be a family of projectors satisfying Assumption \ref{proj}. Let $v_i$ and $E_i$ denote the edges of the \effective unit cell $\Bred^{(2)}$, as defined in Section \ref{sec:Bred}. Let $\set{\Phi(k)}_{k \in E_1 \cup E_2 \cup E_3}$ be a continuous Bloch frame for $\set{P(k)}_{k \in \R^2}$, satisfying the vertex conditions \eqref{VertexCondition} at $v_1, v_2, v_3$ and $v_4$. Assume that
\[ \delta(P) = 0 \in \Z_2. \]
Then one constructs a continuous Bloch frame $\set{\eff{\Phi}(k)}_{k \in \Bred^{(2)}}$, satisfying the vertex conditions \eqref{VertexCondition} and the edge symmetries \eqref{EdgeSymmetries}, and moreover such that
\begin{equation} \label{Don'tMove}
\eff{\Phi}(k) = \Phi(k) \quad \text{for } k \in E_1 \cup E_2.
\end{equation}
\end{lemma}
\begin{proof}
This is just a different formulation of what was proved in Section \ref{sec:FaceExtension}. From the datum of $\Phi(k)$ on the edges $E_1 \cup E_2 \cup E_3$ one can construct the unitary matrix $\widehat{U} \colon \partial \Bred^{(2)} \to \U(\C^m)$, such that $\Phi(k) \act \widehat{U}(k)$ is a symmetric Bloch frame, for $k \in \partial \Bred^{(2)}$. The hypothesis $\delta(P) = 0 \in \Z_2$ ensures that $\deg([\det \widehat{U}]) = 2s$, $s \in \Z$. The proof of Theorem \ref{thm:Q2'} then gives an explicit family of unitary matrices $X \colon \partial \Bred^{(2)} \to \U(\C^m)$ such that $\deg([\det X]) = -2s$ and that $\eff{\Phi}(k) := \Phi(k) \act \left( \widehat{U}(k) \, X(k) \right)$ is still symmetric, so it extends to the whole $\Bred^{(2)}$. The fact that $X(k)$ can be chosen to be constantly equal to the identity matrix on $E_1 \cup E_2 \cup E_5 \cup E_6$ shows that \eqref{Don'tMove} holds.
\end{proof}

We are now ready to prove Theorem \ref{thm:Q3}.
\begin{proof}[Proof of Theorem \ref{thm:Q3}]
Consider the region $Q = \Bred^{(3)} \cap \set{k_2 \le 0, k_3 \ge 0}$. We label its vertices as $v_1=(0,-1/2,0), v_2, v_3, v_4 \in Q \cap \set{k_2 = -1/2}$ (counted counter-clockwise) and $v_5=(0,0,0), v_6, v_7, v_8 \in Q \cap \set{k_2=0}$ (again counted counter-clockwise). We also let $E_{ij}$ denote the edge joining $v_i$ and $v_j$. The labels are depicted in Figure \ref{fig:StepsQ3}.

\begin{figure}
\centering
\subfigure[Step 1]{\scalebox{.5}{%
\begin{tikzpicture}[>=stealth']
\draw [dotted] (0,3)% node {$\bullet$}
	   -- (0,-3) node {$\bullet$} node [anchor = north east] {\LARGE $v_2$}
	   -- (3,-3) node {$\bullet$} node [anchor = north] {\LARGE $v_3$}
	   -- (3,3)% node {$\bullet$}
	   -- (0,3)% node {$\bullet$}
	   -- (5,5.5)% node {$\bullet$}
	   -- (8,5.5)% node {$\bullet$}
	   -- (3,3);% node {$\bullet$};
\draw [dotted] (3,-3)
	   -- (8,-0.5)% node {$\bullet$}
	   -- (8,5.5);
\draw [dotted] (0,-3)
				 -- (5,-0.5)% node {$\bullet$}
				 -- (8,-0.5);
\draw [dotted] (5,-0.5) -- (5,5.5);
\draw [dotted] (5.5,4.25)% node {$\bullet$} 
			     -- (5.5,1.25) node {$\bullet$} node [anchor = west] {\LARGE $v_8$};
\draw [dotted] (5.5,1.25) -- (5.5,-1.75) node {$\bullet$} node [anchor = west] {\LARGE $v_4$};
\draw [dotted] (2.5,4.25)% node {$\bullet$} 
				 -- (2.5,1.25) node {$\bullet$} node [anchor = east] {\LARGE $v_5$};
\draw [dotted] (2.5,1.25) -- (2.5,-1.75) node {$\bullet$} node [anchor = east] {\LARGE $v_1$};
\draw  (0,0) node {$\bullet$} node [anchor = east] {\LARGE $v_6$}
		(3,0) node {$\bullet$} node [anchor = west] {\LARGE $v_7$};
\end{tikzpicture}
}}\quad%end subfigure 1
\subfigure[Step 2]{\scalebox{.5}{%
\begin{tikzpicture}[>=stealth']
\draw (0,0) -- (0,-3) -- (3,-3) -- (3,0);
\draw (3,-3) -- (5.5,-1.75) -- (5.5,1.25);
\draw [dashed] (0,-3) -- (2.5,-1.75) -- (2.5,1.25);
\draw [dotted] (0,3)% node {$\bullet$}
	   -- (0,-3) node {$\bullet$}
	   -- (3,-3) node {$\bullet$}
	   -- (3,3)% node {$\bullet$}
	   -- (0,3)% node {$\bullet$}
	   -- (5,5.5)% node {$\bullet$}
	   -- (8,5.5)% node {$\bullet$}
	   -- (3,3);% node {$\bullet$};
\draw [dotted] (3,-3)
	   -- (8,-0.5)% node {$\bullet$}
	   -- (8,5.5);
\draw [dotted] (0,-3)
				 -- (5,-0.5)% node {$\bullet$}
				 -- (8,-0.5);
\draw [dotted] (5,-0.5) -- (5,5.5);
\draw [dotted] (5.5,4.25)% node {$\bullet$} 
			     -- (5.5,1.25) node {$\bullet$};
\draw [dotted] (5.5,1.25) -- (5.5,-1.75) node {$\bullet$};
\draw [dotted] (2.5,4.25)% node {$\bullet$} 
				 -- (2.5,1.25) node {$\bullet$};
\draw [dotted] (2.5,1.25) -- (2.5,-1.75) node {$\bullet$};
\draw  (0,0) node {$\bullet$}
		(3,0) node {$\bullet$}
		(0,-1.5) node [anchor = east] {\LARGE $E_{26}$}
		(1.5,-3) node [anchor = north] {\LARGE $E_{23}$}
		(4.1,-2.3) node [anchor = north west] {\LARGE $E_{34}$}
		(5.5,0) node [anchor = west] {\LARGE $E_{48}$}
		(3.1,-1.5) node [anchor = west] {\LARGE $E_{37}$}
		(1.75,-1.9) node [anchor = east] {\LARGE $E_{12}$}
		(2.5,0) node [anchor = east] {\LARGE $E_{15}$};
\end{tikzpicture}
}}\quad%end subfigure 2
\subfigure[Step 3a]{\scalebox{.5}{%%
\begin{tikzpicture}[>=stealth']
\filldraw [lightgray] (0,3) -- (0,-3) -- (3,-3) -- (3,3) -- (0,3);
\draw (0,0) -- (0,-3) -- (3,-3) -- (3,0);
\draw (3,-3) -- (5.5,-1.75) -- (5.5,1.25);
\draw [dashed] (0,-3) -- (2.5,-1.75) -- (2.5,1.25);
\draw (0,3) node {$\bullet$}
	   -- (0,-3) node {$\bullet$}
	   -- (3,-3) node {$\bullet$}
	   -- (3,3) node {$\bullet$}
	   -- (0,3);
\draw [dotted] (0,3)% node {$\bullet$}
	   -- (5,5.5)% node {$\bullet$}
	   -- (8,5.5)% node {$\bullet$}
	   -- (3,3);% node {$\bullet$};
\draw [dotted] (3,-3)
	   -- (8,-0.5)% node {$\bullet$}
	   -- (8,5.5);
\draw [dotted] (0,-3)
				 -- (5,-0.5)% node {$\bullet$}
				 -- (8,-0.5);
\draw [dotted] (5,-0.5) -- (5,5.5);
\draw [dotted] (5.5,4.25)% node {$\bullet$} 
			     -- (5.5,1.25) node {$\bullet$};
\draw [dotted] (5.5,1.25) -- (5.5,-1.75) node {$\bullet$};
\draw [dotted] (2.5,4.25)% node {$\bullet$} 
				 -- (2.5,1.25) node {$\bullet$};
\draw [dotted] (2.5,1.25) -- (2.5,-1.75) node {$\bullet$};
\draw  (0,0) node {$\bullet$}
		(3,0) node {$\bullet$};
\end{tikzpicture}
}}\\%end subfigure 3a
\subfigure[Step 3b]{\scalebox{.5}{%%
\begin{tikzpicture}[>=stealth']
\filldraw [lightgray] (0,-3) -- (3,-3) -- (8,-0.5) -- (5,-0.5) -- (0,-3);
\filldraw [gray] (0,3) -- (0,-3) -- (3,-3) -- (3,3) -- (0,3);
\draw (0,0) -- (0,-3) -- (3,-3) -- (3,0);
\draw (3,-3) -- (5.5,-1.75) -- (5.5,1.25);
\draw [dashed] (0,-3) -- (2.5,-1.75) -- (2.5,1.25);
\draw (0,3) node {$\bullet$}
	   -- (0,-3) node {$\bullet$}
	   -- (3,-3) node {$\bullet$}
	   -- (3,3) node {$\bullet$}
	   -- (0,3);
\draw [dotted] (0,3)% node {$\bullet$}
	   -- (5,5.5)% node {$\bullet$}
	   -- (8,5.5)% node {$\bullet$}
	   -- (3,3);% node {$\bullet$};
\draw (3,-3) -- (8,-0.5) node {$\bullet$};
\draw [dotted] (8,-0.5) -- (8,5.5);
\draw [dashed] (0,-3)	 -- (5,-0.5) node {$\bullet$}
				 -- (8,-0.5) ;
\draw [dotted] (5,-0.5) -- (5,5.5);
\draw [dotted] (5.5,4.25)% node {$\bullet$} 
			     -- (5.5,1.25) node {$\bullet$};
\draw [dotted] (5.5,1.25) -- (5.5,-1.75) node {$\bullet$};
\draw [dotted] (2.5,4.25)% node {$\bullet$} 
				 -- (2.5,1.25) node {$\bullet$};
\draw [dotted] (2.5,1.25) -- (2.5,-1.75) node {$\bullet$};
\draw  (0,0) node {$\bullet$}
		(3,0) node {$\bullet$};
\end{tikzpicture}
}}\qquad%end subfigure 3b
\subfigure[Step 3c]{\scalebox{.5}{%%
\begin{tikzpicture}[>=stealth']
\filldraw [lightgray] (0,3) -- (2.5,4.25) -- (2.5,-1.75) -- (0,-3) -- (0,3);
\filldraw [gray] (0,-3) -- (3,-3) -- (8,-0.5) -- (5,-0.5) -- (0,-3);
\filldraw [gray] (0,3) -- (0,-3) -- (3,-3) -- (3,3) -- (0,3);
\draw (0,0) -- (0,-3) -- (3,-3) -- (3,0);
\draw (3,-3) -- (5.5,-1.75) -- (5.5,1.25);
\draw [dashed] (0,-3) -- (2.5,-1.75) -- (2.5,1.25);
\draw (0,3) node {$\bullet$}
	   -- (0,-3) node {$\bullet$}
	   -- (3,-3) node {$\bullet$}
	   -- (3,3) node {$\bullet$}
	   -- (0,3);
\draw [dotted] (0,3)% node {$\bullet$}
	   -- (5,5.5)% node {$\bullet$}
	   -- (8,5.5)% node {$\bullet$}
	   -- (3,3);% node {$\bullet$};
\draw (3,-3) -- (8,-0.5) node {$\bullet$};
\draw [dotted] (8,-0.5) -- (8,5.5);
\draw [dashed] (0,-3)	 -- (5,-0.5) node {$\bullet$}
				 -- (8,-0.5) ;
\draw [dotted] (5,-0.5) -- (5,5.5);
\draw [dotted] (5.5,4.25)% node {$\bullet$} 
			     -- (5.5,1.25) node {$\bullet$};
\draw [dotted] (5.5,1.25) -- (5.5,-1.75) node {$\bullet$};
\draw (0,3) -- (2.5,4.25);
\draw [dashed] (2.5,4.25) node {$\bullet$} 
				 -- (2.5,1.25) node {$\bullet$}
				 -- (2.5,-1.75) node {$\bullet$};
\draw  (0,0) node {$\bullet$}
		(3,0) node {$\bullet$};
\end{tikzpicture}
}}\qquad%end subfigure 3c
\subfigure[Step 3d]{\scalebox{.5}{%%
\begin{tikzpicture}[>=stealth']
\filldraw [gray] (0,3) -- (2.5,4.25) -- (2.5,-1.75) -- (0,-3) -- (0,3);
\filldraw [gray] (0,-3) -- (3,-3) -- (8,-0.5) -- (5,-0.5) -- (0,-3);
\filldraw [gray] (0,3) -- (0,-3) -- (3,-3) -- (3,3) -- (0,3);
\filldraw [lightgray] (3,-3) -- (5.5,-1.75) -- (5.5,4.25) -- (3,3) -- (3,-3);
\draw (0,0) -- (0,-3) -- (3,-3) -- (3,0);
\draw (3,-3) -- (5.5,-1.75) -- (5.5,1.25);
\draw [dashed] (0,-3) -- (2.5,-1.75) -- (2.5,1.25);
\draw (0,3) node {$\bullet$}
	   -- (0,-3) node {$\bullet$}
	   -- (3,-3) node {$\bullet$}
	   -- (3,3) node {$\bullet$}
	   -- (0,3);
\draw [dotted] (0,3)% node {$\bullet$}
	   -- (5,5.5)% node {$\bullet$}
	   -- (8,5.5)% node {$\bullet$}
	   -- (3,3);% node {$\bullet$};
\draw (3,-3) -- (8,-0.5) node {$\bullet$};
\draw [dotted] (8,-0.5) -- (8,5.5);
\draw [dashed] (0,-3)	 -- (5,-0.5) node {$\bullet$}
				 -- (8,-0.5) ;
\draw [dotted] (5,-0.5) -- (5,5.5);
\draw (3,3) -- (5.5,4.25) node {$\bullet$} 
			     -- (5.5,1.25) node {$\bullet$};
\draw [dotted] (5.5,1.25) -- (5.5,-1.75) node {$\bullet$};
\draw (0,3) -- (2.5,4.25);
\draw [dashed] (2.5,4.25) node {$\bullet$} 
				 -- (2.5,1.25) node {$\bullet$}
				 -- (2.5,-1.75) node {$\bullet$};
\draw  (0,0) node {$\bullet$}
		(3,0) node {$\bullet$};
\end{tikzpicture}
}}\\%end subfigure 3d
\subfigure[Step 4]{\scalebox{.5}{%%
\begin{tikzpicture}[>=stealth']
\filldraw [lightgray] (2.5,4.25) -- (2.5,-1.75) -- (5,-0.5) -- (5,5.5) -- (2.5,4.25);
\filldraw [gray] (0,3) -- (2.5,4.25) -- (2.5,-1.75) -- (0,-3) -- (0,3);
\filldraw [gray] (0,-3) -- (3,-3) -- (8,-0.5) -- (5,-0.5) -- (0,-3);
\filldraw [gray] (0,3) -- (0,-3) -- (3,-3) -- (3,3) -- (0,3);
\filldraw [gray] (3,-3) -- (5.5,-1.75) -- (5.5,4.25) -- (3,3) -- (3,-3);
\filldraw [lightgray] (5.5,4.25) -- (5.5,-1.75) -- (8,-0.5) -- (8,5.5) -- (5.5,4.25);
\draw (0,0) -- (0,-3) -- (3,-3) -- (3,0);
\draw (3,-3) -- (5.5,-1.75) -- (5.5,1.25);
\draw [dashed] (0,-3) -- (2.5,-1.75) -- (2.5,1.25);
\draw (0,3) node {$\bullet$}
	   -- (0,-3) node {$\bullet$}
	   -- (3,-3) node {$\bullet$}
	   -- (3,3) node {$\bullet$}
	   -- (0,3);
\draw [dotted] %(0,3)% node {$\bullet$}
	  (5,5.5)% node {$\bullet$}
	   -- (8,5.5);% node {$\bullet$}
%	   -- (3,3);% node {$\bullet$};
\draw (2.5,4.25) -- (5,5.5) node {$\bullet$};
\draw (5.5,4.25) -- (8,5.5) node {$\bullet$};
\draw (3,-3) -- (8,-0.5) node {$\bullet$};
\draw (8,-0.5) -- (8,5.5);
\draw [dashed] (0,-3)	 -- (5,-0.5) node {$\bullet$}
				 -- (8,-0.5) ;
\draw [dashed] (5,-0.5) -- (5,5.5);
\draw (3,3) -- (5.5,4.25) node {$\bullet$} 
			     -- (5.5,1.25) node {$\bullet$};
\draw (5.5,4.25) -- (5.5,1.25) -- (5.5,-1.75) node {$\bullet$};
\draw (0,3) -- (2.5,4.25);
\draw [dashed] (2.5,4.25) node {$\bullet$} 
				 -- (2.5,1.25) node {$\bullet$}
				 -- (2.5,-1.75) node {$\bullet$};
\draw  (8,2.5) node {$\bullet$}
		(5,2.5) node {$\bullet$}
		(0,0) node {$\bullet$}
		(3,0) node {$\bullet$};
\end{tikzpicture}
}}\qquad%end subfigure 4
\subfigure[Step 5]{\scalebox{.5}{%%
\begin{tikzpicture}[>=stealth']
\filldraw [lightgray] (5,-0.5) -- (5,5.5) -- (8,5.5) -- (8,-0.5) -- (5,-0.5);
\filldraw [gray] (2.5,4.25) -- (2.5,-1.75) -- (5,-0.5) -- (5,5.5) -- (2.5,4.25);
\filldraw [gray] (0,3) -- (2.5,4.25) -- (2.5,-1.75) -- (0,-3) -- (0,3);
\filldraw [gray] (0,-3) -- (3,-3) -- (8,-0.5) -- (5,-0.5) -- (0,-3);
\filldraw [gray] (0,3) -- (0,-3) -- (3,-3) -- (3,3) -- (0,3);
\filldraw [gray] (3,-3) -- (5.5,-1.75) -- (5.5,4.25) -- (3,3) -- (3,-3);
\filldraw [gray] (5.5,4.25) -- (5.5,-1.75) -- (8,-0.5) -- (8,5.5) -- (5.5,4.25);
\draw (0,0) -- (0,-3) -- (3,-3) -- (3,0);
\draw (3,-3) -- (5.5,-1.75) -- (5.5,1.25);
\draw [dashed] (0,-3) -- (2.5,-1.75) -- (2.5,1.25);
\draw (0,3) node {$\bullet$}
	   -- (0,-3) node {$\bullet$}
	   -- (3,-3) node {$\bullet$}
	   -- (3,3) node {$\bullet$}
	   -- (0,3);
\draw %(0,3)% node {$\bullet$}
	  (5,5.5)% node {$\bullet$}
	   -- (8,5.5);% node {$\bullet$}
%	   -- (3,3);% node {$\bullet$};
\draw (2.5,4.25) -- (5,5.5) node {$\bullet$};
\draw (5.5,4.25) -- (8,5.5) node {$\bullet$};
\draw (3,-3) -- (8,-0.5) node {$\bullet$};
\draw (8,-0.5) -- (8,5.5);
\draw [dashed] (0,-3)	 -- (5,-0.5) node {$\bullet$}
				 -- (8,-0.5) ;
\draw [dashed] (5,-0.5) -- (5,5.5);
\draw (3,3) -- (5.5,4.25) node {$\bullet$} 
			     -- (5.5,1.25) node {$\bullet$};
\draw [dashed] (5.5,1.25) -- (5.5,-1.75) node {$\bullet$};
\draw (0,3) -- (2.5,4.25);
\draw [dashed] (2.5,4.25) node {$\bullet$} 
				 -- (2.5,1.25) node {$\bullet$}
				 -- (2.5,-1.75) node {$\bullet$};
\draw  (8,2.5) node {$\bullet$}
		(5,2.5) node {$\bullet$}
		(0,0) node {$\bullet$}
		(3,0) node {$\bullet$};
\end{tikzpicture}
}}\qquad%end subfigure 5
\subfigure[Step 6]{\scalebox{.5}{%%
\begin{tikzpicture}[>=stealth']
\filldraw [gray] (2.5,4.25) -- (2.5,-1.75) -- (5,-0.5) -- (5,5.5) -- (2.5,4.25);
\filldraw [gray] (0,3) -- (2.5,4.25) -- (2.5,-1.75) -- (0,-3) -- (0,3);
\filldraw [gray] (0,-3) -- (3,-3) -- (8,-0.5) -- (5,-0.5) -- (0,-3);
\filldraw [gray] (0,3) -- (0,-3) -- (3,-3) -- (3,3) -- (0,3);
\filldraw [gray] (3,-3) -- (5.5,-1.75) -- (5.5,4.25) -- (3,3) -- (3,-3);
\filldraw [gray] (5.5,4.25) -- (5.5,-1.75) -- (8,-0.5) -- (8,5.5) -- (5.5,4.25);
\filldraw [lightgray] (5,5.5) -- (8,5.5) -- (3,3) -- (0,3) -- (5,5.5);
\draw (0,0) -- (0,-3) -- (3,-3) -- (3,0);
\draw (3,-3) -- (5.5,-1.75) -- (5.5,1.25);
\draw [dashed] (0,-3) -- (2.5,-1.75) -- (2.5,1.25);
\draw (0,3) node {$\bullet$}
	   -- (0,-3) node {$\bullet$}
	   -- (3,-3) node {$\bullet$}
	   -- (3,3) node {$\bullet$}
	   -- (0,3);
\draw %(0,3)% node {$\bullet$}
	  (5,5.5)% node {$\bullet$}
	   -- (8,5.5);% node {$\bullet$}
%	   -- (3,3);% node {$\bullet$};
\draw (2.5,4.25) -- (5,5.5) node {$\bullet$};
\draw (5.5,4.25) -- (8,5.5) node {$\bullet$};
\draw (3,-3) -- (8,-0.5) node {$\bullet$};
\draw (8,-0.5) -- (8,5.5);
\draw [dashed] (0,-3)	 -- (5,-0.5) node {$\bullet$}
				 -- (8,-0.5) ;
\draw [dashed] (5,-0.5) -- (5,5.5);
\draw (3,3) -- (5.5,4.25) node {$\bullet$} 
			     -- (5.5,1.25) node {$\bullet$};
\draw [dashed] (5.5,1.25) -- (5.5,-1.75) node {$\bullet$};
\draw (0,3) -- (2.5,4.25);
\draw [dashed] (2.5,4.25) node {$\bullet$} 
				 -- (2.5,1.25) node {$\bullet$}
				 -- (2.5,-1.75) node {$\bullet$};
\draw  (8,2.5) node {$\bullet$}
		(5,2.5) node {$\bullet$}
		(0,0) node {$\bullet$}
		(3,0) node {$\bullet$};
\end{tikzpicture}
}}%end subfigure 6
\caption{Steps in the proof of Theorem \ref{thm:Q3}.}
\label{fig:StepsQ3}
\end{figure}
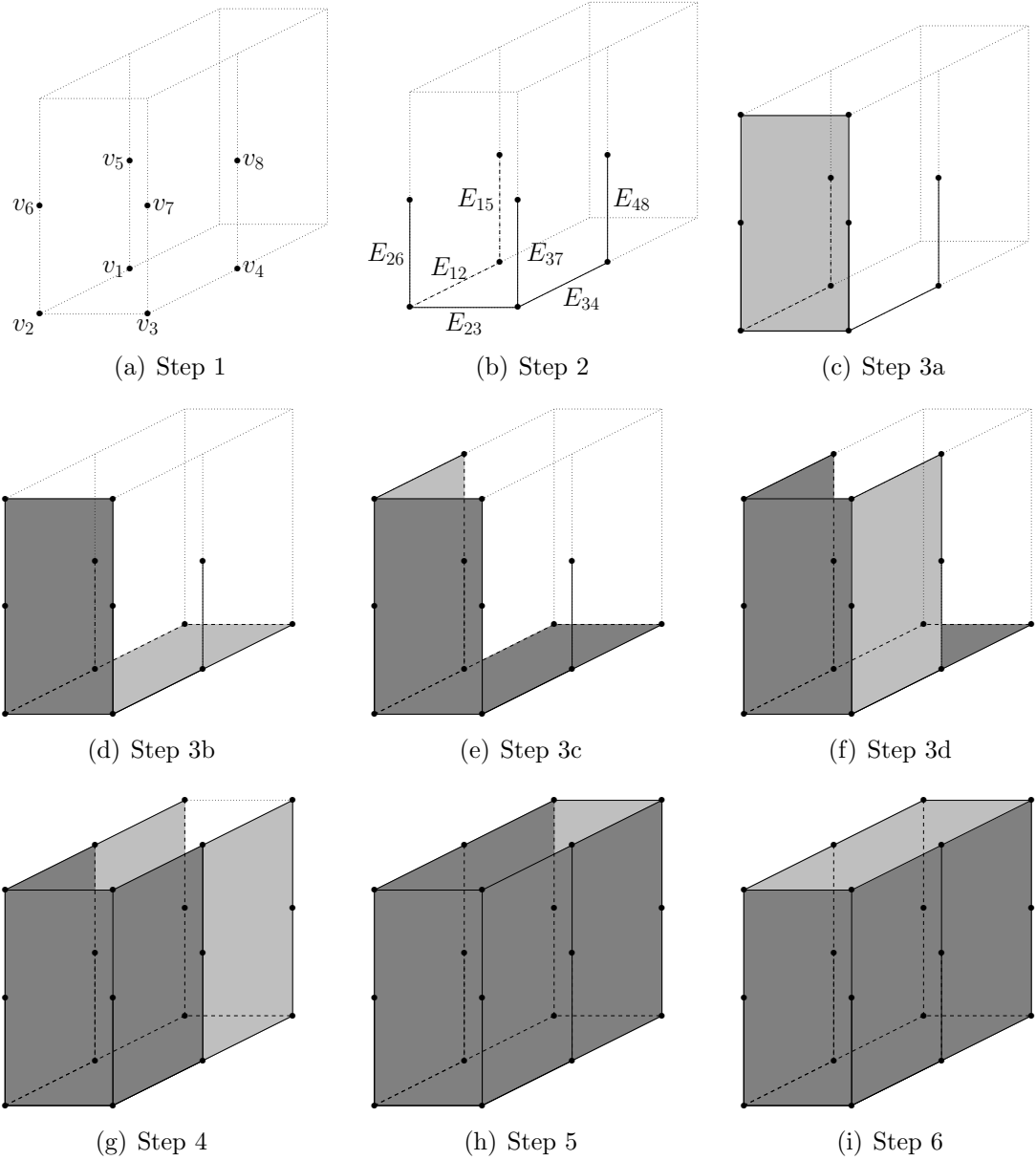

\noindent \textsl{Step 1.} Choose a continuous Bloch frame $\set{\Psi(k)}_{k \in \Bred^{(3)}}$ for $\set{P(k)}_{k \in \R^3}$. Solve the vertex conditions at all $v_i$'s, using Lemma \ref{V->U} to obtain frames $\Phi\sub{ver}(v_i)$, $i = 1, \ldots, 8$.

\noindent \textsl{Step 2.} Using repeatedly Lemma \ref{0d->1d}, extend the definition of $\Phi\sub{ver}$ to the edges $E_{12}$, $E_{15}$, $E_{23}$, $E_{26}$, $E_{34}$, $E_{37}$ and $E_{48}$, to obtain frames $\Phi\sub{edg}(k)$.

\noindent \textsl{Step 3.} Use Lemma \ref{1d->2d} to extend further the definition of $\Phi\sub{edg}$ as follows.
\begin{enumerate}[label=3\alph*.]
\item Extend the definition of $\Phi\sub{edg}$ on $F_{3,+}$: this can be done since $\delta_{3,+}$ vanishes. This will change the original definition of $\Phi\sub{edg}$ only on $E_{37}$.
\item Extend the definition of $\Phi\sub{edg}$ on $F_{2,-}$: this can be done since $\delta_{2,-} = \delta_{2,+}$ vanishes. This will change the original definition of $\Phi\sub{edg}$ only on $E_{34}$.
\item Extend the definition of $\Phi\sub{edg}$ on $F_{1,0} \cap \set{k_3 \ge 0}$: this can be done since $\delta_{1,0}$ vanishes. This will change the original definition of $\Phi\sub{edg}$ only on $E_{15}$. (Here we give $F_{1,0}$ the \emph{opposite} orientation to the one inherited from $\partial \Bred^{(3)}$, but this is inessential.)
\item Extend the definition of $\Phi\sub{edg}$ on $F_{1,+} \cap \set{k_3 \ge 0}$. Although we have already changed the definition of $\Phi\sub{edg}$ on $E_{37} \cup E_{34}$, this extension can still be performed since we have proved in Section \ref{app:EdgeInvariance} that $\delta_{1,+}$ is independent of the choice of the extension of the frame along edges, and hence vanishes by hypothesis. The extension will further modify the definition of $\Phi\sub{edg}$ only on $E_{48}$.
\end{enumerate}
We end up with a Bloch frame $\Phi_S$ defined on
\[ S := F_{2,-} \cup F_{3,+} \cup \left( F_{1,0} \cap \set{k_3 \ge 0} \right) \cup \left( F_{1,+} \cap \set{k_3 \ge 0} \right).\]

\noindent \textsl{Step 4.} Extend the definition of $\Phi_S$ to $F_{1,0} \cap \set{k_3 \le 0}$ and $F_{1,+} \cap \set{k_3 \le 0}$ by setting
\begin{align*}
\Phi_{S'}(k) & := \Phi_S(k) && \text{for } k \in S, \\
\Phi_{S'}(\theta(k)) & := \Theta \Phi_S(k) \act \eps && \text{for } k \in F_{1,0} \cap \set{k_3 \ge 0}, \quad \text{and} \\
\Phi_{S'}(t_{e_1}\theta(k)) & := \tau_{e_1} \Theta \Phi_S(k) \act \eps && \text{for } k \in F_{1,+} \cap \set{k_3 \ge 0}.
\end{align*}

This extension is continuous. Indeed, continuity along $E_{15} \cup \theta(E_{15}) = F_{1,0} \cap \set{k_3 = 0}$ (respectively $E_{48} \cup t_{e_1}\theta(E_{48}) = F_{1,+} \cap \set{k_3 = 0}$) is a consequence of the edge symmetries that we have imposed on $\Phi_S$ in Step 3c (respectively in Step 3d). What we have to verify is that the definition of $\Phi_{S'}$ that we have just given on $t_{-e_2}\theta(E_{12}) = (F_{1,0} \cap \set{k_3 \le 0}) \cap F_{2,-}$ and $t_{e_1-e_2}\theta(E_{34}) = (F_{1,+} \cap \set{k_3 \le 0}) \cap F_{2,-}$ agrees with the definition of $\Phi_S$ on the same edges that was achieved in Step 3b. 

We look at $t_{-e_2}\theta(E_{12})$, since the case of $t_{e_1-e_2}\theta(E_{34})$ is analogous. If $k \in E_{12}$, then $\Phi_S(t_{-e_2}\theta(k)) = \tau_{-e_2}\Theta \Phi_S(k) \act \eps$, since in Step 3b we imposed the edge symmetries on $F_{3,-}$. On the other hand, $\Phi_{S'}(k) = \Phi_S(k)$ on $E_{12} \subset S$ by definition of $\Phi_{S'}$: we get then also
\begin{align*}
\Phi_S(t_{-e_2}\theta(k)) & = \tau_{-e_2} \Theta \Phi_S(k) \act \eps = \tau_{-e_2}\Theta \Phi_{S'}(k) \act \eps = \\
& = \tau_{-e_2} \Theta \tau_{-e_2} \Phi_{S'}(t_{e_2}(k)) \act \eps
\end{align*}
because of the edge symmetries satisfied by $\Phi_{S'}$ on the shorter edges of $F_{1,0} \cap \set{k_3 \ge 0}$. In view of \eqref{item:TRtau} and of the definition of $\Phi_{S'}$ on $F_{1,0} \cap \set{k_3 \le 0}$, we conclude that
\[ \Phi_S(t_{-e_2}\theta(k)) = \Theta \Phi_{S'}(t_{e_2}(k)) \act \eps = \Phi_{S'}(\theta t_{e_2}(k)) = \Phi_{S'}(t_{-e_2}\theta(k)) \]
as wanted. This shows that $\Phi_{S'}$ is continuous on
\[ S' := S \cup \left( F_{1,0} \cap \set{k_3 \le 0} \right) \cup \left( F_{1,+} \cap \set{k_3 \le 0} \right). \]

\noindent \textsl{Step 5.} Extend the definition of $\Phi_{S'}$ to $F_{3,-}$ by setting
\begin{align*}
\Phi_{S''}(k) & := \Phi_{S'}(k) && \text{for } k \in S', \quad \text{and} \\
\Phi_{S''}(t_{-e_3}(k)) & := \tau_{-e_3} \Phi_{S'}(k) && \text{for } k \in F_{3,+}.
\end{align*}

This extension is continuous on $t_{-e_3}(E_{23}) = F_{2,-} \cap F_{3,-}$ because in Step 3b we have imposed on $\Phi_S$ the edge symmetries on the shorter edges of $F_{2,-}$. Similarly, $\Phi_{S''}$ and $\Phi_{S'}$ agree on $\theta(E_{26}) \cup t_{-e_3}(E_{26}) = F_{1,0} \cap F_{3,-}$ and on $t_{-e_3}(E_{37}) \cup t_{e_1} \theta(E_{37}) = F_{1,+} \cap F_{3,-}$, because by construction $\Phi_{S'}$ is $\tau$-equivariant. We have constructed so far a continuous Bloch frame $\Phi_{S''}$ on $S'' := S' \cup F_{3,-}$.

\noindent \textsl{Step 6.} Extend the definition of $\Phi_{S''}$ to $F_{2,+}$ by setting
\begin{align*}
\widehat{\Phi}(k) & := \Phi_{S''}(k) && \text{for } k \in S'', \quad \text{and} \\
\widehat{\Phi}(t_{e_2}(k)) & := \tau_{e_2} \Phi_{S''}(k) && \text{for } k \in F_{2,-}.
\end{align*}

Similarly to what was argued in Step 5, the edge symmetries (\ie $\tau$-equivariance) on $\partial F_{2,-}$ and $\partial F_{2,+}$ imply that this extension is continuous. The Bloch frame $\widehat{\Phi}$ is now defined and continuous on $\partial \Bred^{(3)}$.

\noindent \textsl{Step 7.} Finally, extend the definition of $\widehat{\Phi}$ on the interior of $\Bred^{(3)}$. As was argued before at the end of Section \ref{sec:3dObs}, this step is topologically unobstructed. Thus, we end up with a continuous symmetric Bloch frame $\eff{\Phi}$ on $\Bred^{(3)}$. Set now
\[ \uc{\Phi}(k) := \begin{cases}
\eff{\Phi}(k) & \text{if } k \in \Bred^{(3)}, \\
\Theta \eff{\Phi}(\theta(k)) \act \eps & \text{if } k \in \B^{(3)} \setminus \Bred^{(3)}.
\end{cases} \]
The symmetries satisfied by $\eff{\Phi}$ on $F_{1,0}$ (imposed in Step 4) imply that the Bloch frame $\uc{\Phi}$ is still continuous. We extend now the definition of $\uc{\Phi}$ to the whole $\R^3$ by setting
\[ \Phi(k) := \tau_{\lambda} \uc{\Phi}(k'), \quad \text{if } k = k' + \lambda \text{ with } k' \in \B^{(3)}, \: \lambda \in \Lambda. \]
Again, the symmetries imposed on $\uc{\Phi}$ in the previous Steps of the proof ensure that $\Phi$ is continuous, and by construction it is also symmetric. This concludes the proof of Theorem \ref{thm:Q3}.
\end{proof}

\subsection{Comparison with the Fu-Kane-Mele indices}
\label{sec: GL(3,Z) action}

In their work \cite{FuKaneMele}, Fu, Kane and Mele generalized the definition of $\Z_2$ indices for $2d$ topological insulators, appearing in \cite{FuKa}, to $3d$ topological insulators. There the authors mainly use the formulation of $\Z_2$ indices given by evaluation of certain quantities at the inequivalent time-reversal invariant momenta $k_\lambda$, as in \eqref{DeltaFromPoints}. We briefly recall the definition of the four Fu-Kane-Mele $\Z_2$ indices for the reader's convenience.

In the $3$-dimensional unit cell $\B^{(3)}$, there are only eight vertices which are inequivalent up to periodicity: these are the points $k_\lambda$ corresponding to $\lambda = (n_1/2, n_2/2, n_3/2)$, where each $n_j$ can be either $0$ or $1$. These are the vertices in Figure \ref{fig:Bred3} that are connected by thick lines. Define the matrix $w$ as in \eqref{w(k)}, where $\Psi$ is a continuous and $\tau$-equivariant Bloch frame (whose existence is guaranteed again by the generalization of the results of \cite{Panati} given in \cite{PaMo}). Set
\[ \eta_{n_1, n_2, n_3} := \left. \frac{\sqrt{\det w(k_\lambda)}}{\Pf w(k_\lambda)} \right|_{\lambda = (n_1/2, n_2/2, n_3/2)}. \]
Then the four Fu-Kane-Mele $\Z_2$ indices $\nu_0, \nu_1, \nu_2, \nu_3 \in \Z_2$ are defined as \cite[Eqn.s (2) and (3)]{FuKaneMele}
\begin{align*}
(-1)^{\nu_0} & := \prod_{n_1, n_2, n_3 \in \set{0,1}} \eta_{n_1, n_2, n_3}, \\
(-1)^{\nu_i} & := \prod_{n_i = 1, \: n_{j \ne i} \in \set{0,1}} \eta_{n_1, n_2, n_3}, \quad i \in \set{1, 2, 3}.
\end{align*}
In other words, the Fu-Kane-Mele $3d$ index $(-1)^{\nu_i}$ equals the Fu-Kane $2d$ index $(-1)^{\Delta}$ for the face where the $i$-th coordinate is set equal to $1/2$, while the index $(-1)^{\nu_0}$ involves the product over all the inequivalent time-reversal invariant momenta $k_\lambda$.

Since our invariants $\delta_{1,0}$ and $\delta_{i,+}$, $i = 1, 2, 3$, are defined as the $2$-dimensional $\Z_2$ invariants relative to certain (effective) faces on the boundary of the $3$-dimensional unit cell, they also satisfy identities which express them as product of quantities to be evaluated at vertices of the \effective unit cell, as in Section \ref{sec:EvalVertices} (compare \eqref{easydelta} and \eqref{easierdelta}). Explicitly, these expressions read as
\begin{align*}
(-1)^{\delta_{1,0}} & := \prod_{n_1 = 0, n_2, n_3 \in \set{0,1}} \widehat{\eta}_{n_1, n_2, n_3}, \\
(-1)^{\delta_{i,+}} & := \prod_{n_i = 1, \: n_{j \ne i} \in \set{0,1}} \widehat{\eta}_{n_1, n_2, n_3}, \quad i \in \set{1, 2, 3},
\end{align*}
where (compare \eqref{easydelta})
\[ \widehat{\eta}_{n_1, n_2, n_3} := \left. \frac{\sqrt{\left(\det U(k_\lambda)\right)^2}}{\det U(k_\lambda)} \right|_{\lambda = (n_1/2, n_2/2, n_3/2)}. \]
As we have also shown in Section \ref{sec:Fu-Kane} that our $2$-dimensional invariant $\delta$ agrees with the Fu-Kane index $\Delta$, from the previous considerations it follows at once that
\[ \nu_i = \delta_{i,+} \in \Z_2, \: i \in \set{1,2,3}, \quad \text{and} \quad \nu_0 = \delta_{1,0} + \delta_{1, +} \in \Z_2. \]
This shows that the four $\Z_2$ indices proposed by Fu, Kane and Mele are compatible with ours, and in turn our reformulation proves that indeed these $\Z_2$ indices represent the obstruction to the existence of a continuous symmetric Bloch frame in $3d$. The geometric r\^ole of the Fu-Kane-Mele indices is thus made transparent. 

We also observe that, under this identification, the indices $\nu_i$, $i \in \set{1,2,3}$, depend manifestly on the choice of a basis $\set{e_1, e_2, e_3} \subset \R^d$ for the lattice $\Lambda$, while the index $\nu_0$ is \emph{independent} of such a choice. This substantiates the terminology of \cite{FuKaneMele}, where $\nu_0$ is called ``strong'' invariant, while $\nu_1$, $\nu_2$ and $\nu_3$ are referred to as ``weak'' invariants.

Elaborating on the latter point, notice that there is a group of affine transformations of $\R^3$ which do not change the description of the physical system described by the family of projectors $\set{P(k)}_{k \in \R^3}$, consisting of
\begin{itemize}
 \item the group $GL(3;\Z)$ of linear transformations that leave the lattice $\Lambda$ fixed, and 
 \item the translations by vectors with half-integer coordinates, which correspond to the choice of an origin in $\R^3$ (and thus, to a change in the electromagnetic gauge in the Hamiltonian description of the physical system \cite[Remark 2]{PaMo}).
\end{itemize} 
Any sensible definition of a ``strong invariant'' should indeed stay unchanged even when one of these transformation is applied.

To study which ``strong invariants'', in this sense, can be constructed from our four $\Z_2$ invariants $\set{\delta_{1,0}, \, \delta_{1,+}, \, \delta_{2,+}, \, \delta_{3,+}}$, we first observe that the above group of transformation is generated by the following three:
\begin{align*}
s_1(k_1, k_2, k_3) & := (-k_3, -k_1, -k_2), \\
s_2(k_1, k_2, k_3) & := (k_1 + k_2, k_2, k_3), \\
t(k_1, k_2, k_3) & := (k_1 + 1/2, k_2, k_3)
\end{align*}
(see \eg \cite{HuaReiner}). It is easy to compute how the quadruples $\left(\delta_{1,0}, \, \delta_{1,+}, \, \delta_{2,+}, \, \delta_{3,+}\right) \in \Z_2^4$ change according to these transformations: denoting with a slight abuse of notation the induced transformations with the same symbols, one obtains explicitly
\begin{equation} \label{actionZ24}
\begin{aligned}
s_1\left(\delta_{1,0}, \, \delta_{1,+}, \, \delta_{2,+}, \, \delta_{3,+}\right) & = \left(\delta_{1,0} + \delta_{1,+} + \delta_{2,+}, \, \delta_{2,+}, \, \delta_{3,+}, \, \delta_{1,+}\right), \\
s_2\left(\delta_{1,0}, \, \delta_{1,+}, \, \delta_{2,+}, \, \delta_{3,+}\right) & = \left(\delta_{1,0} + \delta_{2,+}, \, \delta_{1,+} + \delta_{2,+}, \, \delta_{2,+}, \, \delta_{3,+}\right), \\
t\left(\delta_{1,0}, \, \delta_{1,+}, \, \delta_{2,+}, \, \delta_{3,+}\right) & = \left(\delta_{1,+}, \, \delta_{1,0}, \, \delta_{2,+}, \, \delta_{3,+}\right).
\end{aligned}
\end{equation}

The orbits of $\Z_2^4$ under this action are as follows:
\begin{align*}
\mathcal{O}_1 & := \set{(0,0,0,0)}, \\
\mathcal{O}_2 & := \set{\left(\delta_{1,0}, \, \delta_{1,+}, \, \delta_{2,+}, \, \delta_{3,+}\right) \in \Z_2^4 \setminus \set{(0,0,0,0)} : \delta_{1,0} + \delta_{1,+} = 0}, \\ 
\mathcal{O}_3 & := \set{\left(\delta_{1,0}, \, \delta_{1,+}, \, \delta_{2,+}, \, \delta_{3,+}\right) \in \Z_2^4 : \delta_{1,0} + \delta_{1,+} = 1}.
\end{align*}
A ``strong invariant'', in the sense specified above, would be a map $\nu \colon \Z_2^4 \to \Z_2$ that is equivariant with respect to the above action, or equivalently that is constant on these orbits. It can be computed that all such equivariant maps which are algebraic in the entries $\left(\delta_{1,0}, \ldots, \delta_{3,+}\right)$ are combinations of the following three:
\begin{itemize}
 \item $\nu\sub{trivial} \equiv 1$, the trivial constant map;
 \item $\nu_0\left(\delta_{1,0}, \, \delta_{1,+}, \, \delta_{2,+}, \, \delta_{3,+}\right) = \delta_{1,0} + \delta_{1,+}$, namely the Fu-Kane-Mele ``strong'' invariant;
 \item $\nu\sub{tot}$, the sum of all possible monomials in the variables $\left(\delta_{1,0}, \ldots, \delta_{3,+}\right) \in \Z_2^4$:
 \[ \nu\sub{tot} \left(\delta_{1,0}, \, \delta_{1,+}, \, \delta_{2,+}, \, \delta_{3,+}\right) = \sum_{s_0, s_1, s_2, s_3 \in \set{0,1}} \delta_{1,0}^{s_0} \, \delta_{1,+}^{s_1} \, \delta_{2,+}^{s_2} \, \delta_{3,+}^{s_3}. \]
\end{itemize}

Let us also remark that, among the above orbits, $\mathcal{O}_1$ is distinguished by the fact that it corresponds to $3$-dimensional families of projectors, satisfying Assumption \ref{proj}, for which a global continuous symmetric Bloch frame exists. In particular, the vanishing of the Fu-Kane-Mele ``strong'' invariant $\nu_0$ does not suffice to distinguish such families of projectors, since $\nu_0 = 0$ also on $\mathcal{O}_2$. There is, however, a ``secondary invariant'' that is able to characterize the quadruple of invariants $(0,0,0,0)$ among those which satisfy $\nu_0 = 0$: this is
\[ \omega := \sum_{s_1, s_2, s_3 \in \set{0,1}} \delta_{1,+}^{s_1} \, \delta_{2,+}^{s_2} \, \delta_{3,+}^{s_3}. \]
This expression is invariant under the action \eqref{actionZ24} provided that $\nu_0 = 0$, and among the quadruples $\left(\delta_{1,0}, \, \delta_{1,+}, \, \delta_{2,+}, \, \delta_{3,+}\right)$ satisfying the latter condition one has $\omega = 0$ if and only if $\left(\delta_{1,0}, \, \delta_{1,+}, \, \delta_{2,+}, \, \delta_{3,+}\right) = (0,0,0,0)$. In particular, to decide whether a $3$-dimensional family of projectors satisfying Assumption \ref{proj} admits a global continuous symmetric Bloch frame, one needs to compute first the Fu-Kane-Mele invariant $\nu_0$: if this vanishes, then one computes $\omega$, and it is only if also this quantity vanishes, that such a Bloch frame indeed exists.

\newpage

%%%%% SECTION 7

\section{Completeness of the $\Z_2$ invariants} \label{sec:completeness}

In the previous Sections, we have identified the topological obstructions to the existence of a global continuous symmetric Bloch frame for a family of projectors $\mathcal{P} = \set{P(k)}_{k \in \R^d}$ satisfying Assumption \ref{proj}, and encoded them in the $\Z_2$ invariants defined in \eqref{delta} when $d=2$, and in \eqref{3Ddelta} when $d=3$. The vanishing of these invariants characterizes the topological class of those families $\mathcal{P}$ admitting such a Bloch frame.

It is a natural question whether the same holds true in general, namely if the $\Z_2$ invariants characterize completely the topological classes of $d$-dimensional families of projectors satisfying Assumption \ref{proj}; up to now, we have only proved that the ``trivial'' class of projectors admitting a global continuous symmetric Bloch frame is characterized by the vanishing of all $\Z_2$ invariants. An affirmative answer to this question is provided by the following
\begin{thm} \label{thm:Completeness}
Let $\set{P_1(k)}_{k \in \R^d}$ and $\set{P_2(k)}_{k \in \R^d}$, with $d \le 3$, be two families of projectors on the separable Hilbert space $\Hi$ satisfying Assumption \ref{proj}, with unitary representations $\tau_1$ and $\tau_2$ and time-reversal operators $\Theta_1$ and $\Theta_2$, respectively. Assume moreover that $\dim \Ran P_1(k) = \dim \Ran P_2(k) \equiv m$, and:
\begin{itemize}
 \item if $d=2$, that 
 \[ \delta(P_1) = \delta(P_2) \in \Z_2; \]
 \item if $d=3$, that 
 \[ \left(\delta_{1,0}(P_1), \, \delta_{1,+}(P_1), \, \delta_{2,+}(P_1), \, \delta_{3,+}(P_1) \right) = \left(\delta_{1,0}(P_2), \, \delta_{1,+}(P_2), \, \delta_{2,+}(P_2), \, \delta_{3,+}(P_2) \right) \in \Z_2^4. \]
\end{itemize}
Then there exists a family of partial-isometries%%%
\footnote{ A bounded linear operator $V \in \mathcal{B}(\Hi)$  is said to be a \emph{partial isometry} if $V|_{\Ker(V )^{\perp}}$ 
is an isometry. In particular, $V$ restricts to a unitary operator from $\Ker(V )^{\perp}$ to $\mathrm{Im}(V)$, which are called 
\emph{initial subspace} and \emph{final subspace}, respectively. 
The set of partial isometries is hereafter denoted by $\mathcal{I}(\Hi)$.   
} %%% End footnote
 $\set{V(k)}_{k \in \R^d}$, with initial subspace $\Ran P_1(k)$ and final subspace $\Ran P_2(k)$, such that
\begin{enumerate}[label=$(\mathrm{V}_\arabic*)$,ref=$(\mathrm{V}_\arabic*)$]
 \item \label{item:U-smooth} $V \colon \R^d \to \mathcal{I}(\Hi) \subset \mathcal{B}(\Hi)$ is $C^\infty$-smooth;
 \item \label{item:U-tau} $V$ enjoys the $(\tau_2, \tau_1)$-covariance property
 \[ V(k + \lambda) = \tau_2(\lambda) \, V(k) \, \tau_1(\lambda)^{-1}, \quad k \in \R^d, \: \lambda \in \Lambda; \]
 \item \label{item:U-TRS} $V$ enjoys ``time-reversal symmetry'', namely
 \[ V(-k) = \Theta_2 \, V(k) \, \Theta_1^{-1}; \]
 \item \label{item:U-inter} $V$ intertwines $P_1$ and $P_2$, namely
 \[ P_2(k) V(k)  = V(k) \, P_1(k). \]
\end{enumerate}
%Moreover, if $\Hi$ is finite dimensional, then $V(k)$ extends to a \emph{unitary} operator $W(k)$ such that the map  
%$W \colon \R^d \to \U(\Hi)$ satisfies the properties \ref{item:U-smooth}, \ref{item:U-tau}, \ref{item:U-TRS} and \ref{item:U-inter}. 
\end{thm}

\begin{rmk}[Isomorphism classes of time-reversal invariant Bloch bundles] \label{rmk: TRS Bloch bundles}
We provide a geometric reformulation of the above result, in terms of time-reversal invariant Bloch bundles (adopting the terminology of \cite{GrafPorta}). Using the construction of \cite{Panati}, one can naturally associate a smooth Hermitian vector bundle $\mathcal{E} = \left( E \xrightarrow{\pi} \T^d \right)$, called the \emph{Bloch bundle}, to any family of projectors $\set{P(k)}_{k \in \R^d}$ satisfying \ref{item:smooth} and \ref{item:tau}, where the (Brillouin) torus $\T^d$ is the quotient $\R^d / \Lambda$. The fiber $\mathcal{E}_k$ over the point $k \in \T^d$ is the range of the projector $P(k)$. If the family of projectors is also time-reversal symmetric, \ie satisfies also \ref{item:TRS}, the Bloch bundle can be endowed with a further structure, namely that of a fiberwise antiunitary map $\widehat{\Theta} \colon \mathcal{E} \to \mathcal{E}$, squaring fiber-by-fiber to $-\Id$ and making the following diagram commutative:
\[ \xymatrix{ E \ar[r]^{\widehat{\Theta}} \ar[d]_{\pi} & E \ar[d]^{\pi} \\ \T^d \ar[r]^{\theta} & \T^d} \]
where $\theta(k) = -k$. One speaks then of a \emph{time-reversal invariant Bloch bundle}.

The natural notion of isomorphism for time-reversal invariant Bloch bundles $\mathcal{E}_1$ and $\mathcal{E}_2$ requires the existence of a smooth isomorphism of Hermitian bundles $\widehat V \colon \mathcal{E}_1 \to \mathcal{E}_2$ which is compatible with the time-reversal structure (\emph{TRS isomorphism}), namely such that the following diagram
\begin{equation} \label{TRSdiagram}
\xymatrix{
E_1 \ar[rr]^{\widehat V} \ar[dr]_{\pi_1} \ar[ddd]_{\widehat{\Theta}_1} & & E_2 \ar[dl]^{\pi_2} \ar[ddd]^{\widehat{\Theta}_2} \\
& \T^d \ar[d]^{\theta} & \\
& \T^d & \\
E_1 \ar[rr]_{\widehat V} \ar[ur]^{\pi_1} & & E_2 \ar[ul]_{\pi_2} }
\end{equation}
commutes. Such an isomorphism, when looked at the level of fibers, can be recast in a {partial-isometry-valued map $V \colon \R^d \to \mathcal{I}(\Hi)$} satisfying exactly the conditions \ref{item:U-smooth}, \ref{item:U-tau}, \ref{item:U-TRS} and \ref{item:U-inter} stated in Theorem \ref{thm:Completeness}: conditions \ref{item:U-smooth}, \ref{item:U-tau}, and \ref{item:U-inter} (the latter implying the equality of the dimensions of the range of the projectors, \ie of their rank) come from the fact that 
$\widehat V \colon \mathcal{E}_1 \to \mathcal{E}_2$ is a smooth isomorphism of the Bloch bundles associated to the families of projectors $\set{P_1(k)}_{k \in \R^d}$ and $\set{P_2(k)}_{k \in \R^d}$, while \ref{item:U-TRS} encodes the commutativity of the diagram \eqref{TRSdiagram}. \soloarXiv{The equivalence of the two concepts is discussed in Appendix \ref{app:Isomorphisms}.}
 
%An important special case deserves particular attention; whenever the map $V$ takes values in the \emph{unitary group} $\U(\Hi) 
%\subset \mathcal{I}(\Hi)$, as in the second part of the statement of Theorem \ref{thm:Completeness}, it induces both a TRS isomorphism $\mathcal{E}_1 \to \mathcal{E}_2$  and a TRS isomorphism $\mathcal{E}_1^{\perp} \to \mathcal{E}_2^{\perp}$ of the orthogonal complements, the latter having a natural structure of time-reversal invariant Bloch bundles. 

Theorem \ref{thm:Completeness} implies that, if two families of projectors satisfying Assumption \ref{proj} have the same rank and the same $\Z_2$ invariants, then their associated Bloch bundles are isomorphic as time-reversal invariant bundles. 
%%% 
%\footnote{\gp{It also implies that, whenever $\dim \Hi < + \infty$, the orthogonal complements $\mathcal{E}_1^{\perp}$ and $\mathcal{E}_2^{\perp}$ are TRS isomorphic. However, this additional information is not needed to discuss the completeness of the $\Z_2$ invariants.}} 
%%% End footnote 
Conversely, if they are isomorphic, then any isomorphism compatible with the time-reversal structure induces 
a map $V \colon \R^d \to \mathcal{I}(\Hi)$ satisfying \ref{item:U-smooth}, \ref{item:U-tau}, \ref{item:U-TRS} and \ref{item:U-inter}: in particular, in view of \ref{item:U-inter}, a global continuous Bloch frame $\Psi_1$ for $\set{P_1(k)}_{k \in \R^d}$ automatically gives rise to a global continuous Bloch frame $\Psi_2$ for $\set{P_2(k)}_{k \in \R^d}$ 
by setting $\Psi_2(k) := V(k) \Psi_1(k)$. Since these two frames, which can be used to compute the $\Z_2$ invariants, are then related by a unitary isometry, arguing as in the proof of Proposition \ref{deltaTopInv} one can prove that the two families of projectors have the same $\Z_2$ invariants.

We conclude that the $\Z_2$ invariants $\delta(P) \in \Z_2$ (when $d=2$) and $( \delta_{1,0}(P)$, $\delta_{1,+}(P), \, \delta_{2,+}(P), \, \delta_{3,+}(P) ) \in \Z_2^4$ (when $d=3$), {together with the rank}, are indeed a complete set of invariants for a family of projectors $\set{P(k)}_{k \in \R^d}$ satisfying Assumption \ref{proj}, with respect to TRS isomorphism of the associated Bloch bundles.   
This is compatible with the results of \cite{deNittisGomi}, which predict a $\Z_2$ and $\Z_2^4$ classification for isomorphism classes of time-reversal invariant (or ``Quaternionic'') Bloch bundles {with fixed rank $m$} in dimension $d=2$ and $d=3$, respectively: our $\Z_2$ indices provide numerical invariants to distinguish these isomorphism classes. 

An important remark is in order here.  The previous claims hold true for a fixed choice of a linear basis for $\Lambda$ and  
an origin in $\R^d$, as in the statement of Theorem \ref{thm:Completeness} (the choice of the origin is implicit in \ref{item:TRS},  the choice of the basis in the labeling of the $3$-dimensional indices). Without such a choice, 
a complete set of invariants for $3$-dimensional time-reversal invariant Bloch bundles is provided instead 
by the \emph{equivalence class (orbit)} of the quadruple  
$( \delta_{1,0}(P)$, $\delta_{1,+}(P), \, \delta_{2,+}(P), \, \delta_{3,+}(P) ) \in \Z_2^4$ 
under the action of  $GL(3, \Z)$ and of the translation group  ${\tiny \frac{1}{2}} \Lambda$, 
as discussed at the end of  Section \ref{sec: GL(3,Z) action}, together with the rank. 
This distinction is clearly immaterial for $d =2$. 
 \end{rmk}

\begin{rmk}[Quantum spin Hall systems]
The completeness, given by Theorem \ref{thm:Completeness}, of the $2$-dimensional $\Z_2$ invariant $\delta$ defined in \eqref{delta}, also confirms the classification of quantum spin Hall systems in exactly two topologically distinct classes \cite{FuKa,GrafPorta}, the one of the ``trivial'' insulator and the one of the ``topological insulator'', by means of a \emph{bulk} invariant.
\end{rmk}

\begin{proof}[Proof of Theorem \ref{thm:Completeness}]
We give a sketch of the proof, articulating it in various steps; the careful reader can fill in the details.

Before sketching the argument, we make a number of preliminary observations. First of all, in this proof we provide a \emph{continuous} map $V \colon \R^d \to \mathcal{I}(\Hi)$, satisfying all the properties required in the statement; a regular one can be obtained by a smoothing argument similar to the one presented in Appendix \ref{app:Smoothing}. Moreover, it suffices to construct $V$ on $\Bred$, extending it to the whole $\R^d$ by first imposing \ref{item:U-TRS} on $\B$ and then \ref{item:U-tau}. In turn, conditions \ref{item:U-tau} and \ref{item:U-TRS} impose some ``edge'' and ``face symmetries'' on the map $V$ to be defined on $\Bred$. 

\medskip

\noindent \textsl{Step I. Existence of an intertwining unitary.} We begin by constructing a ``reference intertwining unitary'' $W_*$ satisfying \ref{item:U-smooth}, \ref{item:U-tau} and \ref{item:U-inter}. Its existence is a reformulation of the results in \cite{Panati,PaMo}, and it holds even without the assumptions on the $\Z_2$ invariants of the two families of projectors being equal. 

%Later, we will modify this map $W_*$, \gp{restricted to $\Ran P_1$}, to a map $V$ which satisfies also \ref{item:U-TRS}.

As a first auxiliary result, we need the following

\begin{lemma}[Existence of a trivializing unitary] \label{Lem:Trivializing G}
Let $\set{P(k)}_{k \in \R^d}$ be a family of projectors satisfying Assumption \ref{proj}, for $d \leq 3$. 
Then there exists $G : \R^d \to \U(\Hi)$ such that: 
\begin{enumerate}[label=$(\mathrm{G}_\arabic*)$,ref=$(\mathrm{G}_\arabic*)$]
 \item \label{item:G-smooth} $G \colon \R^d \to \U(\Hi)$ is $C^\infty$-smooth;
 \item \label{item:G-tau} $G$ enjoys the $(\Id, \tau)$-covariance property
 \[ G(k + \lambda) =   G(k) \, \tau(\lambda)^{-1},   \quad k \in \R^d, \: \lambda \in \Lambda; \]
 \item \label{item:G-inter} $G$ enjoys the trivializing property
 \[ G(k) \, P(k) \, G(k)^{-1} = P(0). \]
\end{enumerate}
\end{lemma}

\begin{proof}[Proof of Lemma \ref{Lem:Trivializing G}]
In view of the results of \cite{Panati,PaMo}, the Bloch bundle $\mathcal{E}$ associated to the family $\set{P(k)}_{k \in \R^d}$ is trivial, because it has vanishing first Chern class $\mathrm{Ch}_1(\mathcal{E})$ {and $d \leq 3$} (see \cite[Proposition 4]{Panati}). Now, if the ambient Hilbert space $\Hi$ is finite-dimensional, this also implies the triviality of the \emph{orthogonal} Bloch bundle $\mathcal{E}^\perp$, whose fibre over the point $k \in \T^d$ is the space $(\Ran P(k))^\perp$. Indeed, we have
\[ \mathrm{Ch}_1(\mathcal{E}^\perp) = \mathrm{Ch}_1(\mathcal{E}) + \mathrm{Ch}_1(\mathcal{E}^\perp) = \mathrm{Ch}_1(\mathcal{E} \oplus \mathcal{E}^\perp) = \mathrm{Ch}_1(\T^d \times \Hi) = 0 \]
and again by the technical Lemma in \cite{Panati} the bundle $\mathcal{E}^\perp$ is trivial, since $d \le 3$. If $\Hi$ is instead infinite-dimensional, then the orthogonal bundle $\mathcal{E}^\perp$ has also infinite-dimensional typical fibre $\mathcal{K}$; it follows from Kuiper's theorem that $\mathcal{E}^\perp$ is again trivial \cite{Kuiper}.

Thus, regardless of the dimensionality of the ambient Hilbert space $\Hi$, we conclude that there exist isomorphisms of Hermitian bundles $F \colon \mathcal{E} \to \T^d \times \Ran P(0)$ and $F^\perp \colon \mathcal{E}^\perp \to \T^d \times (\Ran P(0))^\perp$. 
Since $\mathcal{E}$ and $\mathcal{E}^\perp$ are constructed via the equivalence relation  $ \sim_\tau$ \cite{Panati}, defined by
\begin{equation} \label{tau equivalence}
(k, \phi) \sim_\tau (k', \phi') \quad \text{if and only if} \quad \exists\, \lambda \in \Lambda : 
k' = k + \lambda \text{ and } \phi' = \tau(\lambda) \phi, 
%% Questo segno e' corretto! E' errato quello in [MoPa]
\end{equation}
the direct sum  $F \oplus F^\perp : \mathcal{E} \oplus \mathcal{E}^\perp \to \T^d \times \Hi$ 
can be regarded as a map on the equivalence classes, namely
$$
F \oplus F^\perp: (\R^d \times \Hi) / _{\sim_\tau} \longrightarrow  (\R^d \times \Hi) / _{\sim_\Id}.
$$   
Notice that, in particular, it maps $\tau$-equivariant Bloch frames into periodic ones. 
Therefore, it induces a smooth map $G \colon \R^d \to \U(\Hi)$ such that
\ref{item:G-tau} and \ref{item:G-inter} hold true. 
\end{proof} %%%%%%%%%%%%%%%%%%%%%%%%%%%%%%%%%%%%%%%

Applying Lemma \ref{Lem:Trivializing G} to both families $\set{P_1(k)}_{k \in \R^d}$ and $\set{P_2(k)}_{k \in \R^d}$ as in the statement of the Theorem, we obtain that there exist continuous $(\Id, \tau_j)$-covariant maps $G_j \colon \R^d \to \U(\Hi)$, $j \in \set{1,2}$, such that
\begin{equation} \label{Id0}
G_1(k) \, P_1(k) \, G_1(k)^{-1} = P_1(0), \quad G_2(k) \, P_2(k) \, G_2(k)^{-1} = P_2(0).
\end{equation}
We now fix an identification of $\Ran P_1(0)$ and $\Ran P_2(0)$, by choosing a unitary $V_* \in \U(\Hi)$ such that 
\[ P_2(0) = V_* \, P_1(0) \, V_*^{-1} 
\mbox{ and hence also }
(\Id_\Hi - P_2(0)) = V_* \, (\Id_\Hi - P_1(0)) \, V_*^{-1}, \]
so that we have the following commutative diagram:
\begin{equation} \label{Diagram}
\xymatrix{ \Ran P_1(k) \ar[r]^{G_1(k)} \ar[d]_{W_*(k)} & \Ran P_1(0) \ar[d]^{V_*} \\ \Ran P_2(k) \ar[r]^{G_2(k)} & \Ran P_2(0)} 
\end{equation}
We conclude that the map
\begin{equation} \label{W_*}
W_*(k) := G_2(k)^{-1} \, V_* \, G_1(k),   
\end{equation}
satisfies properties \ref{item:U-smooth}, \ref{item:U-tau} and \ref{item:U-inter}.
Our next goal is to modify $W_*|_{\Ran P_1}$ into $V$, in such a way that also property  \ref{item:U-TRS} is satisfied.

\bigskip

%\gp{ %%%%
\noindent \textsl{Step II. Reduction to a gauge transform.} In view of the above construction we can assume without loss of generality to have fixed an identification of $\Ran P_1(k)$ with $\Ran P_2(k)$, 
and similarly for the orthogonal subspaces. 
%which turns the required map $W$ in a direct sum of partial isometries $V \oplus V^\perp$, with respect to the orthogonal splitting of $\Hi$ induced by $\Ran P_1(k)$ on the source space, and by $\Ran P_2(k)$ on the target space. 

More explicitly, the map $W_*$ provides an initial identification of $\Ran P_1(k)$ and $\Ran P_2(k)$, which may be re-expressed as an identification of a smooth $\tau_1$-equivariant Bloch frame $\Psi_1$  and its image  $\Psi_2$, namely $\Psi_2(k) := W_*(k) \Psi_1(k)$.  Any  partial isometry $V$ satisfying \ref{item:U-inter} is characterized  by a map $U: \R^d \rightarrow \U(\C^m)$ (gauge transform) via the identity
\begin{equation} \label{W characterization}
V(k) \Psi_1(k) =  \big(W_*(k)\Psi_1(k) \big) \act U(k).    %\quad \qquad \mbox{ where } V = W|_{\Ran P_1}.
\end{equation} 
To prove Theorem \ref{thm:Completeness} we just need to show how to construct a map $U \colon \R^d \to \U(\C^m)$
which induces, via \eqref{W characterization}, a partial-isometry-valued  map $V$ satisfying \ref{item:U-smooth}, \ref{item:U-tau} and \ref{item:U-TRS}.

\bigskip

\noindent \textsl{Step III. Enforcing the time-reversal symmetry.} We conclude the proof with a dimension dependent argument.
 
 \noindent \textsl{Step III.1: $d=1$.} From Theorem \ref{thm:Q1}, we know that there exist global continuous symmetric Bloch frames for $\set{P_1(k)}_{k \in \R}$ and $\set{P_2(k)}_{k \in \R}$: denote them by $\Phi_1$ and $\Phi_2$, respectively. 
%Similarly, choose global continuous symmetric Bloch frames $\Phi_1^\perp$ and $\Phi_2^\perp$ for the orthogonal families of projectors: this can be done by the argument in Step II. 
Let $V$ be defined point-wise by the relation $\Phi_2(k) = V(k) \Phi_1(k)$. 
%and $\Phi_2^\perp(k) = W(k) \Phi_1^\perp(k)$. 
Then \ref{item:U-inter} is trivially satisfied, and moreover the $\tau$-equivariance \ref{tau-cov} and the time-reversal invariance \ref{tr} for the Bloch frames translate exactly into properties \ref{item:U-tau} and \ref{item:U-TRS} for $V$.

\noindent  \textsl{Step III.2: $d=2$.}  In view of the above identification \eqref{Diagram} of $\Ran P_1(k)$ and $\Ran P_2(k)$ with 
$\Ran P_1(0) \simeq \C^m$, fixing a basis $\Psi$ in $\C^m$ gives a global continuous $\tau$-equivariant Bloch frame both for $\set{P_1(k)}_{k \in \R^2}$ and for $\set{P_2(k)}_{k \in \R^2}$. We denote them by $\Psi_1(k)$  and $\Psi_2(k) = W_*(k) \Psi_1(k)$, respectively.

Now construct
\[ \widehat{\Phi}_1(k) = \Psi_1(k) \act \widehat{U}_1(k), \quad \widehat{\Phi}_2(k) = \Psi_2(k) \act \widehat{U}_2(k), \quad k \in \partial \Bred, \]
as in \eqref{hatU}. By assumption
\[ \delta(P_1) \equiv \deg([\det \widehat{U}_1]) \equiv \deg([\det \widehat{U}_2]) \equiv \delta(P_2) \bmod 2, \]
or equivalently
\[ \deg([\det \widehat{U}_1]) = \deg([\det \widehat{U}_2]) - 2 s, \quad s \in \Z. \]
Let $X \colon \partial \Bred \to \U(\C^m)$ be as in the proof of Theorem \ref{thm:Q2'}, so that $X$ satisfies \ref{XiSymm} and $\deg([\det X]) = -2s$. It follows that $\Psi_2(k) \act \left( \widehat{U}_2(k) \, X(k) \right) = \widehat{\Phi}_2(k) \act X(k)$ also satisfies the vertex conditions \eqref{VertexCondition} and the edge symmetries \eqref{EdgeSymmetries} by \ref{XiSymm}, and moreover $\deg([\det \widehat{U}_2 \, X]) = \deg([\det \widehat{U}_1])$. This shows that we may assume, without loss of generality, that 
\begin{equation} \label{equaldeg}
\deg([\det \widehat{U}_1]) = \deg([\det \widehat{U}_2]) \in \Z.
\end{equation}

\medskip

We  define a partial-isometry-valued map $V$, initially for $k \in \partial \Bred$, by imposing the relation 
\begin{equation} \label{V intertwines}
V(k) \widehat \Phi_1(k) =  \widehat \Phi_2(k) \qquad \qquad \mbox{ for any }  k \in \partial \Bred. 
\end{equation}
\noindent In term of the $\Psi$'s frames, the  above definition reads  
\begin{eqnarray} \label{W definition}
V(k) \Psi_1(k) 
&= \Psi_2(k) \act  \left(\widehat U_2(k) \widehat U_1(k)^{-1} \right),  
\end{eqnarray}
where, as before,  $\Psi_2(k) = W_*(k) \Psi_1(k)$.  The map $ k \mapsto \widetilde{U}(k) := \widehat U_2(k) \widehat U_1(k)^{-1}$,  initially defined for $k \in \partial \Bred$,  extends continuously to a map $U \colon \Bred \to \U(\C^m)$ since $\deg([\det \widetilde{U}]) = 0$ by \eqref{equaldeg}. Therefore,  $V$  extends to a continuous 
$\mathcal{I}(\Hi)$-valued map on $\Bred$ by using \eqref{W characterization}.

\newpage 

Since $V(k)$ intertwines two frames that are symmetric on $\partial \Bred$, it satisfies 
\begin{align*}
V(-k)  & = \Theta_2 V(k) \Theta_1^{-1}  && \text{for every }  k \in E_1 \cup E_6, \\
V(-k  + e_1) & = \tau_2(e_1)  \Theta_2 V(k) \Theta_1^{-1}  \tau_1(e_1)^{-1}   && \text{for every }  k \in E_3 \cup E_4,
\end{align*}
as well as the natural $(\tau_2, \tau_1)$-covariance relation for $k \in E_2 \cup E_5$. As a consequence, $V$ extends to a \emph{continuous} map on $\R^d$ which satisfies \ref{item:U-tau}, \ref{item:U-TRS}, and obviously \ref{item:U-inter}.  
This concludes the proof for $d =2$. 

\noindent  \textsl{Step III.3: $d=3$.} The restrictions of the families of projectors to the $2$-dimensional faces 
$F_{1,0} \cap \set{k_3 \geq 0}$, $F_{1,+} \cap \set{k_3 \geq 0}$, $F_{2,\pm}$ and $F_{3,\pm}$ (in the notation of Section \ref{sec:3d}) yield $2$-dimensional families of projectors, which by assumption have the same $2$-dimensional $\Z_2$ invariants. 

Let $F_* \simeq \Bred^{(2)}$ be any of the above faces. By the previous Step of the proof, we can construct a continuous map $  U_* \colon F_* \to \U(\C^m)$ which  induces, via \eqref{W characterization}, a partial isometry 
$V|_{F_*}(k)$ relating, on the boundary of $F_*$, two symmetric Bloch frames for the two families of projectors,  as in \eqref{V intertwines}. 

By imposing properties \ref{item:U-tau} and \ref{item:U-TRS},  one constructs  a \emph{continuous} extension of $V$ to $\partial \Bred^{(3)}$ (compare Section \ref{sec:3d}), which can be re-expressed, via \eqref{W characterization}, in terms of a  continuous gauge transform $\widetilde U \colon \partial \Bred^{(3)} \to \U(\C^m)$. 
Since the group $\pi_2(\U(\C^m))$ is trivial, we can find a continuous extension $U \colon \Bred^{(3)} \to \U(\C^m)$ of $\widetilde{U}$.  Then, by using again \eqref{W characterization}, 
one obtains a partial-isometry-valued \emph{continuous} map $V$  on $\Bred^{(3)}$, which satisfies the natural  $\tau$-covariance and time-reversal symmetry properties on $\partial \Bred^{(3)}$. Therefore, $V$ extends to a \emph{continuous} map on $\R^d$ which satisfies \ref{item:U-tau}, \ref{item:U-TRS}, and obviously \ref{item:U-inter}. 
This concludes the proof for $d=3$. 
\end{proof}

\newpage

%%%%% APPENDICES

\appendix

%%%%% APPENDIX A

\newcommand{\Mid}{\wideparen{M}}

\section{Smoothing procedure} \label{app:Smoothing}

Throughout the main body of the paper, we have mainly considered the issue of the existence of a \emph{continuous} global symmetric Bloch frame for a family of projectors $\set{P(k)}_{k \in \R^d}$ satisfying Assumption \ref{proj}. This Appendix is devoted to showing that, if such a continuous Bloch frame exists, then also a \emph{smooth} one can be found, arbitrarily close to it. The topology in which we measure ``closeness'' of frames is given by the distance
\[ \mathrm{dist} (\Phi, \Psi) := \sup_{k \in \R^d} \left( \sum_{a=1}^{m} \left\| \phi_a(k) - \psi_a(k) \right\|_{\Hi}^2 \right)^{1/2} \]
for two global Bloch frames $\Phi = \set{\phi_a(k)}_{a=1, \ldots, m, \: k \in \R^d}$ and $\Psi = \set{\psi_a(k)}_{a=1, \ldots, m, \: k \in \R^d}$. The following result holds in any dimension $d \ge 0$.

\begin{prop} \label{smoothing}
Let $\set{P(k)}_{k \in \R^d}$ be a family of orthogonal projectors satisfying Assumption \ref{proj}. Assume that a \emph{continuous} global symmetric Bloch frame $\Phi$ exists for $\set{P(k)}_{k \in \R^d}$. Then for any $\mu > 0$ there exists a \emph{smooth} global symmetric Bloch frame $\Phi\sub{sm}$ such that
\[ \mathrm{dist} (\Phi, \Phi\sub{sm}) < \mu. \]
\end{prop}
\begin{proof}
We use the same strategy used to prove Proposition 5.1 in \cite{FiPaPi}, to which we refer for most details.

A global Bloch frame gives a section of the principal $\U(\C^m)$-bundle $\Fr(\PB) \to \T^d$, the \emph{frame bundle} associated with the Bloch bundle $\PB \to \T^d$, whose fibre $\Fr(\PB)_k =: F_k$ over the point $k \in \T^d$ consists of all orthonormal bases in $\Ran P(k)$. The existence of a continuous global Bloch frame $\Phi$ is thus equivalent to the existence of a continous global section of the frame bundle $\Fr(\PB)$. By Steenrod's ironing principle \cite{Steenrod,Wockel}, there exists a \emph{smooth} section of the frame bundle arbitrarily close to the latter continuous one; going back to the language of Bloch frames, this implies for any $\mu > 0$ the existence of a \emph{smooth} global Bloch frame $\widetilde{\Phi}\sub{sm}$ such that
\[ \mathrm{dist} (\Phi, \widetilde{\Phi}\sub{sm}) < \frac{1}{2} \mu. \]

The Bloch frame $\widetilde{\Phi}\sub{sm}$ will not in general be also symmetric, so we apply a symmetrization procedure to $\widetilde{\Phi}\sub{sm}$ to obtain the desired smooth symmetric Bloch frame $\Phi\sub{sm}$. In order to do so, we define the \emph{midpoint} between two sufficiently close frames $\Phi$ and $\Psi$ in $F_k$ as follows. Write $\Psi = \Phi \act U_{\Phi, \Psi}$, with $U_{\Phi, \Psi}$ a unitary matrix. If $\Psi$ is sufficiently close to $\Phi$, then $U_{\Phi, \Psi}$ is sufficiently close to the identity matrix $\Id_m$ in the standard Riemannian metric of $\U(\C^m)$ (see \cite[Section 5]{FiPaPi} for details). As such, we have that $U_{\Phi, \Psi} = \exp \left( A_{\Phi, \Psi} \right)$, for $A_{\Phi, \Psi} \in \mathfrak{u}(\C^m)$ a skew-Hermitian matrix, and $\exp \colon \mathfrak{u}(\C^m) \to \U(\C^m)$ the exponential map%
\footnote{It is known that the exponential map $\exp \colon \mathfrak{u}(\C^m) \to \U(\C^m)$ defines a diffeomorphism between a ball $B_{\delta}(0) \subset \mathfrak{u}(\C^m)$ and a Riemannian ball $B_{\delta}(\Id_m) \subset \U(\C^m)$, for some $\delta > 0$. When we say that the two frames $\Phi$ and $\Psi$ should be ``sufficiently close'', we mean (here and in the following) that the unitary matrix $U_{\Phi, \Psi}$ lies in the ball $B_{\delta}(\Id_m)$.}%
. We define then the \emph{midpoint} between $\Phi$ and $\Psi$ to be the frame
\[ \Mid(\Phi, \Psi) := \Phi \act \exp \left( \frac{1}{2} A_{\Phi, \Psi} \right). \]

Set
\begin{equation} \label{Phism}
\Phi\sub{sm}(k) := \Mid\left( \widetilde{\Phi}\sub{sm}(k), \Theta \widetilde{\Phi}\sub{sm}(-k) \act \eps \right).
\end{equation}
The proof that $\Phi\sub{sm}$ defines a \emph{smooth} global Bloch frame, which satisfies the $\tau$-equivariance property \eqref{Tau-Cov}, and that moreover its distance from the frame $\Phi$ is less than $\mu$ goes exactly as the proof of \cite[Prop. 5.1]{FiPaPi}. A slightly different argument is needed, instead, to prove that $\Phi\sub{sm}$ is also time-reversal symmetric, \ie that it satisfies also \eqref{TR}. To show this, we need a preliminary result.

\begin{lemma} \label{FermionicMidpoint}
Let $\Phi, \Psi \in F_k$ be two frames, and assume that $\Phi$ and $\Psi \act \eps$ are sufficiently close. Then
\begin{equation} \label{MidpointEps}
\Mid(\Phi, \Psi \act \eps) = \Mid(\Phi \act \eps^{-1}, \Psi) \act \eps.
\end{equation}
\end{lemma}
\begin{proof}
Notice first that
\[ \Psi = \Phi \act U_{\Phi, \Psi} \quad \Longrightarrow \quad \Psi \act \eps = \Phi \act \left( U_{\Phi, \Psi} \eps \right) \text{ and } \Psi = (\Phi \act \eps^{-1}) \act \left( \eps U_{\Phi, \Psi} \right), \]
which means that 
\[ U_{\Phi, \Psi \act \eps} = U_{\Phi, \Psi} \eps \quad \text{and} \quad U_{\Phi \act \eps^{-1}, \Psi} = \eps U_{\Phi, \Psi}. \]

Write
\[ \eps U_{\Phi, \Psi} = \exp(A) \quad \text{and} \quad U_{\Phi, \Psi} \eps = \exp(B) \]
for $A,B \in B_{\delta}(0) \subset \mathfrak{u}(\C^m)$, where $\delta > 0$ is such that $\exp \colon B_{\delta}(0) \subset \mathfrak{u}(\C^m) \to B_{\delta}(\Id_m) \subset \U(\C^m)$ is a diffeomorphism. Since the Hilbert-Schmidt norm $\left\| A \right\|\sub{HS}^2 := \tr(A^* A)$ is invariant under unitary conjugation, also $\eps^{-1} A \eps \in B_{\delta}(0)$. Then we have that
\begin{equation} \label{exp} 
\exp\left(\eps^{-1} A \eps\right) = \eps^{-1} \exp(A) \eps = \eps^{-1} \eps U \eps = U \eps = \exp(B).
\end{equation}
It then follows from \eqref{exp} that $B = \eps^{-1} A \eps$, because the exponential map is a diffeomorphism. From this we can conclude that
\begin{align*}
\Mid(\Phi \act \eps^{-1}, \Psi) \act \eps & = (\Phi \act \eps^{-1}) \act \left( \exp \left(\frac{1}{2} A \right) \eps \right) = \\
& = \Phi \act \exp \left(\frac{1}{2} \eps^{-1} A \eps \right) = \Phi \act \exp \left(\frac{1}{2} B \right) = \\
& = \Mid(\Phi, \Psi \act \eps)
\end{align*}
which is exactly \eqref{MidpointEps}.
\end{proof}

We are now in position to prove that $\Phi\sub{sm}$, defined in \eqref{Phism}, satisfies \eqref{TR}. Indeed, we compute
\begin{align*}
\Theta \Phi\sub{sm}(k) \act \eps & = \Theta \Mid \left( \widetilde{\Phi}\sub{sm}(k), \Theta \widetilde{\Phi}\sub{sm}(-k) \act \eps \right) \act \eps = && \\
& = \Mid \left( \Theta \widetilde{\Phi}\sub{sm}(k), \Theta^2 \widetilde{\Phi}\sub{sm}(-k) \act \overline{\eps} \right) \act \eps = && \text{(by \cite[Eqn. (5.11)]{FiPaPi})} \\
& = \Mid \left( \Theta \widetilde{\Phi}\sub{sm}(k), \widetilde{\Phi}\sub{sm}(-k) \act \eps^{-1} \right) \act \eps = && \text{(because $- \overline{\eps} = \eps^{-1}$ by Remark \ref{rmk:eps})} \\
& = \Mid \left( \widetilde{\Phi}\sub{sm}(-k) \act \eps^{-1}, \Theta \widetilde{\Phi}\sub{sm}(k) \right) \act \eps = && \text{(because $\Mid(\Phi,\Psi) = \Mid(\Psi,\Phi)$)} \\
& = \Mid \left( \widetilde{\Phi}\sub{sm}(-k), \Theta \widetilde{\Phi}\sub{sm}(k) \act \eps \right) = \Phi\sub{sm}(-k) && \text{(by \eqref{MidpointEps}).} 
\end{align*}
This concludes the proof of the Proposition.
\end{proof}

%%%%% APPENDIX B

\soloarXiv{
\newpage

\section{Time-reversal symmetric isomorphisms of Bloch bundles and intertwining partial isometries} 
\label{app:Isomorphisms}

\textsl{Note:} This Appendix is not meant for publication in CMP, but it will appear in the final \textsl{arXiv} version of the paper. 

\bigskip

In this Appendix, we consider two families  of projectors $\mathcal{P}_1 = \set{P_1(k)}_{k \in \R^d}$ 
and $\mathcal{P}_2 = \set{P_2(k)}_{k \in \R^d}$, both satisfying Assumption \ref{proj}, and the associated 
Bloch bundles  $\mathcal{E}_1$  and $\mathcal{E}_2$, with $\mathcal{E}_j = \big( E_j \xrightarrow{\pi_j} \T^d \big)$.  We discuss the relation between the existence of a TRS isomorphism 
of the Bloch bundles, in the sense of Remark \ref{rmk: TRS Bloch bundles}, and the existence of a partial-isometry-valued map $V : \R^d \to \mathcal{I}(\Hi)$ satisfying suitable properties. For the sake of a lighter notation, we assume $\tau_1 \equiv \tau \equiv \tau_2$. 

\begin{lemma} \label{Lem:V to isomorphism}
Any map $V \colon \R^d \to \mathcal{I}(\Hi)$ satisfying \ref{item:U-smooth}, \ref{item:U-tau}, \ref{item:U-TRS} and \ref{item:U-inter} induces a TRS isomorphism $\widehat V \colon \mathcal{E}_1 \to \mathcal{E}_2$. 
\end{lemma}

\newcommand{\class}[1]{\left[ #1 \right]_{\tau}}

\begin{proof}  Recall that $\mathcal{E}_j$, for $j \in \set{1,2}$, is defined in terms of the equivalence relation \eqref{tau equivalence}, so that
$$
E_j := \set{\class{k, \phi} \in (\R^d \times \Hi) / _{\sim_\tau}  : \phi \in \Ran P_j(k)}. 
$$

\medskip

For a point $p = \class{k, \phi} \in E_1$, we define $\widehat V$ by 
\begin{equation} \label{hatV definition}
\widehat V  \class{k, \phi}  = \class{k, V(k) \phi}.
\end{equation}

\noindent The independence of definition \eqref{hatV definition} from the choice of the representative follows from property \ref{item:U-tau}. Indeed, for any $p = \class{k, \phi} \in E_1$ and $\lambda \in \Lambda$, one has
\begin{eqnarray*}
\widehat V  \class{k + \lambda, \tau(\lambda) \phi} &=&  \class{k + \lambda, V(k + \lambda)  \tau(\lambda) \phi} 
=   \class{k + \lambda, \tau(\lambda) V(k) \phi} \\
&=& \class{k, V(k) \phi}  = \widehat V  \class{k, \phi}. \\ 
\end{eqnarray*}

\noindent The fact that $\widehat V$ maps $\mathcal{E}_1$ onto $\mathcal{E}_2$ preserving the fibers follows from \ref{item:U-inter}, while the fact that $\widehat V$ is unitary on the fibers is a consequence of the fact that $V(k)$ is a partial isometry. 
The smoothness of  $\widehat V$ follows from \ref{item:U-smooth}.

\noindent Finally, $\widehat V$ is a TRS isomorphism. Indeed, for any $p = \class{k, \phi} \in E_1$ one has
\begin{eqnarray*}
\widehat V  \, \widehat \Theta_1  \class{k, \phi} &=&  \widehat V \class{- k , \Theta_1 \phi} 
=  \class{- k, V(- k) \Theta_1  \phi} \\
&=& \class{- k, \Theta_2 V(k) \phi}  = \widehat \Theta_2 \, \widehat V  \class{k, \phi},  \\ 
\end{eqnarray*}
which proves that $\widehat V  \, \widehat \Theta_1  = \widehat \Theta_2 \, \widehat V $.
\end{proof}

\goodbreak

The converse statement is also true.

\begin{lemma} \label{Lem:Isomorphism to V}
Let $\widehat V \colon \mathcal{E}_1 \to \mathcal{E}_2$ be a TRS isomorphism of Bloch bundles. Then, there exists a partial-isometry-valued map  $V \colon \R^d \to \mathcal{I}(\Hi)$ satisfying \ref{item:U-smooth}, \ref{item:U-tau}, \ref{item:U-TRS} and \ref{item:U-inter}, 
such that 
\begin{equation} \label{V crucial}
\widehat V  \class{k, \phi}  = \class{k, V(k) \phi}.
\end{equation}
\end{lemma}

\begin{proof}  Recall that, by definition, each Bloch bundle has a natural smooth embedding 
\begin{equation} \label{Embedding}
\mathcal{E}_j   \hookrightarrow  (\R^d \times \Hi) / _{\sim_\tau} \simeq  \T^d \times \Hi.  
\end{equation}
We extend $\widehat V $ to a smooth morphism of Hilbert bundles 
$$
M \colon (\R^d \times \Hi) / _{\sim_\tau}  \longrightarrow  (\R^d \times \Hi) / _{\sim_\tau}
$$ 
by setting $M$ equal to zero on the orthogonal complement of $E_1$, and by embedding $E_2$
in $(\R^d \times \Hi) / _{\sim_\tau}$ via \eqref{Embedding}.  
Since $M$ maps smooth $\tau$-equivariant functions into themselves, 
it is identified with an element $V$ of the $*$-algebra 
$
C^{\infty}_{\tau}(\R^d, \mathcal{B}(\Hi)),   %\simeq \mathrm{End} \left(   (\R^d \times \Hi) / _{\sim_\tau}  \right)
$ 
consisting of smooth $\tau$-covariant  $\mathcal{B}(\Hi)$-valued functions. 
Moreover, for any $p = \class{k, \phi}\in E_1$, we have
$$  
\widehat V  \class{k, \phi}  = M \class{k,\phi}  =  \class{k, V(k) \phi}, 
$$
which proves that relation \eqref{V crucial} holds true. Since $\widehat V$ is an isomorfism of Hermitian bundles, $V(k)|_{\Ran P_1(k)}$ is a unitary operator from  $\Ran P_1(k)$ to 
$\Ran P_2(k)$,  so that $V(k)$ is a partial isometry in $\Hi$. Thus, we obtain a $\mathcal{I}(\Hi)$-valued map $V$ satisfying 
\ref{item:U-smooth}, \ref{item:U-tau}, and \ref{item:U-inter}.

%Let $\Pi_{2}(q):  E_2 \to \Hi$  be defined by the following procedure: for any point $\class{k', \phi'} \in E_2$ one selects the unique representative such that $k'$ and $q$ are in the same periodicity cell;  denote the latter by $\class{k'', \phi''}$; then  set  $\Pi_2(q) \class{k', \phi'} = \phi''$. 
%
%Equipped with this definition, we set for every $\phi \in \Hi$
%\begin{equation} \label{V definition}
%V(k) \phi := 
%\begin{cases}  
%   0                                                                          &  \qquad \mbox{if } \phi \notin \Ran P_1(k)  \\
% \Pi_2(k) \big( \widehat V  \class{k, \phi}  \big)   &  \qquad  \mbox{if } \phi \in \Ran P_1(k).
%\end{cases} 
%\end{equation}
%Notice that, as a consequence of the previous definition, the relation \eqref{V crucial} holds true. 
%
%Property \ref{item:U-smooth} follows from the smoothness of $V$ and the definition of the projection $\Pi_2(q)$. 
%
%Property \ref{item:U-tau} is easily checked. Indeed, by the independence of the representative of the definition of $\widehat V$, 
%one has 
%$$
%\widehat V \class{k, \phi} = \widehat V \class{k + \lambda, \tau(\lambda) \phi},  \qquad \mbox{ for all } k \in \R^d, \phi \in \Ran P_1(k), 
%$$
%hence 
%$$
%\class{k, V(k) \phi} = \class{k + \lambda, V(k + \lambda) \tau(\lambda) \phi} = \class{k, \tau(\lambda)^{-1} V(k + \lambda) \tau(\lambda) \phi}
%$$
%which immediately implies \ref{item:U-tau}. 

Finally, property \ref{item:U-TRS} is easily checked. 
Indeed, since $\widehat V$ preserves the time-reversal structure of the Bloch bundles, one has 
$$
\widehat V \widehat \Theta_1 \class{k, \phi} = \widehat \Theta_2 \widehat V \class{k, \phi},  \qquad \mbox{ for all } k \in \R^d, 
\phi \in \Ran P_1(k). 
$$
Since, in view of  \eqref{V crucial},
\begin{eqnarray*}
\widehat V \widehat \Theta_1 \class{k, \phi} &=& \widehat V  \class{- k, \Theta_1 \phi} =  \class{- k, V(- k) \Theta_1 \phi}, \\ 
\widehat \Theta_2 \widehat V \class{k, \phi} &=&  \widehat \Theta_2  \class{k, V(k) \phi} = \class{- k, 	\Theta_2 V(k) \phi}, 
\end{eqnarray*}
one concludes that \ref{item:U-TRS}  holds true. 
\end{proof}

} %%%End soloarXiv

%%%%% BIBLIOGRAPHY

%%%%%%% END MATTER

\bigskip \bigskip

{\footnotesize

\begin{tabular}{ll}
(D. Fiorenza) & \textsc{Dipartimento di Matematica, ``La Sapienza'' Universit\`{a} di Roma} \\
 &  Piazzale Aldo Moro 2, 00185 Rome, Italy \\
 &  {E-mail address}: \href{mailto:fiorenza@mat.uniroma1.it}{\texttt{fiorenza@mat.uniroma1.it}} \\
 \\
(D. Monaco) &  \textsc{SISSA -- International School for Advanced Studies}\\
& Via Bonomea 265, 34136 Trieste, Italy \\
& {E-mail address}: \href{mailto:dmonaco@sissa.it}{\texttt{dmonaco@sissa.it}} \\
\\
(G. Panati) & \textsc{Dipartimento di Matematica, ``La Sapienza'' Universit\`{a} di Roma} \\
 &  Piazzale Aldo Moro 2, 00185 Rome, Italy \\
 &  {E-mail addresses}: \href{mailto:panati@mat.uniroma1.it}{\texttt{panati@mat.uniroma1.it}}, 
 \href{mailto:panati@sissa.it}{\texttt{panati@sissa.it}}  \\

\end{tabular}

}


\newpage

\begin{thebibliography}{OOOOo}

\bibitem[AZ]{Altland Zirnbauer}
\textsc{Altland, A.; Zirnbauer, M.} : Non-standard symmetry classes in mesoscopic normal-superconducting hybrid structures, {\it Phys. Rev. B} {\bf 55} (1997), 1142--1161.

\bibitem[An]{Ando} 
\textsc{Ando, Y.} : Topological insulator materials, {\it J. Phys. Soc. Jpn.} {\bf 82} (2013), 102001.

\bibitem[ASV]{Schulz-Baldes12}
\textsc{Avila, J.C.; Schulz-Baldes, H.; Villegas-Blas, C.} : Topological invariants of edge states for periodic two-dimensional models, {\it Mathematical Physics, Analysis and Geometry} {\bf 16} (2013), 136--170.

\bibitem[CDFG]{Lyon}
\textsc{Carpentier, D.; Delplace, P.; Fruchart, M.; Gawedzki, K.} : Topological index for periodically driven time-reversal invariant 2D systems, {\it Phys. Rev. Lett.} {\bf 114} (2015), 106806.

\bibitem[CDFGT]{Lyon15}
\textsc{Carpentier, D.; Delplace, P.; Fruchart, M.; Gawedzki, K.; Tauber, C.} : Construction and properties of a topological index for periodically driven time-reversal invariant 2D crystals, {\it Nucl. Phys. B} {\bf 896} (2015), 779--834.%, preprint available at \href{http://arxiv.org/abs/1503.04157}{\texttt{arXiv:1503.04157}}.

%\bibitem[CFMP]{Cerulli}
%\textsc{Cerulli Irelli, G. ; Fiorenza, D.; Monaco, D.; Panati, G.} : Geometry of Bloch bundles: a unifying quiver-theoretic approach, in preparation (2015).

\bibitem[Ch]{Experiment}
\textsc{Chang, C.-Z.} \textsl{et al.} : Experimental Observation of the Quantum Anomalous Hall Effect in a Magnetic Topological Insulator, {\it Science} {\bf 340} (2013), 167--170.

\bibitem[DG]{deNittisGomi}
\textsc{De Nittis, G.; Gomi, K.} : Classification of ``Quaternionic'' Bloch bundles, {\it Commun. Math. Phys.} {\bf 339} (2015), 1--55.

\bibitem[DNF]{Dubrovin}
\textsc{Dubrovin, B.A.; Novikov, S.P.; Fomenko, A.T.} : Modern Geometry -- Methods and Applications. Part II: The Geometry and Topology of Manifolds. No. 93 in Graduate Texts in Mathematics. Springer-Verlag, New York, 1985.

\bibitem[FMP]{FiPaPi}
\textsc{Fiorenza, D.; Monaco, D.; Panati, G.} : Construction of real-valued localized composite Wannier functions for insulators, {\it Ann. Henri Poincar\'{e}} {\bf 17} (2016), 63--97.% Preprint available at \href{http://arxiv.org/abs/1408.0527}{\texttt{arXiv:1408.0527}}.

\bibitem[FW]{Frohlich}
\textsc{Fr\"{o}hlich, J. ; Werner, Ph.} : Gauge theory of topological phases of matter, {\it EPL}\ {\bf 101} (2013), 47007.

\bibitem[FC]{FruchartCarpentier2013}
\textsc{Fruchart, M. ; Carpentier, D.} : An introduction to topological insulators, {\it Comptes Rendus Phys.}\ {\bf 14} (2013), 779--815.

\bibitem[FK]{FuKa}
\textsc{Fu, L.; Kane, C.L.} : Time reversal polarization and a $\Z_2$ adiabatic spin pump, {\it Phys. Rev. B} {\bf 74} (2006), 195312.

\bibitem[FKM]{FuKaneMele}
\textsc{Fu, L.; Kane, C.L.; Mele, E.J.} : Topological insulators in three dimensions, {\it Phys. Rev. Lett.} {\bf 98} (2007), 106803.

\bibitem[FKMM]{FKMM}
\textsc{Furuta, M.; Kametani, Y.; Matsue, H.; Minami, N.} : Stable-homotopy Seiberg-Witten invariants and Pin bordisms. UTMS Preprint Series 2000, UTMS 2000-46, 2000.

\bibitem[Gr]{Graf review}
\textsc{Graf, G.M.} : Aspects of the Integer Quantum Hall Effect, \emph{Proceedings of Symposia in Pure Mathematics} \textbf{76} (2007), 429--442.

\bibitem[GP]{GrafPorta}
\textsc{Graf, G.M.; Porta, M.} : Bulk-edge correspondence for two-dimensional topological insulators, {\it Commun. Math. Phys.} {\bf 324} (2013), 851--895.

\bibitem[Hal]{Haldane88}
\textsc{Haldane, F.D.M.} :  Model for a Quantum Hall Effect without Landau levels: condensed-matter realization of the ``parity anomaly'', {\it Phys. Rev. Lett.} {\bf 61} (1988), 2017.

\bibitem[HK]{HasanKane}
\textsc{Hasan, M.Z.; Kane, C.L.} :  Colloquium: Topological Insulators, {\it Rev. Mod. Phys.} {\bf 82} (2010), 3045--3067.

\bibitem[Hua]{Hua}
\textsc{Hua, L.-K.}: On the theory of automorphic functions of a matrix variable I -- Geometrical basis, {\it Am. J. Math.} {\bf 66} (1944), 470--488.

\bibitem[HR]{HuaReiner}
\textsc{Hua, L.-K.; Reiner, I.} : Automorphisms of the unimodular group, {\it Trans. Amer. Math. Soc.} {\bf 71} (1951), 331--348.

\bibitem[Hus]{Husemoller}
\textsc{Husemoller, D.}: {\it Fibre bundles}, 3rd edition. No. 20 in Graduate Texts in Mathematics. Springer-Verlag, New York, 1994.

\bibitem[KM$_1$]{KaneMele2005}
\textsc{Kane, C.L.; Mele, E.J.} : $\Z_2$ Topological Order and the Quantum Spin Hall Effect, {\it Phys. Rev. Lett.} \textbf{95} (2005), 146802.

\bibitem[KM$_2$]{KaneMele_graphene}
\textsc{Kane, C.L.; Mele, E.J.} : Quantum Spin Hall Effect in graphene, {\it Phys. Rev. Lett.} \textbf{95} (2005), 226801.

\bibitem[Ka]{Kato}
\textsc{Kato, T.} : Perturbation theory for linear operators. Springer, Berlin, 1966.

\bibitem[KG]{Kennedy}
\textsc{Kennedy, R.; Guggenheim, C.} : Homotopy theory of strong and weak topological insulators, {\it Phys. Rev. B} {\bf 91} (2015), 245148.%, preprint available at \href{http://arxiv.org/abs/1409.2529}{\texttt{arXiv:1409.2529}}.

\bibitem[KZ$_1$]{KennedyZirnbauer14}
\textsc{Kennedy, R.; Zirnbauer, M.R.} : Bott periodicity for $\Z_2$ symmetric ground states of gapped free-fermion systems, to appear in {\it Commun. Math. Phys.} (2015). Preprint available at \href{http://arxiv.org/abs/1409.2537}{\texttt{arXiv:1409.2537}}.

\bibitem[KZ$_2$]{KennedyZirnbauer14b}
\textsc{Kennedy, R.; Zirnbauer, M.R.} : Bott-Kitaev Periodic Table and the Diagonal Map, {\it Physica Scripta} T {\bf 164} (2015), 014010.%, available at \href{http://arxiv.org/abs/1412.4808}{\texttt{arXiv:1412.4808}}.

\bibitem[Ki]{Kitaev}
\textsc{Kitaev, A.} : Periodic table for topological insulators and superconductors, {\it AIP Conf. Proc.} {\bf 1134} (2009), 22.

\bibitem[Ku]{Kuiper}
\textsc{Kuiper, N.H.} : The homotopy type of the unitary group of Hilbert space, {\it Topology} {\bf 3} (1965), 19--30.

\bibitem[MM]{Mackey}
\textsc{Mackey, D.S.; Mackey, N.} : On the Determinant of Symplectic Matrices. {\it Numerical Analysis Report} {\bf 422} (2003), Manchester Centre for Computational Mathematics, Manchester, England.

\bibitem[MP]{PaMo}
\textsc{Monaco, D.; Panati, G.} : Symmetry and localization in periodic crystals: triviality of Bloch bundles with a fermionic time-reversal symmetry. Proceedings of the conference {\it  ``SPT2014 -- Symmetry and Perturbation Theory'', Cala Gonone, Italy}, {\it Acta App. Math.} {\bf 137} (2015), 185--203.

\bibitem[MB]{Moore Balents}
\textsc{Moore, J.E.; Balents, L.} : Topological invariants of time-reversal-invariant band structures, {\it Phys. Rev. B} {\bf 75} (2007), 121306(R).

\bibitem[Pa]{Panati}
\textsc{Panati, G.}: Triviality of Bloch and Bloch-Dirac bundles, {\it Ann. Henri Poincar\'e} {\bf 8} (2007),  995--1011.

\bibitem[Pr$_1$]{Prodan1} 
\textsc{Prodan, E.} : Robustness of the Spin-Chern number, {\it Phys. Rev. B} {\bf 80} (2009), 125327.

\bibitem[Pr$_2$]{Prodan2}
\textsc{Prodan, E.} : Disordered topological insulators: A non-commutative geometry perspective, {\it J. Phys. A} {\bf 44} (2011), 113001.

\bibitem[Pr$_3$]{Prodan3}
\textsc{Prodan, E.} : Manifestly gauge-independent formulations of the $\Z_2$ invariants, {\it Phys. Rev. B} {\bf 83} (2011),  235115.

\bibitem[RSFL]{RyuSchnyder2010}
\textsc{Ryu, S.;  Schnyder, A.P.;  Furusaki, A.; Ludwig, A.W.W.} :  
Topological insulators and superconductors: Tenfold way and dimensional hierarchy, {\it New J. Phys.}\ {\bf 12} (2010), 065010.

\bibitem[Sch$_1$]{Schulz-BaldesCMP}
\textsc{Schulz-Baldes, H.} : Persistence of spin edge currents in disordered Quantum Spin Hall systems, {\it Commun. Math. Phys.} {\bf 324} (2013), 589--600.

\bibitem[Sch$_2$]{Schulz-Baldes13} 
\textsc{Schulz-Baldes, H.} : $\Z_2$ indices and factorization properties of odd symmetric Fredholm operators, {\it Documenta Mathematica} {\it 20} (2015), 1481--1500.%, preprint available at \href{http://arxiv.org/abs/1311.0379}{\texttt{arXiv:1311.0379}}.

\bibitem[St]{Steenrod}
\textsc{Steenrod, N.} : {\it The Topology of Fibre Bundles}. No. 14 in Princeton Mathematical Series. Princeton University Press, Princeton, 1951.  

\bibitem[SPFKS]{Fuchs}
\textsc{Sticlet, D. ; P\'{e}chon, F.; Fuchs, J.-N.; Kalugin, P.; Simon P.} : Geometrical engineering of a two-band Chern insulator in two dimensions with arbitrary topological index, {\it Phys. Rev. B} {\bf 85} (2012), 165456.

\bibitem[SV$_1$]{Vanderbilt1}
\textsc{Soluyanov, A.A.; Vanderbilt, D.} : Wannier representation of $\Z_2$ topological insulators, {\it Phys. Rev. B} {\bf 83} (2011), 035108.

\bibitem[SV$_2$]{Vanderbilt2}
\textsc{Soluyanov, A.A.; Vanderbilt, D.} : Computing topological invariants without inversion symmetry, {\it Phys. Rev. B} {\bf 83} (2011), 235401.

\bibitem[SV$_3$]{Vanderbilt3}
\textsc{Soluyanov, A.A.; Vanderbilt, D.} : Smooth gauge for topological insulators, {\it Phys. Rev. B} {\bf 85} (2012), 115415.

\bibitem[TKNN]{TKNN}
\textsc{Thouless D.J.; Kohmoto, M.; Nightingale, M.P.; de Nijs, M.} : Quantized Hall conductance in a two-dimensional periodic potential, {\it Phys. Rev. Lett.} \textbf{49} (1982), 405--408.

\bibitem[Wo]{Wockel}
\textsc{Wockel, Ch.} : A generalization of Steenrod's Approximation Theorem,  {\it Arch. Math. (Brno)} {\bf 45} (2009),  95--104. 
\end{thebibliography}
\end{document}